\documentclass[12pt]{article}
\usepackage{amssymb,amsmath,amsfonts,bbm,eurosym,geometry,ulem,graphicx,caption,color,setspace,comment,footmisc,natbib,pdflscape,array,enumerate,nomencl,etoolbox,mathtools,booktabs,makecell,multirow,multicol,float}
\usepackage{subcaption}
\usepackage{threeparttable}
\usepackage{authblk}
\usepackage[ruled]{algorithm2e}
\usepackage[toc]{appendix}
\usepackage[thmmarks,amsmath,thref]{ntheorem}
\newtheorem{theorem}{Theorem}[section]
\newtheorem{definition}{Definition}[section]
\newtheorem{lemma}{Lemma}[section]

\newtheorem{assumption}{Assumption}[section]

\newtheorem{proposition}{Proposition}[section]
\newtheorem{corollary}{Corollary}[section]
\newtheorem{remark}{Remark}[section]
\newtheorem{open problem}{Open Problem}
\theorembodyfont{\normalfont}

\theoremsymbol{{$\square$}}
\newtheorem*{proof}{Proof}
\numberwithin{equation}{section}

\usepackage[colorlinks=true,linkcolor=blue]{hyperref}
\hypersetup{%
  citecolor=blue
}

\usepackage{titlesec}
\titleformat*{\section}{\normalfont\large\bfseries}
\titlespacing*{\section}      {0pt}{1.4ex plus .3ex minus .2ex}{0.8ex plus .2ex}
\titlespacing*{\subsection}   {0pt}{1.2ex plus .3ex minus .2ex}{0.5ex plus .2ex}
\titlespacing*{\subsubsection}{0pt}{1.0ex plus .2ex minus .2ex}{0.4ex plus .1ex}
\makeatletter
\g@addto@macro\normalsize{%
  \setlength\abovedisplayskip{4pt plus 2pt minus 1pt}%
  \setlength\belowdisplayskip{4pt plus 2pt minus 1pt}%
  \setlength\abovedisplayshortskip{2pt plus 2pt}%
  \setlength\belowdisplayshortskip{4pt plus 2pt minus 2pt}%
}
\AtBeginDocument{%
  \def\@listI{\leftmargin\leftmargini \parsep 0pt
    \topsep 3pt plus 1pt minus 1pt \itemsep 2pt plus 1pt minus 1pt}%
  \let\@listi\@listI \@listi}
\@ifundefined{theorempreskipamount}{}{%
  \setlength{\theorempreskipamount}{1.2ex plus .3ex minus .2ex}%
  \setlength{\theorempostskipamount}{1.2ex plus .3ex minus .2ex}}
\makeatother
\allowdisplaybreaks[4]
\makeatletter
\AtBeginDocument{%
  \everydisplay\expandafter{\the\everydisplay \let\@nomath\@gobble \setstretch{1}\selectfont}}
\makeatother
\begin{document}

%% ---- title ---- %%
\title{Recursive contracts in non-convex environments\thanks{We thank Matias Bayas-Erazo, John Geanakoplos, Nicola Pavoni, Tom Sargent, Paolo Siconolfi, Jan Werner, Yucheng Yang and Xiaodong Zhu, as well as seminar participants at Warwick, Yale, HKU and UZH, and participants  the 2026 SED and SAET conferences for their comments. Shen and Zhou acknowledge support from the National Key R\&D Program of China (project number 2021YFA1001200). K\"{u}bler thanks the SNF for financial support.}}
%% ---- authors & affiliations ---- %%
% \author{Chengfeng Shen}
% \affil{Peking University}
\author{Chengfeng Shen\footnote{Peking University. Email: \href{mailto:shencf1999@pku.edu.cn}{shencf1999@pku.edu.cn}} \and Felix K\"{u}bler\footnote{University of Zurich. Email: \href{mailto:fkubler@gmail.com}{fkubler@gmail.com}} \and Zhennan Zhou\footnote{Westlake University. Email: \href{mailto:zhouzhennan@westlake.edu.cn}{zhouzhennan@westlake.edu.cn}}}

		\newcommand\tbbint{{-\mkern -16mu\int}}
	\newcommand\tbint{{\mathchar '26\mkern -14mu\int}}
	\newcommand\dbbint{{-\mkern -19mu\int}}
	\newcommand\dbint{{\mathchar '26\mkern -18mu\int}}
	\newcommand\bint{
		{\mathchoice{\dbint}{\tbint}{\tbint}{\tbint}}
	}
	\newcommand\bbint{
		{\mathchoice{\dbbint}{\tbbint}{\tbbint}{\tbbint}}
	}
	%% ---- body text ---- %%
	{\setstretch{1}\maketitle}
\begin{abstract}
Many dynamic contracting problems in economics are non-convex, and deterministic optima can be dominated by randomized solutions. We develop a recursive method to compute optimal lotteries in the presence of forward-looking constraints, extending the approach of \cite{marcet2019recursive} to non-convex settings. We show that an inf-sup functional equation recovers the correct value when lotteries are conducted after agents commit to actions (``ex-ante'' lotteries), though the resulting solution often lacks economic meaning. Replacing the inf-sup with a sup-inf operator delivers the correct value of an economically more relevant ``ex-post'' lottery problem, where agents act after observing the lottery's outcome. We also develop an algorithm to recover the optimal lottery policies from the solution of the functional equation.
\end{abstract}

%
%\medskip

\medskip

\section{Introduction}

Many central problems in economics are naturally cast as non-convex constrained optimization problems. Non-convexities imply that deterministic optima may be dominated by randomized solutions \citep[see, e.g.,][]{myerson1982optimal,prescott1984pareto,stiglitz1982self}.\footnote{Throughout the paper, we use the terms ``lottery solution'' and ``randomized solution'' interchangeably to refer to the optimal, possibly random, solution of a nonlinear optimization problem.}
In general, the scope for improvement through randomization depends on the economic model, in particular on which constraints can be assumed to hold only in expectation.
In an intertemporal setting with forward-looking constraints, randomization always (weakly) relaxes those constraints, even if, within a given period, all agents can make choices contingent on the outcome of the lottery. It is then a quantitative question whether lotteries improve welfare and whether this improvement is significant. To address this question, we develop a method to compute optimal lotteries in dynamic non-convex contracting problems.

As a first step, we show how to characterize these contracts recursively.
Whenever agents face intertemporal incentive constraints or react to policy in ways that the planner must treat as constraints, the ``natural'' state space is inadequate and promised future utilities must be carried as additional state variables \citep{spear1987repeated}. This expansion frequently makes the problem computationally intractable.
\cite{marcet2019recursive} provide a recursive alternative for dynamic models with forward-looking constraints, which has been widely applied \citep[e.g.,][]{cooley2004aggregate,kehoe2002international,attanasio2000consumption,aiyagari2002optimal}. However, because the approach relies on the existence of a saddle point, it fails, in general, when the problem is non-convex \citep[see][]{pavoni2018dual,bloise2022negishi}.

In this paper, we show that the use of lotteries can reconcile the recursive approach of \cite{marcet2019recursive} with non-convexities. Specifically, when the planner can condition on a public randomization device, appropriately modified versions of their functional equation deliver the optimal value of problems with forward-looking constraints, even when these are non-convex. 
We prove that the dual problem considered in \cite{pavoni2018dual} gives the correct value if one allows for arbitrary lotteries over current actions and requires all forward-looking constraints to hold only in expectation. We refer to these as ``ex-ante lotteries'': the agent commits to a lottery-contingent plan of actions before the lottery is drawn. The plan is optimal ex ante, but the realized action is often not optimal given the outcome of the lottery. Following \cite{pavoni2018dual}, we show that the dual problem can be formulated as a recursive functional equation (FE), similar to the one proposed by \cite{marcet2019recursive}, even if the problem is not convex. We refer to this as the inf-sup FE.

Unfortunately, ex-ante lotteries often have no economic interpretation. In particular, we show that for Ramsey-style problems of optimal government policy \citep{aiyagari2002optimal}, the ex-ante lottery solution cannot be interpreted economically because the agent's budget constraint is violated.
Economic meaning is restored if agents choose optimally after observing the outcome of the current-period lottery, while taking into account that next period's consumption is random even conditional on the exogenous shock. We refer to such randomization as ex-post lotteries; agents take actions after observing the realization of the lottery.
We show that in this case the dual formulation in \cite{pavoni2018dual} does not give the correct value of the problem, but that replacing the inf-sup by a sup-inf in the functional equation yields the correct value of the ex-post lottery problem. The resulting sup-inf FE can be analyzed with the same methods as the inf-sup FE.

As \cite{cole2012recursive} and \cite{marimon2021envelope} point out, it is difficult to recover the optimal (non-lottery) policy from the solution of the functional equation if the problem is not \textit{strictly convex}. The introduction of lotteries transforms non-convex problems into convex ones, but not into strictly convex ones.
We propose a computational method to recover the optimal lottery policy from the solution of the functional equation. It builds on the insight of \cite{marimon2021envelope}, who relate the subdifferential of the value function to the set of possible utility promises. 
Although our algorithm keeps track of utility promises along a simulated path, it does not need to compute the set of feasible promises. Instead, we use a subgradient descent method to compute the optimal lottery at any time period $t$ for given utility promises, using a finite-difference approximation of the derivative of next period's value function. In Appendix~\ref{App:Num_ex}, we illustrate the method with a simple example. In a companion paper\citep{Shen-Kubler_Ramsey}, we apply the algorithm to analyze optimal lotteries in the \cite{aiyagari2002optimal} model of optimal fiscal policy.

The remainder of the paper is organized as follows. After reviewing the existing literature, we illustrate the main ideas with two simple examples in Section~\ref{sec:eclott}. Section~\ref{sec:formulation} gives a general formulation of both the deterministic problem and the two lottery problems.
Section~\ref{sec:duality} links the value of the dual problem to lotteries. In Section~\ref{sec:lotteries}, we formulate the problem recursively. Section~\ref{sec:policies} discusses how to recover optimal policies. The proofs are collected in Appendix~\ref{app:sec2}, with some lemmas in Appendix~\ref{app:math}. The online Appendices~\ref{lagrange}--\ref{App:Num_ex} contain further results.

\subsubsection*{Related Literature}
Following \cite{spear1987repeated}, the
standard method for recursively solving dynamic contracting problems is to treat
the promised utility of the agent(s) as a state variable and to use this
state variable to capture the implicit prior promises. 
\cite{marcet2019recursive} (first circulated as a working paper in 1994) were the first to use Lagrange multipliers as state variables. \cite{golosov2016recursive} provide a more general abstract mathematical foundation for the recursive Lagrangian approach to infinite-horizon convex problems and apply it in several settings.
\cite{pavoni2018dual} show how the dual problem can be characterized by a functional equation even in the absence of convexity. They were also the first to prove the existence of Lagrange multipliers (in an appropriate space) for the problem.

\cite{cole2012recursive} and \cite{marimon2021envelope} discuss how optimal policies can be recovered from a solution to the functional equation. As we explain in detail below, our method for recovering the optimal policy is based on their work.

\cite{bloise2022negishi} propose an alternative approach based on promised values. 
Their approach independently connects sunspots (ex-post lottery)  to the sup‑inf FE. Our paper instead starts from the Lagrangian dual–lottery equivalence, interpreting inf‑sup and sup‑inf as two different lottery solutions in a general framework.
Although it is not clear whether their method can be made computationally practical, they show that in many scenarios it requires fewer assumptions about randomization than the dual approach of \cite{pavoni2018dual}. We discuss their approach in Appendix~\ref{Bloise}.

There is a large literature on the use of lotteries in static non-convex problems \citep[e.g.,][]{doepke2006dynamic, pavoni2007optimal,gauthier2014value}. 
\cite{shen2025lagrangian} were the first in economics to exploit the equivalence between the value of the dual (Lagrangian) problem and that of the lottery problem and to develop an efficient computational method to approximate optimal lotteries.

\section{Two examples}
\label{sec:eclott}
To illustrate and motivate our results, we provide two examples. The first example shows how the sup-inf and inf-sup functional equations correspond to two distinct lottery solutions of a simple non-convex contracting problem.
The second example shows that, with forward-looking constraints, ex-post lotteries have a natural economic interpretation, whereas ex-ante lotteries typically do not. 
\subsection{Lottery improvements in a two-period model}
Consider a two-period contracting problem in which a risk-neutral principal minimizes the total compensation paid to a risk-loving agent,
$$
 V=\max _{(c_1,c_2)\in [0,1]^ 2 } -c_1-c_2, $$
 subject to a participation constraint
$$ (PC) \ \ c_1^2 + \beta c_2^2 \ge 0.8, \ \beta>0 .$$

It is easy to see that the optimal deterministic solution is given by
$$
(c_1,c_2)=\left\{ \begin{array}{ll}
(\sqrt{0.8},0) & \beta \le 1\\
(0,\sqrt{\frac{0.8}{\beta}}) & \beta \ge 1
\end{array} \right. ,
    $$
and that it can be obtained by solving
\begin{equation}
    \label{sex:det} V^D=\max_{c_1\in[0,1]} \max_{c_2\in [0,1]} \min_{\lambda \ge 0} -c_1-c_2+\lambda(c_1^2+\beta c_2^2-0.8)=\max\{-\sqrt{0.8},-\sqrt{\frac{0.8}{\beta}}\} .\end{equation}
Since this is a non-convex problem, there is a  duality gap, and the problem cannot be written recursively with multipliers as state variables. However, one of the main contributions of this paper is to show that interchanging the second max operator and the min operator in \eqref{sex:det} yields the optimal value of a lottery contracting problem.
In fact,
the following naturally recursive formulation,
\begin{eqnarray*} V^*&=&\max_{c_1\in[0,1]}  \min_{\lambda \ge 0 } \max_{c_2\in [0,1]}-c_1-c_2+\lambda(c_1^2+\beta c_2^2-0.8)\\
 &=& \max_{c_1\in[0,1]}  \min_{\lambda \ge 0 } -c_1+\lambda(c_1^2-0.8)+V_2(\lambda),
\end{eqnarray*}
with $ V_2(\lambda)=\max_{c_2\in[0,1]} -c_2+\beta \lambda c_2^2=\max\{0,\beta\lambda-1\} $,
gives the value of a lottery problem in which the principal can randomize over the compensation in the second period. It is easy to verify that
$
V^*=\max\{-\sqrt{0.8},-\frac{0.8}{\beta}\}.
$
For $ \beta > \sqrt{0.8} $, the optimal (ex-post) lottery solution is given by
$$
c_1=0; \quad c_2=\begin{cases} 0 & \text{with probability } 1-0.8/\beta,\\ 1 & \text{with probability } 0.8/\beta. \end{cases}
$$
Moreover, interchanging the first maximum operator and the minimum operator results in the dual problem, which is equivalent to a problem in which the principal can offer randomized compensation in both periods. Consider the alternative  recursive formulation,
\begin{eqnarray*} V^{**}&=&  \min_{\lambda \ge 0}\max_{c_1\in[0,1]} \max_{c_2\in [0,1]}-c_1-c_2+\lambda(c_1^2+\beta c_2^2-0.8)\\
 &=& \min_{\lambda \ge 0}\max_{c_1\in[0,1]}   -c_1+\lambda(c_1^2-0.8)+V_2(\lambda).
\end{eqnarray*}
It is easy to verify that
$
V^{**}=\max\{-0.8,-\frac{0.8}{\beta}\}.
$
For $\beta<1$, the optimal (ex-ante) lottery solution is given by
$
c_1=0 \text{ with probability } 0.2,\ c_1=1 \text{ with probability } 0.8;\ c_2=0.
$
Therefore, for $ \beta \ge 1 $ we have $ V^{**}=V^{*}>V^D $;
for $ \sqrt{0.8}<\beta < 1 $ we have $ V^{**}>V^*>V^D$; for $0<\beta \le \sqrt{0.8}$, we have $V^{**}>V^{*}=V^D$.
In summary, the use of randomized compensation can increase the value of the principal and leads to simple recursive formulas for this optimal value.

\subsection{Optimal randomized fiscal policy}

In models of optimal government policy, the forward-looking constraints often consist of agents' first-order conditions and/or intertemporal budget constraints \citep[see, e.g.,][for the incomplete-markets problem]{aiyagari2002optimal}.
To illustrate the difference between ex-post and ex-ante lotteries, it is useful to consider the simplest possible version  of the problem from \cite{aiyagari2002optimal}.
A single agent takes the risk-free rates $ (R_t) $, income taxes $ (\tau_t) $ and transfers $(T_t)$ as given, and solves
  \begin{equation}  \label{eq:agent}
    \begin{aligned}    &\max_{(c_t,\ell_t,b_t)_{t=0}^\infty} \mathbb{E} \sum_{t=0}^{\infty} \beta^t \left(u(c_t)+v(\ell_t)\right) \\\mbox{ s.t. } & c_t = \ell_t(1-\tau_t) +b_{t-1}  - \frac{b_t}{R_t}+T_t, \quad
     \sup_t | b_t | < \infty, \ \ell_t \in [0,1] \ \forall t.
\end{aligned}
\end{equation}
Under full commitment, the government faces exogenous stochastic expenditure $g_t$ and chooses labor income taxes $ \tau_t $, debt $ b_t $, transfers $ T_t \ge 0 $ and interest rates on its debt $ R_t $ such that, given the agent's optimal choices, the government budget constraint
$
b_{t-1} =    \tau_t \ell_t - g_t-T_t +
                   {b_{t} \over R_t} $
holds for all $t$, and no other feasible taxes and interest rates yield higher utility for the agent.
Assuming that the agent's first-order conditions are necessary and sufficient,\footnote{As in \cite{aiyagari2002optimal}, we assume that the government's admissible bond holdings satisfy the agent's natural borrowing constraint.} and following \cite{marcet2019recursive}, we can rewrite the problem as
  \begin{equation}  \label{eq:Ramsey}
    \begin{aligned}    &\max_{(c_t,\ell_t,b_t)_{t=0}^\infty} \mathbb{E} \sum_{t=0}^{\infty} \beta^t \left(u(c_t)+v(\ell_t)\right) \\\mbox{ s.t. } &\beta b_t \mathbb{E}_t u'(c_{t+1}) \ge u'(c_t)(b_{t-1}-c_t)-v'(\ell_t)\ell_t,\\
&0\le  c_t\le \ell_t-g_t.
\end{aligned}
\end{equation}

The set of feasible actions is generally not convex.
 Ex-ante lotteries would relax the forward-looking constraint to 
 $$ E^{P((c_t,b_t,\ell_t),c_{t+1})} \left( \beta b_t \mathbb{E}_t u'(c_{t+1})-( u'(c_t)(b_{t-1}-c_t)-v'(\ell_t)\ell_t)\right)\ge 0. $$
It is impossible to give an economically meaningful interpretation of this constraint: each realization of the lottery determines a tax rate and new bond purchases. The agent can optimally decide how much to work, but cannot choose savings optimally. In fact, the government makes negative transfers in some lottery outcomes, which contradicts the assumption $T_t\ge 0$.

Ex-post lotteries, on the other hand, would yield a constraint of the form
$$ E^{P(c_{t+1})} \left( \beta b_t \mathbb{E}_t u'(c_{t+1})-( u'(c_t)(b_{t-1}-c_t)-v'(\ell_t)\ell_t)\right)\ge 0, $$
which has the straightforward interpretation that there is randomness in next period's consumption in addition to that induced by the exogenous shock. The agent anticipates this and makes current decisions accordingly.
We show below that the case of ex-post lotteries corresponds to the functional equation
$$ \begin{aligned} W(b,\mu,g) = \sup_{c,b'} \inf_{\gamma \ge 0} \Big[ & u(c)+v(c+g) + \mu b u'(c) +\gamma \left( u'(c)(c-b)+(c+g) v'(c+g) \right) \\ & + \beta E_g  W(b',\gamma,g') \Big] ,\end{aligned} $$
while for ex-ante lotteries the sup- and inf operators are interchanged.
 Whether one uses the sup-inf FE or the inf-sup FE matters greatly, both for the optimal value and for its economic interpretation.
Under the natural assumption that, in each period, agents observe the random tax rate and then make consumption, labor-supply and savings decisions, the relevant concept is that of ex-post lotteries. Random taxes in future periods
affect agents' current decisions. We show in a companion paper \citep{Shen-Kubler_Ramsey} that a lottery that imposes very high taxes with small probability in some states can improve welfare, because it lowers current interest rates and allows the government to finance its expenditure cheaply through debt and to roll over part of its current debt \citep[see also][]{citanna2024taxspots}.

\section{A general formulation} \label{sec:formulation}

We consider an optimization problem with forward-looking constraints similar to those in \cite{cole2012recursive} or \cite{marcet2019recursive}. Time is discrete and infinite, $ t=0,1,\ldots$.
Exogenous shocks $ (s_t) $ follow a Markov chain with transition $ \pi $ and take values in a finite set $ {\mathcal S}$.
A history of shocks up to some time $t$ is denoted by $ s^t=(s_0,\ldots,s_t) \in {\mathcal S}^{t+1}$. Each period, 
the planner chooses an action in a compact set $\mathcal{A}$ contingent on the shock history. We denote the infinite sequence of actions by $ (a(s^t))_{t\in\mathbb{N},s^t\in\mathcal{S}^{t+1}} \in {\mathcal A}^{\infty}$. At each $ s^t $, the value of the physical state is denoted by $x(s^t)$, and an action, $ a$, is feasible  if $ p(x(s^t),a(s^t),s_t) \ge 0$ for any $t=0,1,\cdots,s^t\in\mathcal{S}^{t+1}$. The physical state evolves according to $ x(s^t)=\zeta(x(s^{t-1}),a(s^{t-1}),s_{t-1})$\footnote{\(x(s^t)\) is indeed \(s^{t-1}\)-measurable. We write it as \(x(s^t)\) for convenience in the later theoretical development.} for $ t=1,2,\ldots$, with $x(s^{0})=x_0$. Note that $x(s^t)$ is determined by the past actions $(a(s_0),\dots,a(s^{t-1}))$. We therefore write $x(s^t)$ only when this dependence is clear from context; otherwise, we use the more explicit notation $x(s^t,a^{t-1})$, defined recursively by
\(
x(s^t,a^{t-1}) = \zeta\bigl(x(s^{t-1},a^{t-2}), a_{t-1}, s_{t-1}\bigr)\) for  $t=1,2,\dots$, with $x(s^0)=x_0$.
\subsection{The deterministic problem}
Given $\gamma=(\gamma^i)_{i=1}^{I},x_0,s_0$, the deterministic optimization problem is as follows.
\begin{equation}\label{equ:CK}
    \begin{aligned}      \max_{(a(s^t))_{t\in\mathbb{N}}^{s^t\in\mathcal{S}^{t+1}}\in \mathcal{A}^{\infty}} &\mathbb{E}_{s_0}\sum_{t=0}^{\infty}\beta^t\left(r(x(s^{t}),a(s^t),s_t)+\sum_{i=1}^{I}\gamma^ig^i(x(s^{t}),a(s^t),s_t)\right)\\
      \textbf{s.t. }&\mathbb{E}_{s^t}\sum_{n=0}^{\infty}\beta^n g^i(x(s^{t+n}),a(s^{t+n}),s_{t+n})\ge\bar{g}^i(x(s^{t}),a(s^t),s_t),\ \\&\forall t\in\mathbb{N},\, s^t=(s_0,\cdots, s_t)\in \mathcal{S}^{t+1},\,\forall i\in\{1,\cdots,I\},\\
      \text{where }&x(s^{t+1})=\zeta(x(s^{t}),a(s^t),s_t)\text{ and }p(x(s^{t}),a(s^t),s_t)\ge 0\ 
      \forall t\in \mathbb{N},\,s^t \in \mathcal{S}^{t+1}.
    \end{aligned}
\end{equation}
There are $ i=1,\ldots,I \ge 1 $ forward-looking constraints and each constraint $i$ restricts actions over the infinite future. We focus on this case for notational simplicity, and in Appendix~\ref{app:finho}, we prove that all our arguments can be extended to the case in which the constraints involve only the current and the next period.
For simplicity, we write $ \bar g^i $ instead of taking it as a function of the state and actions.

 The information structure will turn out to be an important issue in the following. In this deterministic problem, the only source of randomness is the exogenous shock process $(s_t)$. Hence, at time $t$, the available information is the shock history $s^t \in \mathcal{S}^{t+1}$, and the forward-looking constraint should be satisfied conditional on this information. In the lottery problem considered below, the action in each period is random. To verify that the incentive constraint in a certain period is satisfied, the agent observes the actions realized before (or even at) that period; thus, the information set includes $s^t$ together with $a^{t-1} \in \mathcal{A}^{t}$ or even $a^t \in \mathcal{A}^{t+1}$. We will address these details in the next section.

Throughout the paper, we make the following assumptions.
\begin{assumption}\label{ass:dynamic} \mbox{}
    \begin{enumerate}
        \item $(s_t)_{t=0}^{\infty}$ is a Markov process. The set of Markov states $\mathcal{S}$ is a finite set, and $\pi(s'|s)>0,\,\forall s,\,s'\in \mathcal{S}$.
        \item $\mathcal{A}\subset \mathbb{R}^n$ is a finite set, and $\mathcal{X}\subseteq\mathbb{R}^m$ is a countable set. $\zeta$ is a map from $\mathcal{X}\times \mathcal{A}\times \mathcal{S}$ to $\mathcal{X}$.
        \item The functions $r:\mathcal{X}\times \mathcal{A}\times \mathcal{S}\rightarrow\mathbb{R},$ and $g^i:\mathcal{X}\times \mathcal{A}\times \mathcal{S}\rightarrow\mathbb{R},\,\forall i\in\{1,\cdots,I\}$ are bounded. 
        \item $0<\beta<1$.
        \item For any $s_0\in \mathcal{S}$ and $x_0\in\mathcal{X}$, there exists a feasible point $(a(s^t))\in\mathcal{A}^{\infty}$ to \eqref{equ:CK}.
    \end{enumerate}
\end{assumption}
 The assumption that $\mathcal{A}$ is finite ensures that the information set considered later is countable, thereby simplifying the measure-theoretic arguments below. It has no practical relevance, as the set can be taken arbitrarily large.  Throughout the paper, for each state $x\in \mathcal{X}$ and shock $s\in \mathcal{S}$, we denote the set of feasible actions by
$
\tilde{\mathcal{A}}(x,s):=\{ a\in\mathcal{A} \mid p(x,a,s)\ge 0\}.
$
For each initial state $x_0\in\mathcal{X}$, we denote the corresponding feasible set of policies $a(s^t)$ by
\[
\tilde{\mathcal{A}}^{\infty}(x_0,s_0) := 
\left\{
(a(s^t))\in\mathcal{A}^{\infty} 
\ \middle|\ 
\begin{aligned}
& x(s^0)=x_0,\\
& a(s^t)\in\tilde{\mathcal{A}}(x(s^t),s_t),\\
& x(s^{t+1})=\zeta(x(s^t),a(s^t),s_t)
\end{aligned}
\quad \forall t\in\mathbb{N},\ s^t\in\mathcal{S}^{t+1}
\right\}.
\]

\subsection{\texorpdfstring{Ex-ante lotteries}{Ex-ante lotteries}}

 Let $\mathcal{P}(\tilde{\mathcal{A}}^{\infty}(x_0,s_0))$ denote the space of probability measures on $\tilde{\mathcal{A}}^{\infty}(x_0,s_0)$. It is straightforward to verify that $\tilde{\mathcal{A}}^{\infty}(x_0,s_0)$ is compact in the product topology generated by $\mathcal{A}$\footnote{Since $\mathcal{A}$ is compact, by Tychonoff's theorem, $\mathcal{A}^{\infty}$ is compact in the product topology. Furthermore, since the topology in $\mathcal{A}$ is metrizable, the product topology in $\mathcal{A}^{\infty}$ is metrizable (see Theorem 3.36 in \cite{aliprantis2006infinite}). Given $x_0\in\mathcal{X}$, since $\mathcal{A}$ is finite, it is obvious that $\tilde{\mathcal{A}}^{\infty}(x_0,s_0)$ is a closed subset of $\mathcal{A}^{\infty}$ in the product topology. Hence $\tilde{\mathcal{A}}^{\infty}(x_0,s_0)$ is compact in the product topology.} and hence standard to show that $\mathcal{P}(\tilde{\mathcal{A}}^{\infty}(x_0,s_0))$ is compact in the weak-* topology, using the Portmanteau theorem (for closedness, see Theorem 2.1(iv) in \cite{billingsley2013convergence}) and Prokhorov's theorem (for relative compactness, see Theorem 6.1 in \cite{billingsley2013convergence}).

We allow the planner to choose an element of  $\mathcal{P}(\tilde{\mathcal{A}}^{\infty}(x_0,s_0))$. In this case, the realization of the lottery must become part of the history. There is a one-to-one mapping between this realization and the realized action. 
At each $t$, the set of possible histories up to $t$ is defined by \(\mathcal{H}^t:=\mathcal{S}^{t+1}\times \mathcal{A}^{t}\), where any $h^t\in\mathcal{H}^t$ can be represented as \(h^t =(s_0,a_0,\cdots,s_{t-1},a_{t-1},s_t)\)\footnote{Note that the history up to time $t$ includes the exogenous shock realized at $t$ but not the action at $t$.}.  The finiteness of $\mathcal{A}$ then guarantees that histories are countable.
The incentive constraints must hold for all $h^t\in \mathcal{H}^t$ that have positive probability.
It will turn out to be useful to define the functional
    {\small \begin{equation}\label{equ:deff_lot}
f:\tilde{\mathcal{A}}^{\infty}(x_0,s_0)\rightarrow\mathbb{R},\  (a(s^t))\mapsto\mathbb{E}_{s_0}\sum_{t=0}^{\infty}\beta^t\left(r(x(s^t),a(s^t),s_t)+\sum_{i=1}^{I}\gamma^ig^i(x(s^t),a(s^t),s_t)\right),
   \end{equation}}
    and the map
    {\small \begin{equation}\label{equ:defg_lot}
    \begin{aligned}
& g:\tilde{\mathcal{A}}^{\infty}(x_0,s_0)\rightarrow\ell^{\infty},\ \ 
(a(s^t))\mapsto (g_{t,h^t,i}(a(s^t)))_{t\in\mathbb{N},\,h^t\in\mathcal{H}^t,\,i\in\{1,\cdots,I\}}:=\\&\Big(1_{\{(\tilde{a}_0,\cdots,\tilde{a}_{t-1})=(a(s_0),\cdots,a(s^{t-1}))\}}
\left(\mathbb{E}_{s^t}\sum_{n=0}^{\infty}\beta^n g^i(x(s^{t+n}),a(s^{t+n}),s_{t+n})-\bar{g}^i\right)\Big)_{\substack{t\in\mathbb{N},\,h^t\in\mathcal{H}^t,\\ i\in\{1,\cdots,I\}}}.
    \end{aligned}
    \end{equation}}
In this definition, $ 1_{\{(\tilde{a}_0,\cdots,\tilde{a}_{t-1})=(a(s_0),\cdots,a(s^{t-1}))\}}$ denotes the indicator function that is equal to one if the actions realized in the past are equal to $ (\tilde{a}_0,\cdots,\tilde{a}_{t-1}) $ (and zero otherwise).
When no ambiguity arises, we sometimes write, for brevity,
\[
\tilde{a}^{t-1}:=(\tilde{a}_0,\dots,\tilde{a}_{t-1}),
\qquad
a^{t-1}:=(a(s_0),\dots,a(s^{t-1})).
\]

We have the following definition.
\begin{definition}\label{def:lot_problem}
    The ex-ante lottery problem of \eqref{equ:CK} is as follows.
\begin{equation}\label{equ:simplified_CK_lot}
     \begin{aligned}
      \sup_{P\in\mathcal{P}(\tilde{\mathcal{A}}^{\infty}(x_0,s_0))} &\mathbb{E}_{s_0}^{(a(s^t))\sim P}\sum_{t=0}^{\infty}\beta^t\left(r(x(s^t),a(s^t),s_t)+\sum_{i=1}^{I}\gamma^ig^i(x(s^t),a(s^t),s_t)\right),\\
          \textbf{s.t. }&\mathbb{E}^{(a(s^t))\sim P}_{s^t}1_{\{\tilde{a}^{t-1}=a^{t-1}\}}\left(\sum_{n=0}^{\infty}\beta^n g^i(x(s^{t+n}),a(s^{t+n}),s_{t+n})-\bar{g}^i\right)\ge 0,\\
      &\forall t\in\mathbb{N},\,\forall h^t=(s_0,\tilde{a}_0,\cdots,s_{t-1},\tilde{a}_{t-1},s_t)\in \mathcal{H}^t,\,\forall i\in\{1,\cdots,I\}.
    \end{aligned}
\end{equation}
\end{definition}

Randomization is performed over infinite sequences of actions $ (a(s^t)) $, and each forward-looking constraint is assumed to hold on average, conditional on the entire history of actions and shocks up to $t.$ 

Using the definitions \eqref{equ:deff_lot} and \eqref{equ:defg_lot}, the ex-ante lottery problem can be written compactly as
\begin{equation}\label{equ:simplified_CK_lot_short}
      \sup_{P\in\mathcal{P}(\tilde{\mathcal{A}}^{\infty}(x_0,s_0))} \mathbb{E}^{(a(s^t))\sim P}f((a(s^t)))
          \textbf{ s.t. } \mathbb{E}^{(a(s^t))\sim P}g((a(s^t)))\ge 0.
\end{equation}
Since $\mathcal{A}$ is finite and the discount factor satisfies $0<\beta<1$, it is straightforward to verify that $f$ and $g_{t,h^t,i}$ defined in \eqref{equ:deff_lot} and \eqref{equ:defg_lot} are continuous under the product topology.

    The next theorem shows that the maximum of problem \eqref{equ:simplified_CK_lot} is attained because of the weak-* compactness of $\mathcal{P}(\tilde{\mathcal{A}}^{\infty}(x_0,s_0))$, and hence we can replace `sup' in \eqref{equ:simplified_CK_lot} with `max'.

    \begin{theorem}\label{thm:exist_opt_lot}
        There exists a maximizer $P^*\in\mathcal{P}(\tilde{\mathcal{A}}^{\infty}(x_0,s_0))$ to \eqref{equ:simplified_CK_lot}.
    \end{theorem}
    
  As in Kuhn's theorem on the equivalence of mixed and behavioral strategies in extensive-form games, attention can here be restricted to lotteries over $ {\cal A}$ after each history $ h^t $. 
    To this end, we denote by $ \Pi_{h^t\in\mathcal{H}^t}{\cal P}(\tilde{A}(h^t)) $ the set of sequences of probability measures which satisfy at each history $ h^t=(s_0,\tilde{a}_0,\cdots,s_{t-1},\tilde{a}_{t-1},s_t)\in \mathcal{H}^t$,
    $$ \psi(h^t) \in \mathcal{P}(\tilde{\mathcal{A}}(h^t)):={\cal P}(\tilde{\cal A}(x(s^t,\tilde{a}^{t-1}),s_t)).$$ That is to say, the constraint $ p(x_t,a,s_t)\ge 0 $ has to hold for each $ a $ in the support of $\psi(h^t)$.
$a(h^t)\sim\psi(h^t)$ then means
$
\mbox{prob}(a(h^t)=a|s^t,a_{0}=\tilde{a}_0,\cdots,a_{t-1}=\tilde{a}_{t-1})=\psi(h^t)(a). 
$
    The following theorem formalizes that it is sufficient to focus on lotteries in $ \Pi_{h^t\in\mathcal{H}^t}{\cal P}(\tilde{A}(h^t))$.
\begin{theorem}\label{thm:lot_equiv}
    The lottery problem \eqref{equ:simplified_CK_lot} is equivalent to
    {\small
    \begin{equation}\label{equ:simplified_CK_lot_statewise}
     \begin{aligned}
      \max_{\psi\in\Pi_{h^t\in\mathcal{H}^t}{\cal P}(\tilde{A}(h^t))} & \mathbb{E}_{s_0}^{(a(h^t)\sim\psi(h^t))}\sum_{t=0}^{\infty}\beta^t\left(r(x(s^t,a^{t-1}),a(h^t),s_t)
      +\sum_{i=1}^{I}\gamma^ig^i(x(s^t,a^{t-1}),a(h^t),s_t)\right)\\
      \textbf{s.t. } & \mathbb{E}^{(a(h^t)\sim \psi(h^t))}_{s^t} 1_{\{\tilde{a}^{t-1}=a^{t-1}\}}\\
&\qquad\left(\sum_{n=0}^{\infty}\beta^n \left( g^i(x(s^{t+n},a^{t+n-1})),a(h^{t+n}),s_{t+n})-\bar{g}^i\right)\right)\ge 0,\\
      &\forall t\in\mathbb{N},\,\forall h^t=(s_0,\tilde{a}_0,\cdots,s_{t-1},\tilde{a}_{t-1},s_t)\in \mathcal{H}^t,\,\forall i\in\{1,\cdots,I\}.  
    \end{aligned}
\end{equation}}
\end{theorem}

\subsection{\texorpdfstring{Ex-post lotteries}{Ex-post lotteries}}

In the ex-ante lottery problem \eqref{equ:simplified_CK_lot}, the incentive constraint must be satisfied conditional on the history $h^t$, before the period-$t$ lottery has been realized. It is natural to consider an alternative lottery problem in which the incentive constraint must hold after the period-$t$ lottery has been realized. To discuss the issue in more detail, we define the post-action history set slightly differently as 
$
\tilde{\mathcal{H}}^t=\mathcal{S}^{t+1}\times \mathcal{A}^{t+1}, 
$
where any $\tilde{h}^t\in \tilde{\mathcal{H}}^t$ is represented as
$
\tilde{h}^t=(s_0,a_0,\cdots,s_t,a_{t}),\text{ or }\tilde{h}^t=(s^t,a^t),\text{ where } a^t=(a_0,\cdots,a_{t}).
$
The difference between $h^t$ and $\tilde{h}^t$ is that $\tilde{h}^t$ includes both the shock \textit{and} the action at $t$, while $h^t$ includes only the shock at $t$ (and the history before $t$).
We then define the ex-post lottery problem as follows.
\begin{definition}\label{def:lot_problem-expost}
    The ex-post lottery problem is as follows
\begin{equation}\label{equ:simplified_CK_lotexpost}
     \begin{aligned}
      \max_{P\in\mathcal{P}(\tilde{\mathcal{A}}^{\infty}(x_0,s_0))} &\mathbb{E}_{s_0}^{(a(s^t))\sim P}\sum_{t=0}^{\infty}\beta^t\left(r(x(s^t,a^{t-1}),a(s^t),s_t)+\sum_{i=1}^{I}\gamma^ig^i(x(s^t,a^{t-1}),a(s^t),s_t)\right)\\
      \textbf{s.t. }&\mathbb{E}^{(a_t(s^t))\sim P}_{s^t}1_{\{\tilde{a}^t=a^{t}\}}\left(\sum_{n=0}^{\infty}\beta^n g^i(x(s^{t+n},a^{t+n-1}),a(s^{t+n}),s_{t+n})-\bar{g}^i\right)\ge 0,\\
      &\forall t\in\mathbb{N},\,\forall \tilde{h}^t=(s_0,\tilde{a}_0,\cdots,s_{t-1},\tilde{a}_{t-1},s_t,\tilde{a}_{t})\in \tilde{\mathcal{H}}^t,\,\forall i\in\{1,\cdots,I\}.
    \end{aligned}
\end{equation}
\end{definition}
\section{\texorpdfstring{Duality gap and lotteries}{Duality gap and lotteries}}\label{sec:duality}
In this section, we relate the Lagrangian dual problem to the lottery problems defined above. We show that if the Lagrange multiplier is allowed to depend on the history $h^t$, the dual value corresponds to  an ex-ante lottery problem in which the incentive constraint must hold before the current action is realized. If the multiplier is also allowed to depend on the current action $a_t$, the dual value corresponds to an ex-post lottery problem in which the incentive constraint must hold  after the current lottery has been realized. The theoretical foundation for this relationship lies in the fact that the duality gap can be bridged by randomization. To make this precise, we need to introduce a few mathematical facts.

\subsection{Mathematical facts}
 We first consider an abstract setting of constrained optimization in infinite dimensional spaces. To be consistent with standard conventions in optimization theory, we consider the minimization problem in this section. For the purpose of the main text, we present only the essential background; additional details are deferred to Appendix~\ref{app:math}.

\begin{definition}\label{def:opt}
    Let $\Omega$ be a non-empty set, $Y$ be a Banach space with a closed positive cone $P\subset Y$, $f:\Omega\rightarrow\mathbb{R}\cup\{+\infty\}$ be a proper extended real-valued functional, and $g:\Omega\rightarrow Y$ be an arbitrary map. The \textbf{optimization problem} is defined by
    \begin{equation}\label{equ:def_opt}
    \inf_{x\in \Omega}f(x)
        \text{ s.t. } g(x)\le \theta_{Y},
    \end{equation}
    where $g(x)\le \theta_Y$ means $g(x)\in -P$.
\end{definition}
We define the \textbf{perturbation functional} $v:Y\rightarrow\mathbb{R}\cup\{+\infty,-\infty\}$ by letting $v(y)$ be the optimal value of the following \textbf{perturbed problem}
    \begin{equation}\label{equ:def_pert}
       v(y)= \inf_{x\in \Omega}f(x)
        \text{ s.t. } g(x)\le y,
    \end{equation}
     where $g(x)\le y$ means $g(x)-y\in -P$.
 The \textbf{Lagrangian functional} $L:\Omega\times Y^*_+\rightarrow\mathbb{R}\cup\{+\infty\}$ to the optimization problem \eqref{equ:def_opt} is defined by
     \begin{equation}\label{equ:def_Lagdual}
         L(x,y^*)=f(x)+\langle y^*,g(x)\rangle.
     \end{equation}
     Here, $Y^*_+$ denotes the set
     $
     \{y^*\in Y^{*}|\langle y^*,y\rangle \ge 0,\,\forall y\ge \theta_{Y}\}.
     $
     Let $\Omega,\,Y,\,f,\,g$ be defined as in Definition \ref{def:opt}. The \textbf{inf-sup problem}, or the \textbf{primal problem}, is defined by
     \begin{equation}\label{equ:def_infsup}
         p:=\inf_{x\in \Omega}\sup_{y^*\in Y^{*}_+}L(x,y^*).
     \end{equation}
     Similarly, the \textbf{sup-inf problem}, or the \textbf{dual problem}, is defined by
     \begin{equation}\label{equ:def_supinf}
        d:=\sup_{y^*\in Y^{*}_+} \inf_{x\in \Omega}L(x,y^*).
     \end{equation}

In general, $ d\le p $, and the difference between the two values is called the duality gap. To show how this gap can be bridged by convexification, we need the following definitions.
    
\begin{definition} \label{def:biconjcugate}
    Let $X$ be a Banach space, $f:X\rightarrow\mathbb{R}\cup\{+\infty,-\infty\}$ be an extended real-valued functional.
Its epigraph is defined as
$$
    \text{epi}(f):=\{(x,r)\in X\times \mathbb{R}| f(x)\le r\}.
    $$
    The functional   $f^*:X^*\rightarrow\mathbb{R}\cup\{+\infty,-\infty\}$ defined by
    $$
    f^*(x^*):=\sup_{x\in X}\{\langle x^*,x \rangle-f(x)\}
    $$
    is called the \textbf{convex conjugate}, or \textbf{conjugate} of $f$, and the mapping
    $
    f\mapsto f^*
    $
    is called the \textbf{Legendre--Fenchel} transformation. Furthermore, $f^{**}:X\rightarrow\mathbb{R}\cup\{+\infty,\,-\infty\}$\footnote{If $X$ is reflexive, then $X=X^{**}$, hence $f^{**}$ is indeed the conjugate of $f^*$. If $X$ is not reflexive, then $X\subset X^{**}$ is a closed subset of $X^{**}$, hence $f^{**}$ can be regarded as the conjugate of $f^*$ restricted to $X$.} defined by
    $$
    f^{**}(x):=\sup_{x^*\in X^*}\{\langle x^*,x \rangle-f^*(x^*)\}
    $$
    is called the \textbf{biconjugate} of $f$.
\end{definition}

The following theorem, which is a variation of the Fenchel--Moreau theorem, constitutes the theoretical foundation of our analysis. It appears to be known in mathematics\footnote{According to Theorem 5 in \cite{rockafellar1974conjugate} (p.~16), one might obtain the same conclusion when $f$ satisfies
        $
        \text{lsc }\text{co }(f)>-\infty,\quad \forall x\in X,
        $
         which is a weaker condition than ours. However, the proof in \cite{rockafellar1974conjugate}  relies on a geometric result that was not rigorously stated in the book.
         The finite-dimensional version of this theorem was provided in Theorem 1.3.5 in \cite{urruty1993convex} (p.~45).} but a clean reference is difficult to find. For completeness, we provide a proof in Appendix~\ref{app:sec2}.

\begin{theorem}\label{thm:bicon}Let $X$ be a Banach space, $f:X\rightarrow\mathbb{R}\cup \{+\infty\}$ be a proper extended real-valued functional and assume that
   there exists $\underline{x}^*\in X^*$, $\underline{\beta}\in\mathbb{R}$, such that
    \begin{equation}\label{equ:regularity}
        \langle \underline{x}^*,x\rangle+\underline{\beta} \le f(x),\quad \forall x\in X.
    \end{equation}
Then the following holds.
$$
\text{epi}(f^{**})=\text{cl } \text{co } \text{epi}(f),
$$
where cl co denotes the closure of the convex hull.
\end{theorem}
    
    The following theorem states that the duality gap can be bridged by randomization, in the sense that the epigraph of the perturbed dual problem is the closure of the convex hull of the epigraph of the perturbed primal problem.

\begin{theorem}\label{thm:dualgap}
    Let $\Omega,\,Y,\,f,\,g$ be defined as in Definition \ref{def:opt}. Then $p=v(\theta_Y)$ and $d=v^{**}(\theta_Y)$, where $v^{**}$ denotes the biconjugate of $ v $. 
    Moreover, if $v$ is proper and admits an affine minorant in the sense of \eqref{equ:regularity}, then $$
\text{epi}(v^{**})=\text{cl } \text{co } \text{epi}(v).
$$
\end{theorem}
The basic intuition for the results in a finite-dimensional setting is explained in detail in
Theorems 2.3 and 2.4 in \cite{shen2025lagrangian}. The details of how Theorem \ref{thm:dualgap} relates to randomization are provided in Appendix~\ref{app:math_cor}.

\subsection{The dual problem of the deterministic Lagrangian}

Since the forward-looking constraints in the lottery problem \eqref{equ:simplified_CK_lot} have to hold for each history that has positive probability, it is useful to require the same for the deterministic problem. Although in this problem this requirement does not change the optimal solution compared to \eqref{equ:CK}, it clarifies how its dual corresponds to the lottery problem. This modified deterministic problem is as follows.
{\small
\begin{equation}\label{equ:CK_equiv}
    \begin{aligned}
      \sup_{(a(s^t))\in \tilde{\mathcal{A}}^{\infty}(x_0,s_0)} &\mathbb{E}_{s_0}\sum_{t=0}^{\infty}\beta^t\left(r(x(s^t),a(s^t),s_t)+\sum_{i=1}^{I}\gamma^ig^i(x(s^t),a(s^t),s_t)\right)\\
      \textbf{s.t. }&1_{\{(\tilde{a}_0,\cdots,\tilde{a}_{t-1})=(a(s_0),\cdots,a(s^{t-1}))\}}\\
      &\qquad\left(\mathbb{E}_{s_t}\sum_{n=0}^{\infty}\beta^n g^i(x(s^{t+n}),a(s^{t+n}),s_{t+n})-\bar{g}^i\right)\ge 0,\\
      &\forall t\in\mathbb{N},\,\forall h^t=(s_0,\tilde{a}_0,\cdots,s_{t-1},\tilde{a}_{t-1},s_t)\in \mathcal{H}^t,\,\forall i\in\{1,\cdots,I\},
    \end{aligned}
\end{equation}
}or for short,
\begin{equation}\label{equ:CK_reformulatemt}
\sup_{(a(s^t))\in\mathcal{\tilde{A}}^{\infty}(x_0,s_0)} f((a(s^t)))
    \textbf{ s.t. } g((a(s^t)))\ge 0.
\end{equation}
We define $\Lambda$ as the $\ell^1$ space\footnote{Strictly speaking, the multipliers  lie in \(\ell^{\infty,*}\) rather than in \(\ell^1\), with \(\ell^1 \subset \ell^{\infty,*}\). 
For now we proceed heuristically with multipliers in \(\ell^1\); when we turn to the dual problem, we will show that this restriction does not affect the dual value.} of the ``non-normalized''  Lagrange multipliers,  $$ ((\beta^t\lambda^i(h^t)\pi^t(s^t|s_0))_{i\in\{1,\cdots, I\},t\in \mathbb{N},\,h^t\in\mathcal{H}^t},$$ i.e., 
\begin{equation}\label{Lag:space}(\lambda^i(h^t))\in\Lambda:=\{(\lambda^i(h^t)\ge 
0)|\sum_{t=0}^{\infty}\sum_{h^t\in\mathcal{H}^t}\sum_{i=1}^{I}\beta^t \lambda^i(h^t)\pi^t(s^t|s_0)<\infty\}.\end{equation}

For any given multiplier $\{\lambda^i(h^t)\}$ (which may depend on the full history), the Lagrangian of \eqref{equ:CK_equiv} can then be formulated as\footnote{In the formula \eqref{equ:Lag_CK_refor}, we have \[\begin{aligned}
&\mathbb{E}_{s_0}\left[\lambda^i(h^t)\mathbb{E}_{s_t}\left(\sum_{n=0}^{\infty}\beta^n g^i(x(s^{t+n}),a(s^{t+n}),s_{t+n})-\bar{g}^i\right)\right]\\=&\sum_{h^t\in\mathcal{H}^t}\pi(s^t|s_0)\lambda^i(h^t)1_{\{(\tilde{a}_0,\cdots,\tilde{a}_{t-1})=(a(s_0),\cdots,a(s^{t-1}))\}}\left(\mathbb{E}_{s_t}\sum_{n=0}^{\infty}\beta^n g^i(x(s^{t+n}),a(s^{t+n}),s_{t+n})-\bar{g}^i\right).\end{aligned}\]}
\begin{equation}\label{equ:Lag_CK_refor}
    \begin{aligned}
    &L((a(s^t)),(\lambda^i(h^t));(\gamma^i),x_0,s_0)\\
=&\mathbb{E}_{s_0}\sum_{t=0}^{\infty}\beta^t\left\{r(x(s^t),a(s^t),s_t)+\sum_{i=1}^{I}\gamma^ig^i(x(s^t),a(s^t),s_t)\right.\\
&\left.+\sum_{i=1}^I \left[\lambda^i(h^t)\mathbb{E}_{s_t}\left(\sum_{n=0}^{\infty}\beta^n g^i(x(s^{t+n}),a(s^{t+n}),s_{t+n})-\bar{g}^i\right)\right]\right\},
\end{aligned}
\end{equation}
and the dual value of problem \eqref{equ:CK_equiv} is given by
\begin{equation}\label{equ:dual_value_def}
D(\gamma,x_0,s_0)=\inf_{(\lambda^i(h^t))\in \Lambda} \sup_{(a(s^t))\in \tilde{\mathcal{A}}^{\infty}(x_0,s_0)} L((a(s^t)),(\lambda^i(h^t));(\gamma^i),x_0,s_0).
\end{equation}
According to the definition of the dual problem (see \eqref{equ:def_supinf}), the dual problem of \eqref{equ:CK_reformulatemt} should be defined as\footnote{Note that the primal problem now is a maximization problem, and hence the dual problem has the inf-sup form.}
\begin{equation}\label{equ:dual_infsup}
d=\inf_{\lambda\in\ell^{\infty,*}_+}\sup_{(a(s^t))\in\tilde{\mathcal{A}}^{\infty}(x_0,s_0)}f((a(s^t)))+\langle\lambda,g((a(s^t)))\rangle.
\end{equation}
To establish the existence of Lagrange multipliers in $\ell^1$ rather than in $\ell^{\infty,*}$ as in \eqref{equ:dual_infsup}, i.e., to show that
\begin{equation}\label{equ:dual_value_l1}
\begin{aligned}
d=&\inf_{\lambda\in\ell^{1}_+}\sup_{(a(s^t))\in\tilde{\mathcal{A}}^{\infty}(x_0,s_0)}f((a(s^t)))+\langle\lambda,g((a(s^t)))\rangle\\
=&\inf_{(\lambda^i(h^t))\in\Lambda}\sup_{(a(s^t))\in\tilde{\mathcal{A}}^{\infty}(x_0,s_0)}L((a(s^t)),(\lambda^i(h^t));(\gamma^i),x_0,s_0)=D(\gamma,x_0,s_0),
\end{aligned}
\end{equation}
 we follow the approach of \cite{dechert1982lagrange} and \cite{pavoni2018dual}. The detailed argument is provided in Online Appendix~\ref{lagrange} (see Theorem \ref{dual_for}).
\subsection{\texorpdfstring{Ex-ante lotteries and the dual problem}{Ex-ante lotteries and the dual problem}}
Motivated by Theorem \ref{thm:dualgap} and its corollaries in Appendix~\ref{app:math_cor}, we now explore how the interpretation and the optimal solution of \eqref{equ:CK} change if the planner can randomize over actions. 
The following theorem states that the value of the ex-ante lottery problem is identical to the dual value of the modified deterministic problem \eqref{equ:CK_equiv}.
\begin{theorem}\label{thm:lot_dual_equiv}
    Under Assumption \ref{ass:dynamic}, the maximum of problem \eqref{equ:simplified_CK_lot}, denoted by $V(\gamma,x_0,s_0)$, equals the dual value $D(\gamma,x_0,s_0)$ of the modified deterministic problem \eqref{equ:CK_equiv}.
\end{theorem}
\subsection{Ex-post lotteries}

We allow the Lagrange multiplier associated with the forward-looking constraint starting at $t$ to depend on $ \tilde{h}^t $ (and not only on $ h^t $),
and define $ \tilde \Lambda$ as the $\ell^1$ space of the ``non-normalized''  Lagrange multipliers $((\beta^t\lambda^i(\tilde{h}^t)\pi^t(s^t|s_0))_{t\in \mathbb{N},\,\tilde{h}^t\in \tilde{\mathcal{H}}^t}$ analogously to (\ref{Lag:space}). The Lagrangian function then becomes
\begin{equation}\label{equ:Lag_CKextend}
\begin{aligned}
    &\tilde{L}((a(s^t)),(\lambda^i(\tilde{h}^t));(\gamma^i),x_0,s_0)\\
=&\mathbb{E}_{s_0}\sum_{t=0}^{\infty}\beta^t\left\{r(x(s^t),a(s^t),s_t)+\sum_{i=1}^{I}\gamma^ig^i(x(s^t),a(s^t),s_t)\right.\\
&\left.+\sum_{i=1}^I\left[\lambda^i(\tilde h^t)\mathbb{E}_{s_t}\left(\sum_{n=0}^{\infty}\beta^n g^i(x(s^{t+n}),a(s^{t+n}),s_{t+n})-\bar{g}^i\right)\right]\right\}.
\end{aligned}
\end{equation}
For any $\gamma\in\mathbb{R}_{+}^{I},\,x_0\in \mathcal{X},\,s_0\in \mathcal{S}$, we define the dual value function as
\begin{equation}\label{equ:dual_value_def_extend}
\tilde{D}(\gamma,x_0,s_0)=\inf_{(\lambda^i(\tilde{h}^t))\in\tilde{\Lambda}}\sup_{(a(s^t))\in\tilde{\mathcal{A}}^{\infty}(x_0,s_0)}\tilde{L}((a(s^t)),(\lambda^i(\tilde{h}^t));(\gamma^i),x_0,s_0).
\end{equation}

As in the ex-ante problem, the planner chooses a probability distribution over infinite sequences. However, 
as in our example in Section~\ref{sec:eclott}, the difference between the two problems is that in the ex-post problem the constraint has to hold conditional on each $ a_t $ in the support of the lottery. A history is defined as $ \tilde{h}^t $ and is assumed to include the action at $t$. As in the case of ex-ante lotteries, it is equivalent to consider lotteries over actions after each history. 
We obtain the following analog of Theorem \ref{thm:lot_dual_equiv}.
\begin{theorem}\label{thm:lot_dual_equiv-expost}
    Under Assumption \ref{ass:dynamic}, the maximum of problem \eqref{equ:simplified_CK_lotexpost}, denoted by $\tilde{V}$, equals $\tilde D(\gamma,x_0,s_0)$.
\end{theorem}

\section{\texorpdfstring{Recursive formulations}{Recursive formulations}}
\label{sec:lotteries}
In this section, we show that the inf-sup problem for the Lagrangian functionals can be solved recursively. 
In the case discussed in \cite{marcet2019recursive}, where the objective function is concave and the constraint set is convex, the inf-sup and the sup-inf functional equations have the same value, which equals the optimal value of the original problem. With non-convexities, the two values may differ and both may strictly exceed the value of the original problem. Indeed, the dual value $D(\gamma,x_0,s_0)$ defined in \eqref{equ:dual_value_def} can be computed using the inf-sup FE, while the dual value $\tilde{D}(\gamma,x_0,s_0)$ defined in \eqref{equ:dual_value_def_extend} can be computed using the sup-inf FE. Together with the results of Section~\ref{sec:duality}, this shows that once lotteries are allowed, both functional equations have a natural interpretation: each gives the value of a lottery problem.
\subsection{\texorpdfstring{Recursive dual value function: the inf-sup FE}{Recursive dual value function: the inf-sup FE}}
For the dual value $D(\gamma,x_0,s_0)$ defined in \eqref{equ:dual_value_def}, we have the following theorem.
\begin{theorem}\label{thm:recursive1}
    For any $x\in \mathcal{X},\,s\in \mathcal{S},\,\gamma\in\mathbb{R}_+^{I}$, the dual value function $D(\gamma,x,s)$ defined in \eqref{equ:dual_value_def} satisfies the following recursive equation
    \begin{equation}\label{equ:recursive_dual}
        \begin{aligned}
        D(\gamma,x,s)=\inf_{\lambda\in \mathbb{R}_+^I}\sup_{a\in\tilde{\mathcal{A}}(x,s)}\Big[&\left(r(x,a,s)+\sum_{i=1}^{I}\left(\gamma^ig^i(x,a,s)+\lambda^i(g^i(x,a,s)-\bar{g}^i)\right) \right)\\
&\quad+\beta\mathbb{E}_{s}D(\gamma+\lambda,x',s')\Big],
        \end{aligned}
        \end{equation}
        where $ x'=\zeta(x,a,s).$
\end{theorem}

The spirit of the proof is similar to that of Proposition 3 in \cite{pavoni2018dual}. The key observation is that the Lagrangian $L(a(s^t),\lambda^i(h^t);(\gamma^i),x_0,s_0)$ admits a natural recursive structure. Formally, we have
{\footnotesize 
 $$
\begin{aligned}
&L((a_t)_{t\ge 0},(\lambda_t)_{t\ge 0};(\gamma^i),x_0,s_0)=\\&\mathbb{E}_{0}\sum_{t=0}^{\infty}\beta^t\left\{r(x_t,a_t,s_t)+\sum_{i=1}^{I}\gamma^ig^i(x_t,a_t,s_t)+\sum_{i=1}^I \left[\lambda^i_t\mathbb{E}_{t}\left(\sum_{n=0}^{\infty}\beta^n g^i(x_{t+n},a_{t+n},s_{t+n})-\bar{g}^i\right)\right]\right\} \\
=&r(x_0,a_0,s_0)+\sum_{i=1}^{I}\left[\gamma_0^{i}g^i(x_0,a_0,s_0)+\lambda_0^i(g^i(x_0,a_0,s_0)-\bar{g}^i)\right] + \\
& \beta\mathbb{E}_0\underbrace{
\mathbb{E}_1\sum_{t=1}^{\infty}\beta^{t-1}\left\{ (r(x_{t},a_{t},s_{t})+\sum_{i=1}^{I}(\gamma_0^i+\lambda_0^i)g^i(x_t,a_t,s_t))+\sum_{i=1}^{I}\left[\lambda_{t}^i\mathbb{E}_t\left(\sum_{n=0}^{\infty}\beta^n g^i(x_{t+n},a_{t+n},s_{t+n})-\bar{g}^i\right)\right]\right\}}_{L((a_t)_{t\ge 1},(\lambda_t)_{t\ge 1};\gamma_0+\lambda_0,x_1,s_1)}.
\end{aligned}
 $$}
 Given that
\begin{equation}\label{eq:before_interchange}
D(\gamma,x_0,s_0)=\inf_{\lambda_0}\inf_{\lambda_1}\cdots\sup_{a_0}\sup_{a_1}\cdots L((a_t)_{t\ge 0},(\lambda_t)_{t\ge 0};(\gamma^i),x_0,s_0),
\end{equation}
the recursive equation for $D$ follows once we establish the alternative representation
\begin{equation}\label{eq:after_interchange}
D(\gamma,x_0,s_0)=\inf_{\lambda_0}\sup_{a_0}\inf_{\lambda_1}\sup_{a_1}\cdots L((a_t)_{t\ge 0},(\lambda_t)_{t\ge 0};(\gamma^i),x_0,s_0).
\end{equation}
This interchange of inf and sup does not rely on a minimax theorem; instead, it relies only on the following lemma, which allows us to pass each inf operator to the front of the expression.

 \begin{lemma}\label{lem:tech_infsup}
  Suppose that $X$ and $Y$ are non-empty sets and that $f:X\times Y\to\mathbb{R}$ is an arbitrary functional. Set $\Gamma:=\{\boldsymbol{y}:X\to Y\}$. If $\tilde{X}\subseteq X$, then
    $$
    \inf_{\boldsymbol{y}\in \Gamma}\sup_{x\in \tilde{X}}f(x,\boldsymbol{y}(x))=\sup_{x\in \tilde{X}}\inf_{\boldsymbol{y}\in \Gamma}f(x,\boldsymbol{y}(x)).
    $$
\end{lemma}   

The intuition behind Lemma \ref{lem:tech_infsup} is that, whenever $y$ can be regarded as a policy of $x$, the operators $\sup_{x}$ and $\inf_{y}$ can be interchanged. This applies to the Lagrange multipliers, since each $\lambda_t$ depends on the history $h^t$, which contains information about $a_0,\ldots,a_{t-1}$; hence the operator $\inf_{\lambda_t}$ in \eqref{eq:after_interchange} can be passed to the front of the expression. The rigorous proofs of Theorem \ref{thm:recursive1} and Lemma \ref{lem:tech_infsup} are provided in Appendix~\ref{app:sec2}.

In Online Appendix~\ref{app:fixedpoint}, we define a space of functions that are convex and Lipschitz continuous in $ \gamma $ and satisfy certain boundedness conditions. We denote this set by $ {\mathcal N}$. The following theorem establishes the existence of a solution in this space and shows that the dual value $D$ equals the largest solution in this space.
    \begin{theorem}\label{thm:contraction_infsup}
The following functional equation, 
 $$ \begin{aligned}  F(\gamma,x,s)=\inf_{\lambda\in \mathbb{R}_+^I}\sup_{a\in\tilde{\mathcal{A}}(x,s)}\Big[&\left(r(a,x,s)+\sum_{i=1}^{I}\left(\gamma^ig^i(a,x,s)+\lambda^i(g^i(a,x,s)-\bar{g}^i)\right)\right)\\
&\quad+\beta\mathbb{E}_{s}F(\gamma+\lambda,x',s')\Big] \end{aligned} $$
 where $x'=\zeta(x,a,s)$, has a solution in $\mathcal{N}$.
Moreover, $D(\gamma,x,s)$ is the largest solution in $\mathcal{N}$.
\end{theorem}

\subsection{The sup-inf functional equation}
What would happen if one interchanged the inf and sup operators in the inf-sup functional equation
(\ref{equ:recursive_dual})? The following theorem follows directly from Lemma \ref{lem:tech_infsup} and shows that the dual value $\tilde{D}$ defined in \eqref{equ:dual_value_def_extend} satisfies the sup-inf functional equation.
\begin{theorem}\label{thm:recursive2}
    For any $x\in \mathcal{X},\,s\in \mathcal{S},\,\gamma\in\mathbb{R}_+^{I}$, the dual value function $\tilde{D}(\gamma,x,s)$ defined in \eqref{equ:dual_value_def_extend} satisfies the following recursive equation
    \begin{equation}\label{equ:recursive_dual_extend}
 \begin{aligned}
 \tilde{D}(\gamma,x,s)=\sup_{a\in\tilde{\mathcal{A}}(x,s)}\inf_{\lambda\in \mathbb{R}_+^I}\Big[&\left(r(x,a,s)+\sum_{i=1}^{I}\left(\gamma^ig^i(x,a,s)+\lambda^i(g^i(x,a,s)-\bar{g}^i)\right)\right)\\
&\quad+\beta\mathbb{E}_{s}\tilde{D}(\gamma+\lambda,x',s')\Big],
 \end{aligned}\\
    \end{equation}
    where $x'=\zeta(x,a,s)$.
\end{theorem}
The intuition behind the theorem is the same as for the inf-sup case. The crucial difference is that $\lambda(\tilde{h}^t)$ depends on $a(s^t)$, and hence one can also interchange $\inf_{\lambda_t}$ and $\sup_{a_t}$  in \eqref{eq:after_interchange}.
As in the inf-sup case, we also establish the existence of a solution.
   \begin{theorem}\label{thm:contraction_supinf}
The following functional equation, 
 $$ \begin{aligned}  F(\gamma,x,s)=\sup_{a\in\tilde{\mathcal{A}}(x,s)}\inf_{\lambda\in \mathbb{R}_+^I}\Big[&\left(r(a,x,s)+\sum_{i=1}^{I}\left(\gamma^ig^i(a,x,s)+\lambda^i(g^i(a,x,s)-\bar{g}^i)\right)\right)\\
&\quad+\beta\mathbb{E}_{s}F(\gamma+\lambda,x',s')\Big] \end{aligned} $$
 where $x'=\zeta(x,a,s)$, has a solution in $\mathcal{N}$.
Moreover, $\tilde{D}(\gamma,x,s)$ is the largest solution in $\mathcal{N}$.

\end{theorem}

\section{Policies}\label{sec:policies}

Having established in the previous section that both the inf-sup and sup-inf functional equations yield the correct optimal value for their respective lottery problems, the remaining task is to recover the corresponding optimal lotteries themselves. This section develops an approach for doing so and for characterizing optimal lottery policies.

We begin with the inf-sup equation \eqref{equ:recursive_dual} and extend the approach to the sup-inf equation at the end of this section.
Because we have established monotone convergence, 
solving for the value function $D$ (the FE's largest fixed point) can be done iteratively with a sufficiently large starting point. In each iteration, the outer problem is a convex minimization problem that can be efficiently solved with subgradient descent. Solving the maximization problem (i.e., the inner problem) can be challenging, and one might have to resort to grid search for this step (see \cite{Shen-Kubler_Ramsey} for a practical application).

Solving for the value function $D$ also delivers, for each state, a minimizer\footnote{The minimizer might not be unique, but the argument holds for any minimizer.} $\lambda^*$ and a corresponding set of maximizers $\{ a(\lambda^*) \}$. Although this arises from a problem without lotteries, it provides the key to recovering optimal randomized policies, since it follows from \cite{shen2025lagrangian} that the support of an optimal lottery must lie in the set of maximizers $\{a(\lambda^*)\}$. The challenge is to construct an optimal lottery over these support points with the critical requirement that it satisfies all forward-looking constraints.
To do this, we follow \cite{marimon2021envelope} in incorporating the promised value into the functional equation, to track these forward-looking constraints directly. As in \cite{marimon2021envelope}, we first establish that the set of feasible promised values coincides with the subdifferential of 
$D$. Our key contribution is to adapt the iterative method of \cite{shen2025lagrangian} to recover the optimal policy in a way that avoids the costly step of computing the full set of feasible promised values: any single element of the subdifferential suffices to compute optimal policies, so our method requires only a subgradient of 
$D$ with respect to the multiplier.

\subsection{Preliminary results}

A key prerequisite for recovering the policy is to establish that the infimum and supremum in \eqref{equ:recursive_dual} are attained. Since the action set is assumed to be finite, this is equivalent to showing that there exists a $\lambda^*$ that solves \eqref{equ:recursive_dual}.
A standard Slater condition (Assumption~\ref{ass:Slater} in the Online Appendix) is assumed throughout the rest of the paper. The condition is sufficient to ensure existence; we include a formal result as Theorem \ref{thm:exist_bdd_lag} in the Online Appendix.

As explained above, to recover optimal policies, we need to
 introduce a utility promise $ \phi ^i$ for $ i=1,\ldots,I $ and define
{\small \begin{equation}\label{equ:simplified_CK_lot_statewise_promised}
     \begin{aligned}
       W(\gamma_0,x_0,s_0,\phi)= & \max_{\psi\in\Pi_{h^t\in\mathcal{H}^t}{\cal P}(\tilde{A}(h^t))}  \mathbb{E}_{s_0}^{(a(h^t)\sim\psi(h^t))}\sum_{t=0}^{\infty}\beta^t\left(r(x_t,a_t,s_t)+\sum_{i=1}^{I}\gamma^ig^i(x_t,a_t,s_t)\right),\\
      \textbf{s.t. }&\mathbb{E}^{(a(h^t)\sim \psi(h^t))}_{h^t}
\left(\sum_{n=0}^{\infty}\beta^n g^i(x_{t+n},a_{t+n},s_{t+n})-\bar{g}^i\right)\ge 0,\\
      &\mathbb{E}_{s_0}^{(a(h^t)\sim \psi(h^t))}\left(\sum_{t=0}^{\infty}\beta^t g^i(x_{t},a_{t},s_{t})-\phi^i\right)\ge 0, \\
            &\forall t\in\mathbb{N},\,\forall h^t=(s_0,\tilde{a}_0,\cdots,s_{t-1},\tilde{a}_{t-1},s_t)\in \mathcal{H}^t,\ \forall i\in\{1,\cdots,I\}.
    \end{aligned}
\end{equation}}

 The following lemma shows the equivalence between the feasible promised values and the subdifferential of $D$, which we denote by $\partial D(\gamma,x,s).$\footnote{For brevity, \(\partial D(\gamma,x,s)\) denotes the subdifferential of \(D(\cdot,x,s)\) with respect to \(\gamma\).  Note that $D(\gamma,x,s)$ extends to all $\gamma\in\mathbb{R}^I$. Throughout this section, any subdifferential $\partial D$ at a point on $\partial\mathbb{R}_+^I$ is understood with respect to this extension.}
\begin{lemma}\label{lem:equiv_subgradient_lot}
    For any $\gamma\in\mathbb{R}_+^I,\,x\in\mathcal{X},\, s\in\mathcal{S}$, the following are equivalent:
    \begin{enumerate}
        \item $\phi\in\partial D(\gamma,x,s)$;
        \item There exists $P^*\in\mathcal{P}(\tilde{\mathcal{A}}^{\infty}(x_0,s_0))$ such that
 $P^*$ solves  problem \eqref{equ:simplified_CK_lot}, and
     $$\mathbb{E}_{s_0}^{(a_t)\sim P^*}\sum_{t=0}^{\infty}\beta^t g(x_t,a_t,s_t)=\phi.$$
    \end{enumerate}
\end{lemma}

The lemma directly implies the following corollary, which provides conditions to ensure that $ W(\gamma,x,s,\phi) =  D(\gamma,x,s) $.
\begin{corollary}\label{cor:equiv_subgradient_lot}
For any $\gamma\in\mathbb{R}_+^I,\,x\in\mathcal{X},\,s\in\mathcal{S}$, 
    $W(\gamma,x,s,\phi)=D(\gamma,x,s)$ if and only if there exists $\phi_D\in \partial D(\gamma,x,s)$, s.t. $\phi\le \phi_D$. 

\end{corollary}

With this, 
    following the proof of Theorem \ref{thm:recursive1}, it is straightforward to show the following.
    For any $x\in \mathcal{X},\,s\in \mathcal{S},\,\gamma\in\mathbb{R}_+^{I},\phi\in\mathbb{R}^I$, the dual value function $W(\gamma,x,s,\phi)$  satisfies the following equation
    {\small \begin{equation}\label{equ:recursive_dual_promised}
    \begin{aligned}
W(\gamma,x,s,\phi)=&\inf_{\mu\in\mathbb{R}_+^I,\,\lambda\in \mathbb{R}_+^I}\sup_{a\in\tilde{\mathcal{A}}(x,s)}\Bigg[\left(r(x,a,s)+\sum_{i=1}^{I}\left( \gamma^ig^i(x,a,s)+\lambda^i(g^i(x,a,s)-\bar{g}^i)\right.\right.\\&\left.\left.+\mu^i(g^i(x,a,s)-\phi^i)\right)\right)+\beta\mathbb{E}_{s}D(\gamma+\lambda+\mu,x',s')\Bigg],\\
        \text{where }&x'=\zeta(x,a,s).
    \end{aligned}
    \end{equation}}
Moreover, Corollary \ref{cor:equiv_subgradient_lot} implies that 
{\small \begin{equation}\label{equ:recursive_dual_promisedW}
    \begin{aligned}
&W(\gamma,x,s,\phi) =\inf_{\mu\in\mathbb{R}_+^I,\,\lambda\in \mathbb{R}_+^I}\sup_{\psi\in\mathcal{P}(\tilde{\mathcal{A}}(x,s)),\phi'(a,s')\in\mathbb{R}^I}\sum_{a\in\mathcal{A}}\psi(a)\\
        &\Bigg[\left(r(x,a,s)+\sum_{i=1}^{I} \left( \gamma^ig^i(x,a,s)+\lambda^i(g^i(x,a,s)-\bar{g}^i)+\mu^i(g^i(x,a,s)-\phi^i)\right)\right)\\
&\qquad+\beta\mathbb{E}_{s}W(\gamma+\lambda+\mu,x',s',\phi'(a,s'))\Bigg],\\
        \text{where }&x'=\zeta(x,a,s).
    \end{aligned}
    \end{equation}}

    It is easy to verify
    that if $\lambda^*$ is a solution to \eqref{equ:recursive_dual}, then $\lambda^*$ together with $ \mu^*=0 $ is a solution to \eqref{equ:recursive_dual_promised} and a solution to \eqref{equ:recursive_dual_promisedW} when $\phi\le \phi_D$ for some $\phi_D\in\partial D(\gamma,x,s)$.
The following two theorems now describe how an optimal policy can be recovered.
\begin{theorem}\label{thm:recover_policy_main}
     For any $x\in \mathcal{X},\,s\in \mathcal{S},\,\gamma\in\mathbb{R}_+^{I},\phi\in\partial D(\gamma,x,s)$, the recursive equation  \eqref{equ:recursive_dual_promisedW} admits at least one solution $(\mu^*,\lambda^*,\psi^*,\phi'^{*})$, which satisfies
     \begin{equation}\label{equ:recover_policy_cond0}
         \phi'^*(a,s')\in\partial D(\gamma+\mu^*+\lambda^*,\zeta(x,a,s),s'),\quad\forall a\in\text{supp}(\psi^*),s'\in\mathcal{S},
     \end{equation}
     
      \begin{equation}\label{equ:recover_policy_cond1}(\sum_{a\in \mathcal{A}}\psi(a)(g(x,a,s)+\beta\mathbb{E}_s\phi'(a,s'))-\phi\ge 0)\perp (\gamma+\mu^*)
    \end{equation}
    and
   \begin{equation}\label{equ:recover_policy_cond2}
    (\sum_{a\in\mathcal{A}}\psi(a)(g(x,a,s)+\beta\mathbb{E}_s\phi'(a,s'))-\bar{g})\ge 0)\perp \lambda^*.\end{equation}

 Moreover, we recover the policy as follows.
   Given $x_0$, $s_0$, $\gamma^*(h^{-1}):=\gamma_0$, and $\phi^*(h^0):=\phi_0\in\partial D(\gamma_0,x_0,s_0)$, we define recursively for each $t\ge 0$ and $h^t\in\mathcal{H}^t$:
\[
(\mu^*(h^t),\lambda^*(h^t),\psi^*(h^t),\phi^{*}(h^{t+1}))
\]
to be a solution to \eqref{equ:recursive_dual_promisedW} with $x=x(h^t)$, $s=s_t$, $\gamma=\gamma^*(h^{t-1})$, $\phi=\phi^*(h^t)$ satisfying \eqref{equ:recover_policy_cond0}, \eqref{equ:recover_policy_cond1} and \eqref{equ:recover_policy_cond2}, and set $\gamma^*(h^t)=\gamma^*(h^{t-1})+\lambda^*(h^t)+\mu^*(h^t)$.
   Then $(\psi^*(h^t))$ is a solution to \eqref{equ:simplified_CK_lot_statewise}.
\end{theorem}

Now the problem reduces to finding a solution to \eqref{equ:recursive_dual_promisedW} satisfying \eqref{equ:recover_policy_cond0}, \eqref{equ:recover_policy_cond1} and \eqref{equ:recover_policy_cond2}. 
The following theorem allows us to identify such solutions.
\begin{theorem}\label{thm:policy_recover_adapt}
    For every $x\in\mathcal{X},\,s\in\mathcal{S},\,\gamma\in\mathbb{R}_+^{I},\,\phi\in\partial D(\gamma,x,s)$, there exists a solution $(\mu^*,\lambda^*,\psi^*,\phi'^*)$ to the following problem
{\footnotesize \begin{equation}\label{equ:recursive_dual_policy_adapt_promised}
    \begin{aligned}
        &W(0,x,s,\phi)=\inf_{\mu\in\mathbb{R}_+^I,\,\lambda\in \mathbb{R}_+^I}\ \sup_{\psi\in\mathcal{P}(\tilde{\mathcal{A}}(x,s)),\phi'(a,s')\in\mathbb{R}^I}\sum_{a}\psi(a)\\
    &\Bigg[\left(r(x,a,s)+\sum_{i=1}^I\left(\lambda^i(g^i(x,a,s)-\bar{g}^i)+\mu^i(g^i(x,a,s)-\phi^i)\right)\right)+\beta\mathbb{E}_{s}W(\lambda+\mu,x',s',\phi'(a,s'))\Bigg],\\
        &\text{where }x'=\zeta(x,a,s),
    \end{aligned}
    \end{equation}}
    that satisfies
    \begin{equation}\label{equ:adaptcond0}
        \phi'^*(a,s')\in\partial D(\lambda^*+\mu^*,x',s'),
    \end{equation}
    \begin{equation}\label{equ:adaptcond1}
    (\sum_{a\in \mathcal{A}}\psi^*(a)(g(x,a,s)+\beta\mathbb{E}_s\phi'^*(a,s'))-\phi\ge0)\perp\mu^*,
    \end{equation}
    and
    \begin{equation}\label{equ:adaptcond2}
    (\sum_{a\in\mathcal{A}}\psi^*(a)(g(x,a,s)+\beta\mathbb{E}_s\phi'^*(a,s')-\bar{g})\ge 0)\perp \lambda^*.
    \end{equation}
    Moreover, for $(\psi^*,\phi^{\prime*})$, there exists $(\tilde{\lambda}^*,\tilde{\mu}^*)$ such that
    $
(\tilde{\mu}^*,\tilde{\lambda}^*,\psi^*,\phi'^*)
    $
    is a solution to \eqref{equ:recursive_dual_promisedW} that satisfies \eqref{equ:recover_policy_cond0},\eqref{equ:recover_policy_cond1} and \eqref{equ:recover_policy_cond2}.
\end{theorem}

As explained above, Corollary \ref{cor:equiv_subgradient_lot} allows us in our algorithm to replace $ W $ on the right-hand side of \eqref{equ:recursive_dual_policy_adapt_promised} by $D(\cdot)$.
\subsection{An algorithm to compute the optimal policy}
To find a solution to the problem in Theorem \ref{thm:policy_recover_adapt},
we introduce a simple algorithm that computes the optimal policy along a simulated path. We assume that the optimal dual value function $D(\cdot)$ is given.
For $\gamma=0$, $x\in\mathcal{X}$, $s\in\mathcal{S}$, $\phi$, we define the inner function in \eqref{equ:recursive_dual_promised} as
{\small $$
\begin{aligned}
&F(\mu,\lambda):=\sup_{a\in\tilde{\mathcal{A}}(x,s)}\Bigg[\left(r(x,a,s)+\sum_{i=1}^{I}\left( \lambda^i(g^i(x,a,s)-\bar{g}^i)+\mu^i(g^i(x,a,s)-\phi^i)\right)\right)\\
&\qquad+\beta\mathbb{E}_{s}D(\lambda+\mu,x',s')\Bigg].
\end{aligned}
$$}
It is easy to check that $F(\mu,\lambda)$ is a convex function and given $\mu,\lambda$, any
{\small $$
\begin{aligned}
&\left((g(x,a^*,s)-\phi)+\beta\mathbb E_s \partial D(\lambda+\mu,x',s'),(g(x,a^*,s)-\bar{g})+\beta\mathbb E_s  \tilde{\phi}_{s'}\right),\\
& \tilde{\phi}_{s'} \in  \partial D(\lambda+\mu,x',s') ,\forall s'
\end{aligned}
$$}is a subgradient of $F$, where $a^*\in \arg \max F$, and $x'=\zeta(x,a^*,s)$. In practical computations, the convex function $D$ is represented by piecewise-linear interpolation on a uniform grid, which preserves convexity. To compute a subgradient of $D$ at any point $ \lambda+\mu$, it suffices to locate the simplex containing $\lambda+\mu$ in the triangulation of the grid and take the gradient of the linear piece on that simplex (if $\lambda+\mu$ lies on a boundary, any adjacent simplex yields a valid subgradient). In the one-dimensional case $I=1$, this approach reduces to the standard forward difference
$$
\tilde{\phi}_{s'} = \frac{D(\lambda+\mu+\epsilon)-D(\lambda+\mu)}{\epsilon}.
 $$

This observation directly leads to the subgradient descent method in Algorithm \ref{algorithm}. For a given value of the (extended) state $ \phi,x,s$, the method computes an optimal lottery over current actions as well as next period's utility promises $ \phi(s')$ for all $ s'\in {\mathcal S} $. Along a simulated path the algorithm can then be used to compute an optimal policy: starting from a non-binding initial promised value (for example, a promise of zero), we apply Algorithm~\ref{algorithm} iteratively at each encountered state to compute the optimal lottery and next period's promised values, which in turn become the input for the subsequent step. This process generates a sequence of optimal lotteries  over actions and promises.

\begin{algorithm}[!htb]
\caption{Algorithm for the Computation of Optimal Policies}
\label{algorithm}
Given $\phi,x,s$, initial Lagrange multipliers $\lambda^1,\mu^1$, the learning rates $\{\sigma^k\}_{k=0}^{\infty}$ and the number of iterations $N>0$. 

 \textbf{For $k=1:N$}
\begin{enumerate}
\item Solve the optimization problem
{\small $$
\begin{aligned}
&a^k\in\arg\max_{a\in\tilde{\mathcal{A}}(x,s)}\Big[\left(r(x,a,s)+\sum_{i=1}^{I}\lambda^{k,i}(g^i(x,a,s)-\bar{g}^i)+\mu^{k,i}(g^i(x,a,s)-\phi^i)\right)\\
&\qquad+\beta\mathbb{E}_{s}D(\lambda^k+\mu^k,x',s')\Big],
\end{aligned}
$$}
and compute 
$$
\phi^k(s')\in \partial D(\lambda^k+\mu^k,x',s'),\,\text{for all  }s'\in\mathcal{S}.
$$
\item Update the Lagrange multipliers by
$$
\mu^{k+1}=\max\{\mu^k- \sigma^k\left(g(x,a^k,s)-\phi+\beta\mathbb E_s\phi^k(s')\right),0\}
$$
$$
\lambda^{k+1}=\max\{\lambda^k- \sigma^k\left(g(x,a^k,s)-\bar{g}+\beta\mathbb E_s \phi^k(s')\right),0\}
$$
\textbf{End}
\end{enumerate}
\begin{enumerate}
\item Construct the stage lottery $\psi$ as
$$
\psi^N(a)=\frac{\sum_{k=1}^{N}1_{a^k=a}\sigma^k}{\sum_{k=1}^N\sigma^k}.
$$

\item Construct the promised value $\phi'$ as
$$
\phi'^N(a,s')=\frac{\sum_{k=1}^{N}1_{a^k=a}\sigma^k\phi^k(s')}{\sum_{k=1}^{N}1_{a^k=a}\sigma^k}.
$$
\end{enumerate}

\end{algorithm}

The final theorem of this section formalizes the convergence of this algorithm.
\begin{theorem}
\label{thm:algocon}
    Assume that $\phi\in\partial D(\gamma,x,s)$ for some $\gamma\in\mathbb{R}_+^I$, and that the learning rate $\{\sigma^k\}_{k=0}^{\infty}$ satisfies
    $$
    \sum_{k=0}^{\infty}\sigma^k=\infty,\,\text{and }\sum_{k=0}^{\infty}(\sigma^k)^2<\infty.
    $$
    
    Then, in Algorithm~\ref{algorithm}, $(\lambda^k,\mu^k)\rightarrow(\lambda^*,\mu^*)$, for some $(\lambda^*,\mu^*)\in \arg \min F(\lambda,\mu)$\footnote{The existence of the minimizer of $F$ is guaranteed by Theorem \ref{thm:policy_recover_adapt}.}. Furthermore, for any $\epsilon>0$, there exists $\bar{N}>0$, such that when $N>\bar{N}$, we have
    \begin{enumerate}
        \item $\|(\lambda^k,\mu^k)-(\lambda^*,\mu^*)\|\le  \epsilon$,
        \item $\psi^N(a)\mathrm{dist}(\phi'^N(a,s'),\partial D(\lambda^*+\mu^*,x',s'))<\epsilon$, for any $a\in\tilde{\mathcal{A}}(x,s)$
        \item $\sum_{a}\psi^N(a)(g(x,a,s)+\beta\mathbb{E}_s\phi'^N(a,s'))-\phi\ge -\epsilon$,
        \item $\sum_{a}\psi^N(a)(g(x,a,s)+\beta\mathbb{E}_s\phi'^N(a,s')-\bar{g})\ge -\epsilon$,
        \item $|\langle\mu^*,\sum_{a}\psi^N(a)(g(x,a,s)+\beta\mathbb{E}_s\phi'^N(a,s'))-\phi\rangle|\le \epsilon$,
        \item $|\langle\lambda^*,\sum_{a}\psi^N(a)(g(x,a,s)+\beta\mathbb{E}_s\phi'^N(a,s')-\bar{g})\rangle|\le \epsilon$.

    \end{enumerate}
\end{theorem}

We illustrate the algorithm with a simple example in Online Appendix~\ref{App:Num_ex}. For the sup-inf FE, the policy can be recovered using the same algorithm. By the same arguments as above, we can formulate a functional equation that depends on $ \gamma,x,s$ and, in addition, on promised utility $ \phi $. We obtain the following equation, which is identical to \eqref{equ:recursive_dual_promised} except that $\inf_{\lambda}\sup_{a}$ is replaced by $\sup_{a}\inf_{\lambda}$.
\begin{equation}\label{equ:recursive_dual_extend_promised}
    \begin{aligned}
        &W(\gamma,x,s,\phi)
        =\inf_{\mu\in\mathbb{R}_+^I}\sup_{a\in\tilde{\mathcal{A}}(x,s)}\inf_{\lambda\in\mathbb{R}_+^I}\Bigg[\Bigg(r(x,a,s)\\
&+\sum_{i=1}^{I}\left(\gamma^ig^i(x,a,s)+\lambda^i(g^i(x,a,s)-\bar{g}^i)+\mu^i(g^i(x,a,s)-\phi^i)\right)\Bigg) \\
        &  +\beta\mathbb{E}_{s}\tilde{D}(\gamma+\lambda+\mu,x',s')\Bigg],\quad\text{where }x'=\zeta(x,a,s).
    \end{aligned}
    \end{equation}
Here $\tilde{D}$ is the largest solution of the sup-inf functional equation \eqref{equ:recursive_dual_extend}.
In the algorithm, we can change $\sup_{a}\inf_{\lambda}$ to $\inf_{\lambda_a} \sup_a$ (by Lemma~\ref{lem:tech_infsup}), which brings us back to the same setting.

\section{Conclusion}
We show how the recursive multiplier approach to optimization problems with forward-looking constraints can be used in a model with non-convexities. In general, neither the inf-sup nor the sup-inf formulation of the functional equation gives the correct solution for a model in which lotteries are not allowed. We argue that for many applications the sup-inf approach gives the value of an economically meaningful lottery solution.

How to make our approach numerically viable in large-scale models is a subject for further research.

\singlespacing
 \bibliographystyle{ecta}
\bibliography{refe}

\renewcommand{\appendixpagename}{Appendix}
\begin{appendices}

\section{\texorpdfstring{Appendix: Proofs}{Proofs}} \label{app:sec2}

\subsubsection*{Proof of Theorem \ref{thm:exist_opt_lot}}
We consider a sequence of $P_n\in\mathcal{P}(\tilde{\mathcal{A}}^{\infty}(x_0,s_0))$ that satisfies all constraints in \eqref{equ:simplified_CK_lot} and has the property that
$$
\mathbb{E}_{s_0}^{(a(s^t))\sim P_n}\sum_{t=0}^{\infty}\beta^t\left(r(x(s^t,a^{t-1}),a(s^t),s_t)+\sum_{i=1}^{I}\gamma^ig^i(x(s^t,a^{t-1}),a(s^t),s_t)\right)
$$
tends to the supremum of problem \eqref{equ:simplified_CK_lot}. 
 According to the weak-* compactness of $\mathcal{P}(\tilde{\mathcal{A}}^{\infty}(x_0,s_0))$, $P_n$ weak-* converges to some $P^*\in\mathcal{P}(\tilde{\mathcal{A}}^{\infty}(x_0,s_0))$ up to a subsequence. 
 For any $a\in\tilde{\mathcal{A}}^{\infty}(x_0,s_0)$, we define
$$
f(a)=\mathbb{E}_{s_0}\sum_{t=0}^{\infty}\beta^t\left(r(x(s^t,a^{t-1}),a(s^t),s_t)+\sum_{i=1}^{I}\gamma^ig^i(x(s^t,a^{t-1}),a(s^t),s_t)\right),
$$
and
$$
g_{t,i,h^t}(a)=\mathbb{E}_{s^t}1_{\{\tilde{a}^{t-1}=a^{t-1}\}}\left(\sum_{n=0}^{\infty}\beta^n g^i(x(s^{t+n},(\tilde{a}^{t-1},a_t^{t+n-1})),a(s^{t+n}),s_{t+n})-\bar{g}^i\right).
$$
 It is straightforward to see that $f$ and $g_{t,i,h^t}$ are bounded, and continuous in $a$ under the product topology. Hence $\mathbb{E}^{a\sim P_n}f(a)\rightarrow\mathbb{E}^{a\sim P^*}f(a)$ and $\mathbb{E}^{a\sim P_n}g_{t,i,h^t}(a)\rightarrow\mathbb{E}^{a\sim P^*}g_{t,i,h^t}(a)$, which implies that $P^*$ is feasible for problem \eqref{equ:simplified_CK_lot} and attains the supremum. \hfill $ \Box$

\subsubsection*{Proof of Theorem \ref{thm:lot_equiv}}
    For any $P\in\mathcal{P}(\tilde{\mathcal{A}}^{\infty}(x_0,s_0))$ that is feasible for problem \eqref{equ:simplified_CK_lot}, it is straightforward to compute $\psi(h^t)(a)$ as\footnote{Note that when $P(a(s^{t-1})=\tilde{a}_{t-1},\cdots ,a(s_0)=\tilde{a}_0)=0$ we can arbitrarily define $\psi(h^t)$.}
    $$
    \psi(h^t)(a)=\frac{P(a(s^t)=a,a(s^{t-1})=\tilde{a}_{t-1},\cdots ,a(s_0)=\tilde{a}_0)}{P(a(s^{t-1})=\tilde{a}_{t-1},\cdots ,a(s_0)=\tilde{a}_0)}.
    $$
    It is straightforward to show that $(\psi(h^t))$ is feasible for \eqref{equ:simplified_CK_lot_statewise} and attains the same objective function value as $P$ in problem \eqref{equ:simplified_CK_lot}.

     For $(\psi(h^t))\in\Pi_{h^t\in\mathcal{H}^t}{\cal P}(\tilde{A}(h^t))$, fix a finite horizon $T$ and let 
$\mathcal{V}_T=\bigcup_{t=0}^T \mathcal{S}^{t+1}$ denote the set of all shock-history nodes up to depth $T$. 
For any finite subset $F\subseteq \mathcal{V}_T$ and any fixed actions $(a_v)_{v\in F}$, the finite-dimensional marginal distribution is defined by
\[
\mu_F\bigl((a_v)_{v\in F}\bigr)
=
\sum_{\substack{a_v\in\mathcal{A}\\ v\in\mathcal{V}_T\setminus F}}
\prod_{v\in\mathcal{V}_T} \psi(h^v)(a_v),
\]
where $h^v$ is the history at node $v$, recursively given by $h^{s^0}=(s_0)$ and 
$h^{s^{t+1}}=(h^{s^t},a_{s^t},s_{t+1})$ for $v=s^{t+1}$. 
    The Kolmogorov extension theorem implies that there exists a $P\in\mathcal{P}(\tilde{\mathcal{A}}^{\infty}(x_0,s_0))$ that has marginal distributions $\mu_F$, and hence $P$ attains the same objective function value as $(\psi(h^t))$ in problem \eqref{equ:simplified_CK_lot_statewise}. \hfill $ \Box $

\subsubsection*{Proof of Theorem \ref{thm:bicon}}
\begin{itemize}
    \item \textbf{Step 1.} For step 1, we need to show that for any $x\in X$, we have
    \begin{equation}\label{equ:halfspace}
    f^{**}(x)=\sup_{(x^*,\beta)\in X^*\times \mathbb{R}} \langle x^*,x\rangle -\beta
    \text{ s.t. } \langle x^*,y\rangle -\beta\le f(y),\quad \forall y\in X.
    \end{equation}
    First, by the definition of biconjugate, for any $x\in X$, we have
    \begin{equation}\label{equ:bicon_epirepresent}
    f^{**}(x)=\sup_{x^*\in X^*}\langle x^*,x\rangle-f^*(x^*)=\sup_{(x^*,\beta)\in \text{epi}(f^*)}\langle x^*,x\rangle-\beta.
    \end{equation}
    Second, by the definitions of epigraph and conjugate, we know that
    \begin{equation}\label{equ:equiv_epi}
        \begin{aligned}
            &(x^*,\beta)\in \text{epi}(f^*)
            \Leftrightarrow f^*(x^*)\le \beta
            \Leftrightarrow\sup_{y\in X}\langle x^*,y\rangle-f(y)\le \beta\\
            \Leftrightarrow&\langle x^*,y\rangle-f(y)\le \beta,\quad\forall y\in X
            \Leftrightarrow\langle x^*,y\rangle-\beta\le f(y),\quad\forall y\in X.
        \end{aligned}
    \end{equation}
    Therefore \eqref{equ:halfspace} can be obtained by combining \eqref{equ:bicon_epirepresent} and \eqref{equ:equiv_epi}.
    \item \textbf{Step 2.} For step 2, we aim to show that
    \begin{equation}\label{equ:f**<=f}
    \text{epi}(f^{**})\subseteq \text{cl } \text{co } \text{epi}(f).
    \end{equation}
    For any $(x_0,\alpha_0)\not\in \text{cl } \text{co } \text{epi}(f)$, by Lemma \ref{lem:ascoli}, there exist $x^*\in X^*$, $\lambda\in \mathbb{R}$, and $\gamma\in\mathbb{R}$ such that
    \begin{equation}\label{equ:specific_dual}
    \langle x^*,x\rangle-\lambda \alpha<\gamma<\gamma_0=\langle x^*,x_0\rangle -\lambda \alpha_0,\quad\forall(x,\alpha)\in \text{cl } \text{co } \text{epi}(f).
    \end{equation}
    Since $f$ is proper, $\text{dom}(f)\not=\emptyset$, and hence there exists $x'\in \text{dom}(f)$. For any $\alpha'\ge f(x')$, $(x',\alpha')\in \text{epi}(f)\subseteq\text{cl } \text{co } \text{epi}(f)$. We take $x=x',\alpha=\alpha'\rightarrow+\infty$ on the left-hand side of \eqref{equ:specific_dual}, and obtain $\lambda\ge 0$. \begin{enumerate}
        \item If $\lambda >0$, \eqref{equ:specific_dual} implies
        $
            \langle \frac{x^*}{\lambda},x\rangle-\frac{\gamma}{\lambda}<\alpha,\quad \forall (x,\alpha)\in \text{cl } \text{co } \text{epi}(f).
       $
        Therefore, $(\frac{x^*}{\lambda},\frac{\gamma}{\lambda})$ satisfies the constraint in \eqref{equ:halfspace}. Furthermore, from \eqref{equ:specific_dual} we have
    $
    \langle \frac{x^*}{\lambda},x_0\rangle -\frac{\gamma}{\lambda}=\alpha_0+\frac{\gamma_0-\gamma}{\lambda}>\alpha_0,
    $
    implying that 
    $
    f^{**}(x_0)>\alpha_0,
    $
    according to \eqref{equ:halfspace}. Hence $(x_0,\alpha_0)\not\in \text{epi}(f^{**})$.
    \item If $\lambda=0$, \eqref{equ:specific_dual} implies that
    \begin{equation}\label{equ:specific_dual_0lambda}
        \langle x^*,x\rangle<\gamma<\gamma_0=\langle x^*,x_0\rangle,\quad \forall x\in \text{epi}(f).
    \end{equation}
    We take $\underline{x}^*\in X^*$, $\underline{\beta}\in \mathbb{R}$, such that \eqref{equ:regularity} holds. Then for any $K>0$, 
    $$
    \begin{aligned}
    \langle \underline{x}^*+Kx^*,x\rangle-(K\gamma-\underline{\beta})
    &= K\left(\langle x^*,x\rangle -\gamma\right)+\langle \underline{x}^*,x\rangle +\underline{\beta}\\
    &\le \langle \underline{x}^*,x\rangle +\underline{\beta}\le f(x),\quad \forall x\in X.
    \end{aligned}
    $$
    Therefore $( \underline{x}^*+Kx^*,K\gamma-\underline{\beta})$ satisfies the constraint in \eqref{equ:halfspace}. Furthermore, from \eqref{equ:specific_dual_0lambda} we have
    $$
    \begin{aligned}
    &\langle \underline{x}^*+Kx^*,x_0\rangle-(K\gamma-\underline{\beta})\\
    =&K\left(\langle x^*,x_0\rangle -\gamma\right)+\langle \underline{x}^*,x_0\rangle +\underline{\beta}
    =K(\gamma_0-\gamma)+\langle \underline{x}^*,x_0\rangle+\underline{\beta},
    \end{aligned}
    $$
    implying that
    \begin{equation}\label{equ:fstarstar_lb}
        f^{**}(x_0)\ge K(\gamma_0-\gamma)+\langle \underline{x}^*,x_0\rangle+\underline{\beta},
    \end{equation}
    according to \eqref{equ:halfspace}. We take 
    $
    K>\frac{\alpha_0-\underline{\beta}-\langle \underline{x}^*,x_0\rangle}{\gamma_0-\gamma}
    $
    in \eqref{equ:fstarstar_lb} and obtain
    $
    f^{**}(x_0)>\alpha_0,
    $ yielding that $(x_0,\alpha_0)\not\in \text{epi}(f^{**})$.
    \end{enumerate}
   Therefore we have shown that, for any $(x_0,\alpha_0)\not\in \text{cl } \text{co } \text{epi}(f),$ we have $(x_0,\alpha_0)\not\in \text{epi}(f^{**})$, which is equivalent to \eqref{equ:f**<=f}.
\item \textbf{Step 3.} For step 3, we  show that
\begin{equation}\label{equ:f<=f**}
       \text{cl } \text{co } \text{epi}(f) \subseteq \text{epi}(f^{**}).
\end{equation}
From \eqref{equ:halfspace} we know $f^{**}(x)\le f(x),\,\forall x\in X$ from the constraint. Therefore, $\text{epi}(f)\subseteq\text{epi}(f^{**}).$ It then suffices to show that $\text{epi}(f^{**})$ is a closed convex set, or equivalently, to show that $f^{**}$ is convex and lower semicontinuous, according to Proposition \ref{prop:convex} and Proposition \ref{prop:lsc}.
\begin{enumerate}
    \item \textbf{Convexity.} Let $x\in X$ and assume that $x^*\in X^*$ satisfies $\langle x^*,y\rangle -\beta\le f(y),\,\forall y\in X$. According to \eqref{equ:halfspace}, for any $x_1,\,x_2\in X$, $0\le p_1,\,p_2\le 1$ such that $p_1+p_2=1$ and $p_1x_1+p_2x_2=x$, we have
    {\footnotesize
    \begin{equation}\label{equ:proof_convex}
    \langle x^*,x\rangle -\beta
        = p_1\left(\langle x^*,y_1\rangle-\beta\right)+p_2\left(\langle x^*,y_2\rangle -\beta\right) \le  p_1f^{**}(y_1)+p_2f^{**}(y_2).
    \end{equation}
    }
    Since \eqref{equ:proof_convex} holds for any $x^*\in X^*$ satisfying $\langle x^*,y\rangle -\beta\le f(y),\,\forall y\in X$, we can deduce from \eqref{equ:proof_convex} that
    $
    f^{**}(x)\le p_1f^{**}(x_1)+p_2f^{**}(x_2).
    $
    Hence $f^{**}$ is convex. 
    \item \textbf{Lower Semicontinuity.} If $f^{**}(\bar{x})=+\infty$, then it follows directly from \eqref{equ:halfspace} that $\lim_{n\rightarrow\infty} f^{**}(x_0)=+\infty$. To show that $f^{**}$ is l.s.c., it suffices to show that for any $\epsilon>0$ and any $\bar{x}\in X$ s.t. $f^{**}(\bar{x})<+\infty$, any sequence $\{x_n\}$ s.t. $x_n\rightarrow \bar{x}$ in $X$, we have
    \begin{equation}\label{equ:proof_lsc_equiv}
        f^{**}(\bar{x})-\epsilon\le \liminf_{n\rightarrow\infty}f^{**}(x_n).
    \end{equation}
    Indeed, according to \eqref{equ:halfspace}, there exists $x^*\in X^*$ such that $\langle x^*,y\rangle -\beta\le f(y),\,\forall y\in X$, and
    \begin{equation}\label{equ:prop_lsc}
    \langle x^*,\bar{x}\rangle -\beta\ge f^{**}(\bar{x})-\frac{\epsilon}{2}.
    \end{equation}
    Therefore, for any $x\in X$ such that 
    \begin{equation}\label{equ:proof_lsc_nearxbar}
    \|x-\bar{x}\|_{X}\le \frac{\epsilon}{2\max\{\|x^*\|_{X^*},1\}},
    \end{equation}
    we can deduce from \eqref{equ:prop_lsc} that
    \begin{equation}
    \begin{aligned}\label{equ:proof_lsc_nearbarx_prop}
    & \langle x^*, x\rangle-\beta=\left(\langle x^*, \bar{x}\rangle-\beta\right)+\langle x^*, x-\bar{x}\rangle\\
    & \ge f^{**}(\bar{x})-\frac{\epsilon}{2}-\|x^*\|_{X^*}\|x- \bar{x}\|_{X}
    \ge f^{**}(\bar{x})-\frac{\epsilon}{2}-\frac{\epsilon}{2}=f^{**}(\bar{x})-\epsilon.
    \end{aligned}
    \end{equation}
    Since \eqref{equ:proof_lsc_nearbarx_prop} holds for any $x^*\in X^*$ satisfying $\langle x^*,y\rangle -\beta\le f(y),\,\forall y\in X$, we then know from \eqref{equ:halfspace} that
    $
    f^{**}(x)\ge f^{**}(\bar{x})-\epsilon,
    $
    for all $x\in X$ satisfying \eqref{equ:proof_lsc_nearxbar}. Hence, we obtain \eqref{equ:proof_lsc_equiv}.
\end{enumerate}
\end{itemize}
The proof is completed by combining \eqref{equ:f**<=f} in Step 2 and \eqref{equ:f<=f**} in Step 3. \hfill $ \Box $

\subsubsection*{Proof of  Theorem \ref{thm:dualgap}}
    \begin{itemize}
        \item We first show that $p=v(\theta_Y)$.
        
        For any $x\in \Omega$ such that $g(x)\le \theta_{Y}$, since $\langle y^*,g(x)\rangle\le 0,\,\forall y^*\in Y^*_{+}$ and $\langle \theta_{Y^*},g(x)\rangle=0$, we have
        \begin{equation}\label{equ:dualgap_primal_xfeas}
        \sup_{y^*\in Y^*_{+}}L(x,y^*)= f(x).
        \end{equation}
        
        For any $x\in \Omega$, such that $g(x)\le \theta_Y$ does not hold, i.e., $g(x)\not\in -P$, since $-P$ is a closed convex set, according to Lemma \ref{lem:ascoli}, we know that there exists $y^*\in Y^*$ and $\gamma\in \mathbb{R}$, such that
        \begin{equation}\label{equ_ascoli_proofdualgap}
        \langle y^*,y\rangle<\gamma<\langle y^*,g(x)\rangle ,\quad \forall y\in -P.
        \end{equation}
        Since $\theta_Y\in -P$ and $\langle y^*,\theta_Y\rangle=0$, we know that $\gamma>0$. On the other hand, if there exists $y\in -P$, such that $\langle y^*,y\rangle>0$, then one can take $\alpha>0$ sufficiently large, such that $\langle y^*,\alpha y\rangle=\alpha \langle y^*,y\rangle\ge \langle y^*,g(x)\rangle$, which contradicts \eqref{equ_ascoli_proofdualgap}. Therefore, $y^*\in Y^*_{+}$, and $\langle y^*,g(x)\rangle >\gamma>0$. Hence for any $\alpha>0$,
        $
        \sup_{y^*\in Y_+^{*}}L(x,y^*)\ge f(x)+\alpha \langle y^*,g(x)\rangle>f(x)+\alpha\gamma,
        $
        implying that 
         \begin{equation}\label{equ:dualgap_primal_xinfeas}
        \sup_{y^*\in Y_+^{*}}L(x,y^*)=+\infty.
        \end{equation}
        We combine \eqref{equ:dualgap_primal_xfeas} and \eqref{equ:dualgap_primal_xinfeas} to obtain
        $$
        \sup_{y^*\in Y^*_{+}}L(x,y^*)=\begin{cases}
            f(x),&g(x)\le \theta_{Y};\\
            +\infty,&g(x)\not\le \theta _Y,
        \end{cases}
        $$
        yielding that $p=\inf_{x\in \Omega}\sup_{y^*\in Y^{*}_+}L(x,y^*)=v(\theta_Y)$.
        \item For step 2, we show that $d=v^{**}(\theta_Y)$. By direct computation, we have
        \begin{equation}\label{equ:vstarstar_refor}
        \begin{aligned}
            v^{**}(\theta_Y)&=\sup_{y^*\in Y^*}-v^*(y^*)\\
            &=\sup_{y^*\in Y^*} -\left(\sup_{y\in Y} (\langle y^*,y\rangle -v(y))\right)
            =\sup_{y^*\in Y^*} \inf_{y\in Y} (v(y)-\langle y^*,y\rangle)\\
            &=\sup_{y^*\in Y^*}\inf_{y\in Y}(\inf_{x\in\Omega:g(x)\le y}f(x)-\langle y^*,y\rangle)
            =\sup_{y^*\in Y^*}\inf_{x\in \Omega}\inf_{y\ge g(x)}\left(f(x)-\langle y^*,y\rangle\right)\\
            &=\sup_{y^*\in Y^*}\inf_{x\in \Omega} \left[\left(f(x)-\langle y^*,g(x)\rangle\right) +\inf_{y\ge g(x)}\langle y^*,g(x)-y\rangle\right]\\
            &=\sup_{y^*\in Y^*}\inf_{x\in \Omega} \left[\left(f(x)+\langle y^*,g(x)\rangle\right) +\inf_{y\ge g(x)}\langle y^*,y-g(x)\rangle\right].
        \end{aligned}
        \end{equation}
        If $y^*\not\in Y^{*}_+$, then it is straightforward to verify that for any $x\in \Omega$, there exists $y\ge g(x)$, such that
        $
        \inf_{y\ge g(x)}\langle y^*,y-g(x)\rangle=-\infty,
        $
        and hence
        $$
        \inf_{x\in \Omega} \left[\left(f(x)+\langle y^*,g(x)\rangle\right) +\inf_{y\ge g(x)}\langle y^*,y-g(x)\rangle\right]=-\infty.
        $$
        If $y^*\in Y^*_+$, then for any $x\in \Omega$, $y\ge g(x)$, we have
        $
        \langle y^*,y-g(x)\rangle\ge 0,
        $
        and the equality is obtained when $y=g(x)$.
        Therefore, we conclude from \eqref{equ:vstarstar_refor} that
        $$
        v^{**}(\theta_Y)=
            \sup_{y^*\in Y^*_+}\inf_{x\in \Omega} \left(f(x)+\langle y^*,g(x)\rangle\right) =\sup_{y^*\in Y^*_+}\inf_{x\in \Omega} L(x,y^*) =d.
        $$
    \end{itemize}
    We combine the two steps to finish the proof of the first statement. 
  The fact that  $ \text{epi}(v^{**})=\text{cl } \text{co } \text{epi}(v) $ follows directly from Theorem \ref{thm:bicon}. \hfill $ \Box $

\subsubsection*{Proof of Theorem \ref{thm:lot_dual_equiv}}

     According to Theorem \ref{dual_for}, $D(\gamma,x_0,s_0)$ equals to the dual of the deterministic problem \eqref{equ:CK_equiv}. According to Corollary \ref{cor:infinite}, for any $\epsilon>0$, there exists $N\in\mathbb{N}_+$, $(a^k_t)\in \tilde{\mathcal{A}}^{\infty}(x_0,s_0)\,(k=1,\cdots,N)$, $p_k\ge 0,\,\sum_{k=1}^{N}p_k=1$, and $P_\epsilon\in\mathcal{P}(\tilde{\mathcal{A}}^{\infty}(x_0,s_0))$ defined by
    $
    P_\epsilon(\{(a^k_t)\})=p_k,
    $
    such that
    $$
    \mathbb{E}_{s_0}^{(a_t(s^t))\sim P_\epsilon}\sum_{t=0}^{\infty}\beta^t\left(r(x_t,a_t,s_t)+\sum_{i=1}^{I}\gamma^ig^i(x_t,a_t,s_t)\right)\ge D(\gamma,x_0,s_0)-\epsilon,
    $$
    and
    $$
    \begin{aligned}
       &\mathbb{E}_{s^t}^{(a(s^t))\sim P_\epsilon}1_{\{(\tilde{a}_0,\cdots,\tilde{a}_{t-1})=(a(s_0),\cdots,a(s^{t-1}))\}}\left(\sum_{n=0}^{\infty}\beta^n g^i(x_{t+n},a(s^{t+n}),s_{t+n})-\bar{g}^i\right)\ge-\epsilon,\\&
       \forall t\in\mathbb{N},\,\forall h^t=(s_0,\tilde{a}_0,\cdots,s_{t-1},\tilde{a}_{t-1},s_t)\in \mathcal{H}^t,\,\forall i\in\{1,\cdots,I\}.
       \end{aligned}
    $$
    Since $\mathcal{P}(\mathcal{\tilde{A}}^{\infty}(x_0,s_0))$ is weak-* compact, $P_{\epsilon}$ weak-* converges to some $P^*\in\mathcal{P}(\tilde{\mathcal{A}}^{\infty}(x_0,s_0))$ as $\epsilon\rightarrow0$ up to a subsequence. It is straightforward to verify that $P^*$ is a feasible probability measure for problem \eqref{equ:simplified_CK_lot} and 
    $$
    \mathbb{E}_{s_0}^{(a(s^t))\sim P^*}\sum_{t=0}^{\infty}\beta^t\left(r(x_t,a(s^t),s_t)+\sum_{i=1}^{I}\gamma^ig^i(x_t,a(s^t),s_t)\right)\ge D(\gamma,x_0,s_0).
    $$
     Therefore, $V(\gamma,x_0,s_0)\ge D(\gamma,x_0,s_0)$.
    
    Then it suffices to prove that $V(\gamma,x_0,s_0)\le D(\gamma,x_0,s_0)$. We consider the set of finite support probability measures,
    $$
    \mathcal{D}=\{\sum_{i=1}^{n}\lambda_i\delta_{a^i}|n\in\mathbb{N},\,\lambda_i\ge 0,\,\sum_{i=1}^{n}\lambda_i=1,\,a^i\in\tilde{\mathcal{A}}^{\infty}(x_0,s_0)\}.
    $$
     $\mathcal{D}$ is weak-* dense in $\mathcal{P}(\tilde{\mathcal{A}}^{\infty}(x_0,s_0))$ (see Theorem 15.10 in \cite{aliprantis2006infinite}), i.e., for any $P^*\in\mathcal{P}(\tilde{\mathcal{A}}^{\infty}(x_0,s_0))$, there exists $\{P^k_D\}_{k=1}^{\infty}\subset\mathcal{D}$, s.t.
    $$
    \int_{a\in \tilde{\mathcal{A}}^{\infty}(x_0,s_0)} f(a)dP_D^k\rightarrow\int_{a\in\tilde{\mathcal{A}}^{\infty}(x_0,s_0)}f(a)dP^*,
    $$
    for any $f\in C_b(\tilde{\mathcal{A}}^{\infty}(x_0,s_0))$. We take $P^*\in\mathcal{P}(\tilde{\mathcal{A}}^{\infty}(x_0,s_0))$ as the solution to \eqref{equ:simplified_CK_lot}, and let
    $
    P_D^k=\sum_{i=1}^{n_k}\lambda_{k,i}\delta_{a^{k,i}}.
    $
    According to part 5 of Assumption \ref{ass:dynamic}, for any $a\in\tilde{\mathcal{A}}^{\infty}(x_0,s_0)$, there exists $(\underline{a}(s^t))\in\mathcal{A}^{\infty}$, such that when defining $[a]^T\in\mathcal{A}^{\infty}$ as $$
    [a]^T=\begin{cases}
        a(s^t),&t\le T;\\
        \underline{a}(s^t),&t>T,
    \end{cases}
    $$
    $[a]^T\in\tilde{\mathcal{A}}^{\infty}(x_0,s_0)$ and satisfies the constraints in \eqref{equ:CK} for any $t>T$.
    Furthermore, we can define
    $$
    \underline{P}_D^k=\sum_{i=1}^{n_k}\lambda_{k,i}\delta_{[a^{k,i}]^k}.
    $$

    To finish the proof, we need two lemmas.    \begin{lemma}\label{lem:technique_4_thm:lot_dual_equi}
        For any $P\in\mathcal{P}(\tilde{\mathcal{A}}^{\infty}(x_0,s_0))$, we define $G(P)\in\ell^{\infty}$ as
        {\footnotesize $$
G(P)=\Bigg(\min\Bigg\{\mathbb{E}^{(a(s^t))\sim P}_{s^t}1_{\{\tilde{a}^{t-1}=a^{t-1}\}}\left(\sum_{n=0}^{\infty}\beta^n g^i(x_{t+n},a(s^{t+n}),s_{t+n})-\bar{g}^i\right) ,0\Bigg\}\Bigg)_{\substack{t\in\mathbb{N},\, h^t\in \mathcal{H}^t\\[2pt] i\in\{1,\cdots,I\}}}.
$$}
        Then $G(\underline{P}_D^k)$ converges weakly to $\theta$ in $\ell^{\infty}$.
    \end{lemma}
\begin{proof}
By definition and the Yosida--Hewitt decomposition theorem (see Lemma \ref{lem:Yosida} for details), it suffices to show that for any $\lambda^1\in \ell^1,\,\lambda^s\in\ell^s$, we have
\begin{equation}\label{appprove_aim}
\langle \lambda^1+\lambda^s,G(\underline{P}_D^k)\rangle\rightarrow0,\quad\text{ as }k\rightarrow\infty.
\end{equation}
By the definition of $\underline{P}_D^k$, it is easy to check that
$
G(\underline{P}_D^k)_{t>k,\,h^t\in\mathcal{H}^t,\,i\in\{1,\cdots,I\}}=0,
$
yielding that
\begin{equation}\label{appprove_aim1}
\langle \lambda^s,G(\underline{P}_D^k)\rangle =0,\quad\forall k\in\mathbb{N}_+.
\end{equation}
For any $t,\,h^t,\,i$, for $a\in\tilde{\mathcal{A}}^{\infty}(x_0,s_0)$, we define
$$
g_{t, h^t, i}(a):=\mathbb{E}_{s_t}1_{\{(\tilde{a}^{t-1}=a^{t-1})\}}\left(\sum_{n=0}^{\infty}\beta^ng^i(x_{t+n},a(s^{t+n}),s_{t+n})-\bar{g}^i\right).
$$
Therefore, $$\int_{a\in\tilde{\mathcal{A}}^{\infty}(x_0,s_0)}g_{t,h^t,i}(a)dP_D^k\rightarrow \int_{a\in\tilde{\mathcal{A}}^{\infty}(x_0,s_0)}g_{t,h^t,i}(a)dP^*\ge 0,\quad\text{ as }k\rightarrow\infty,\quad\forall t,\,h^t,\,i,$$ since $P_D^k$ converges weakly to $P^*$. Hence 
\begin{equation}\label{equ:appprove1}
    G(P_D^k)_{t,h^t,i}=\min\{\int_{a\in\tilde{\mathcal{A}}^{\infty}(x_0,s_0)}g_{t,h^t,i}(a)dP_D^k,0\}\rightarrow0,\quad\text{ as }k\rightarrow\infty,\quad \forall t,\,h^t,\,i.
    \end{equation}
It is straightforward to verify that\footnote{This is the Uniformly Asymptotically Non-Anticipatory (UANA) property see Appendix~\ref{lagrange}).}
\begin{equation}\label{equ:appprove2}
G(\underline{P}_D^k)_{t,h^t,i}-G(P_D^k)_{t,h^t,i}\rightarrow0,\quad\text{ as }k\rightarrow\infty,\quad \forall t,\,h^t,\,i.
\end{equation}
We combine \eqref{equ:appprove1} and \eqref{equ:appprove2} to obtain
$
G(\underline{P}_D^k)_{t,h^t,i}\rightarrow0.
$
This pointwise convergence property, together with the fact that $G(\underline{P}_D^k)$ is uniformly bounded in $\ell^{\infty}$, yields that $G(\underline{P}_D^k)$ weak-* converges to $\theta$ in $\ell^{\infty}$, implying that
\begin{equation}\label{appprove_aim2}
\langle \lambda^1, G(\underline{P}_D^k)\rangle \rightarrow0,\text{ as }k\rightarrow\infty.
\end{equation}
We then conclude \eqref{appprove_aim} by combining \eqref{appprove_aim1} and \eqref{appprove_aim2}. 
\end{proof}

\begin{lemma}\label{lem:technique_4_thm2:lot_dual_equi}
        For any $P\in\mathcal{P}(\tilde{\mathcal{A}}^{\infty}(x_0,s_0))$, we define 
        $$
        F(P)=\mathbb{E}_{s_0}^{(a(s^t))\sim P}\sum_{t=0}^{\infty}\beta^t\left(r(x_t,a(s^t),s_t)+\sum_{i=1}^{I}\gamma^ig^i(x_t,a(s^t),s_t)\right).
        $$
        Then $F(\underline{P}_D^k)$ converges to $F(P^*)=V(\gamma,x_0,s_0)$.
    \end{lemma}
\begin{proof}
For $a\in\tilde{\mathcal{A}}^{\infty}(x_0,s_0)$, we define 
$$
f(a):=\mathbb{E}_{s_0}\sum_{t=0}^{\infty}\beta^t\left(r(x_t,a(s^t),s_t)+\sum_{i=1}^{I}\gamma^ig^i(x_t,a(s^t),s_t)\right).
$$
Therefore,
$$
F(P_D^k)=\int_{a\in \tilde{\mathcal{A}}^{\infty}(x_0,s_0)}f(a)dP_D^k\rightarrow\int_{a\in\tilde{\mathcal{A}}^{\infty}(x_0,s_0)}f(a)dP^*=F(P^*),\quad\text{ as }k\rightarrow\infty.
$$
As in the proof of Lemma \ref{lem:technique_4_thm:lot_dual_equi}, it is straightforward to verify that
$
F(\underline{P}_D^k)-F(P_D^k)\rightarrow0,\,\text{ as } k\rightarrow \infty,
$
and we omit details here.
    \end{proof}
    
    Let $v$ be the perturbation functional of the deterministic problem \eqref{equ:CK}.\footnote{See \eqref{equ:def_pert} for the definition of the perturbation functional of a minimization problem. Since the problem here is a maximization problem, unlike the minimization problems considered in the standard optimization literature, the epigraph is replaced by the hypograph, which is defined analogously.} It can be verified from the definition that
    $
(G(\underline{P}_D^k),F(\underline{P}_D^k))\in \text{co }\text{hypo }(v).
    $
    By Mazur's theorem (see Corollary 3.8 in \cite{brezis2011functional}), the weak limit of $(G(\underline{P}_D^k),F(\underline{P}_D^k))$ belongs to $\text{cl }\text{co }\text{hypo }(v)$, i.e.,
    $
    (\theta, V(\gamma,x_0,s_0))\in \text{cl }\text{co }\text{hypo }(v)=\text{hypo }(v^{**}),
    $
    where the last equality holds according to Theorem \ref{thm:bicon}. From Theorem \ref{thm:dualgap}, we have
    $
    D(\gamma,x_0,s_0)=v^{**}(\theta).
    $
    Therefore $V(\gamma,x_0,s_0)\le D(\gamma,x_0,s_0)$ by definition of hypograph.
\hfill $ \Box $

\subsubsection*{Proof of Theorem \ref{thm:lot_dual_equiv-expost}}
We consider the problem
{\footnotesize \begin{equation}\label{equ:CK_equiv_sup_inf}
    \begin{aligned}
      \sup_{(a(s^t))\in \tilde{\mathcal{A}}^{\infty}(x_0,s_0)\subset \ell^{\infty}} &\mathbb{E}_{s_0}\sum_{t=0}^{\infty}\beta^t\left(r(x(s^t),a(s^t),s_t)
+\sum_{i=1}^{I}\gamma^ig^i(x(s^t),a(s^t),s_t)\right),\\
      \textbf{s.t. }&1_{\{(\tilde{a}_0,\cdots,\tilde{a}_t)=(a(s_0),\cdots,a(s^t))\}}\left(\mathbb{E}_{s_t}\sum_{n=0}^{\infty}\beta^n g^i(x(s^{t+n}),a(s^{t+n}),s_{t+n})-\bar{g}^i\right)\ge 0,\\
      &\forall t\in\mathbb{N},\,\forall \tilde{h}^t=(s_0,\tilde{a}_0,\cdots,s_t,\tilde{a}_t)\in \tilde{\mathcal{H}}^t,\,\forall i\in\{1,\cdots,I\}.
    \end{aligned}
\end{equation}} Analogously to Theorem \ref{dual_for}, it is easy to show that $\tilde{D}(\gamma,x_0,s_0)$ equals the Lagrangian dual value of problem \eqref{equ:CK_equiv_sup_inf}. Then, employing a similar line of argument as in Theorem \ref{thm:lot_dual_equiv}, we arrive at the conclusion of this theorem. \hfill $\Box$

\subsubsection*{Proof of Lemma \ref{lem:tech_infsup}}
    It suffices to show that
    \begin{equation}\label{lem:tech_infsup_aim}
    \inf_{\boldsymbol{y}\in \Gamma}\sup_{x\in \tilde{X}}f(x,\boldsymbol{y}(x))\le \sup_{x\in \tilde{X}}\inf_{\boldsymbol{y}\in \Gamma}f(x,\boldsymbol{y}(x)).
    \end{equation}
    For any $\epsilon>0$ and $x\in \tilde{X}$, there exists $y_{\epsilon}(x)\in Y$, such that
    $
    f(x,y_{\epsilon}(x))\le \inf_{y\in Y}f(x,y)+\epsilon.
    $
    Therefore, taking $\boldsymbol{y}\in \Gamma$ s.t. $\boldsymbol{y}(x)=y_\epsilon(x)$ for any $x\in\tilde{X}$, we have
    $$
    \sup_{x\in \tilde{X}}f(x,y_{\epsilon}(x))\le \sup_{x\in \tilde{X}}\inf_{y\in Y}f(x,y)+\epsilon
    = \sup_{x\in \tilde{X}}\inf_{\boldsymbol{y}\in \Gamma}f(x,\boldsymbol{y}(x))+\epsilon,
$$
implying that
$$
\inf_{\boldsymbol{y}\in \Gamma} \sup_{x\in \tilde{X}}f(x,\boldsymbol{y}(x))\le \sup_{x\in \tilde{X}}\inf_{\boldsymbol{y}\in \Gamma}f(x,\boldsymbol{y}(x))+\epsilon.
$$
Since $\epsilon$ is chosen arbitrarily, this implies \eqref{lem:tech_infsup_aim}. \hfill $ \Box $

\subsubsection*{Proof of Theorem \ref{thm:recursive1}}
    Using  \eqref{equ:dual_value_def}, we obtain the following
\begin{equation}\label{equ:proof_recursive_mid1}
\footnotesize
    \begin{aligned}
       &D(\gamma,x_0,s_0)=\inf_{(\lambda^i(h^t))\in \Lambda} \sup_{(a(s^t))\in \tilde{\mathcal{A}}^{\infty}(x_0,s_0)} L((a(s^t)),(\lambda^i(h^t));(\gamma^i),x_0,s_0)\\
        =&\inf_{(\lambda^i(h^t))\in \Lambda} \sup_{(a(s^t))\in \tilde{\mathcal{A}}^{\infty}(x_0,s_0)} \mathbb{E}_{s_0}\sum_{t=0}^{\infty}\beta^t\left\{r(x(s^t),a(s^t),s_t)+\sum_{i=1}^{I}\gamma^ig^i(x(s^t),a(s^t),s_t)\right.\\
&\qquad\left.+\sum_{i=1}^I \left[\lambda^i(h^t)\mathbb{E}_{s_t}\left(\sum_{n=0}^{\infty}\beta^n g^i(x(s^{t+n}),a(s^{t+n}),s_{t+n})-\bar{g}^i\right)\right]\right\}\\
        =&\inf_{(\lambda^i_0)\ge 0}\inf_{(\lambda^i(h^t))_{t\ge 1}} \sup_{a_0}\sup_{(a(s^t))_{t\ge 1}}\underbrace{\left(r(a_0,x_0,s_0)+\sum_{i=1}^I \left(\gamma^ig^i(a_0,x_0,s_0)+\lambda_0^i(g^i(a_0,x_0,s_0)-\bar{g}^i)\right)\right)}_{h_1((\lambda_0^i),a_0)}\\
&+\underbrace{\mathbb{E}_{s_0}\sum_{t=1}^{\infty}\beta^t\left\{r(a_t,x_t,s_t)+\sum_{i=1}^{I}\left[(\gamma^i+\lambda_0^i)g^i(a_t,x_t,s_t)+\lambda^i(h^t)\mathbb{E}_{s_t}\left(\sum_{n=0}^{\infty}\beta^n g^i(a_{t+n},x_{t+n},s_{t+n})-\bar{g}^i\right)\right]\right\}
}_{h_2((\lambda_0^i),a_0,(a(h^t))_{t\ge 1},(\lambda^i(h^t))_{t\ge 1})},
    \end{aligned}
    \end{equation}
where $a_0:=a_0(s_0)$, and the third equation separates the terms concerning $t=0$ and $t\ge 1$\footnote{In the expression $\sup_{a_0}\sup_{(a(s^t))_{t\ge 1}}$, we omit the feasible sets for $a_0$ and $(a(s^t))_{t\ge 1}$. Specifically:  
\begin{itemize}  
   \item $a(s_0) \in\tilde{\mathcal{A}}(x_0,s_0)$,
   \item 
    $(a(s^{t};s_0,s_1))_{t\ge 1}\in\tilde{\mathcal{A}}^{\infty}(x_1,s_1)$, where $x_1=\zeta(x_0, a_0, s_0)$, for all $s_0\in\mathcal{S}$.
   Here $a(s^t; s_0, s_1)$ denotes the continuation policy given the initial shocks $s_0$ and $s_1$.
\end{itemize}  
}. We know that {\footnotesize \[\begin{aligned}
&\mathbb{E}_{s_0}\left[\lambda^i(h^t)\mathbb{E}_{s_t}\left(\sum_{n=0}^{\infty}\beta^n g^i(x(s^{t+n}),a(s^{t+n}),s_{t+n})-\bar{g}^i\right)\right]\\=&\pi(s^t|s_0)\lambda^i(h^t)1_{\{(\tilde{a}_0,\cdots,\tilde{a}_{t-1})=(a(s_0),\cdots,a(s^{t-1}))\}} \left(\mathbb{E}_{s_t}\sum_{n=0}^{\infty}\beta^n g^i(x(s^{t+n}),a(s^{t+n}),s_{t+n})-\bar{g}^i\right)>0\end{aligned}\]}only if $h^t=(s_0,a_0,h^{t-1})$. Therefore, $h_2((\lambda_0^i),a_0,(a(s^t))_{t\ge 1},(\lambda^i(h^t))_{t\ge 1})$ can be reduced to $h_2((\lambda_0^i),a_0,(a(s^t))_{t\ge 1},(\lambda^i(s_0,a_0,h^{t-1}))_{t\ge 1})$. We define $X=\mathcal{A}$, $\tilde{X}=\mathcal{A}(x_0,s_0)$ and $Y=\Lambda^{|\mathcal{S}|}$. Since $(\lambda^i(s_0,a_0,h^{t-1}))_{t\ge 1}$ can be viewed  as a map from $ X $ to $Y$ for each $s_0\in\mathcal{S}$, and 
$$
\sup_{(a(s^t))_{t\ge 1}}h_1((\lambda_0^i),a_0)+h_2((\lambda_0^i),a_0,(a(s^t))_{t\ge 1},(\lambda^i(s_0,a_0,h^{t-1}))_{t\ge 1})
$$
can be regarded as a functional $f(a_0,(\lambda^i(s_0,a_0,h^{t-1}))_{t\ge 1};(\lambda_0^i))$.
By Lemma \ref{lem:tech_infsup}, we can exchange the operator $\inf_{(\lambda(h^t))_{t\ge 1}}$ and $\sup_{a_0}$ to obtain
 {\footnotesize \begin{equation}\label{equ:proof_recursive_mid2}
    \begin{aligned}
        D(\gamma,x_0,s_0)=&\inf_{(\lambda_0^i)}\inf_{(\lambda^i(h^t))_{t\ge 1}} \sup_{a_0}\sup_{(a(s^t))_{t\ge 1}}h_1((\lambda_0^i),a_0)+h_2((\lambda_0^i),a_0,(a(s^t))_{t\ge 1},(\lambda^i(s_0,a_0,h^{t-1}))_{t\ge 1})\\
        =&\inf_{(\lambda_0^i)} \sup_{a_0}\left[h_1((\lambda_0^i),a_0)+\inf_{(\lambda^i(h^t))_{t\ge 1}} \sup_{(a(s^t))_{t\ge 1}}h_2((\lambda_0^i),a_0,(a(s^t))_{t\ge 1},(\lambda^i(s_0,a_0,h^{t-1}))_{t\ge 1})\right].
    \end{aligned}
    \end{equation}}
    Given $a_0\in\mathcal{A},\,s_0\in\mathcal{S},(\lambda^i_0=\lambda^i(s_0))\in\mathbb{R}_+^I$, we know that $x_1=\zeta(x_0,a_0,s_0)$, and
    {\footnotesize 
    $$
    \begin{aligned}
    &h_2((\lambda_0^i),a_0,(a(s^t))_{t\ge 1},(\lambda^i(s_0,a_0,h^{t-1}))_{t\ge 1})\\
    =&\mathbb{E}_{s_0}\sum_{t=1}^{\infty}\beta^t\left\{r(a_t,x_t,s_t)+\sum_{i=1}^{I}\left[(\gamma^i+\lambda_0^i)g^i(a_t,x_t,s_t)+\lambda^i(h^t)\mathbb{E}_{s_t}\left(\sum_{n=0}^{\infty}\beta^n g^i(a_{t+n},x_{t+n},s_{t+n})-\bar{g}^i\right)\right]\right\} \\
    =&\sum_{s_1\in \mathcal{S}}\beta\pi(s_1|s_0)\mathbb{E}_{s_1}\sum_{t=1}^{\infty}\beta^{t-1}\left\{r(a_t,x_t,s_t)+\sum_{i=1}^{I}\Bigg[(\gamma^i+\lambda_0^i)g^i(a_t,x_t,s_t)\right.\\
&\qquad\qquad\left.\left.+\lambda^i(h^t)\mathbb{E}_{s_t}\left(\sum_{n=0}^{\infty}\beta^n g^i(a_{t+n},x_{t+n},s_{t+n})-\bar{g}^i\right)\right]\right\}\\
    =&\sum_{s_1\in\mathcal{S}}\beta\pi(s_1|s_0)L((a(s^t;s_0,s_1))_{t\ge 1},\lambda(s_0,a_0,h^{t-1})_{t\ge 1};(\gamma^i+\lambda^i_0),x_1,s_1).
    \end{aligned}
    $$ } 
    Therefore, \begin{equation}\label{equ:proof_recursive_mid3}\small 
    \begin{aligned}
    &\inf_{(\lambda^i(h^t))_{t\ge 1}} \sup_{(a(s^t))_{t\ge 1}}h_2((\lambda_0^i),a_0,(a(s^t))_{t\ge 1},(\lambda^i(s_0,a_0,h^{t-1}))_{t\ge 1})\\
    =&\inf_{(\lambda^i(h^t))_{t\ge 1}} \sup_{(a(s^t))_{t\ge 1}}\sum_{s_1\in\mathcal{S}}\beta\pi(s_1|s_0) L((a(s^t;s_0,s_1))_{t\ge 1},\lambda(s_0,a_0,h^{t-1})_{t\ge 1};(\gamma^i+\lambda^i_0),x_1,s_1)\\
    =&\beta\sum_{s_1\in\mathcal{S}}\pi(s_1|s_0)\inf_{(\lambda^i(s_0,a_0,h^{t-1}))_{t\ge 1}} \sup_{(a(s_0,s^t))_{t\ge 1}}L((a(s^t;s_0,s_1))_{t\ge 1},\lambda(s_0,a_0,h^{t-1})_{t\ge 1};(\gamma^i+\lambda^i_0),x_1,s_1)\\
    =&\beta\sum_{s_1\in\mathcal{S}}\pi(s_1|s_0)D(\gamma+\lambda_0,x_1,s_1)  =\beta\mathbb{E}_{s_0}D(\gamma+\lambda_0,x_1,s_1).
    \end{aligned}
    \end{equation}
    where the second equation separates the terms concerning different $s_1\in\mathcal{S}$, and the third equation uses the definition of the dual value function $D$\footnote{Note that $D$ is finite according to part 5 of Assumption \ref{ass:dynamic}, which is important to ensure the finiteness  of the sum of  several dual value functions.}.  The recursive formulation \eqref{equ:recursive_dual} can then be obtained by combining \eqref{equ:proof_recursive_mid2} and \eqref{equ:proof_recursive_mid3}.
    \hfill $ \Box $
\subsubsection*{Proof of Theorem \ref{thm:recursive2}}
    Following the same procedure as in the proof of Theorem \ref{thm:recursive1}, we have
{\footnotesize \begin{equation}\label{equ:recursive_dual_extend_mid}
    \begin{aligned}
        &\tilde{D}(\gamma,x,s)=\inf_{\lambda\in \mathcal{A}\rightarrow \mathbb{R}_+^I}\sup_{a\in\mathcal{A}}\Bigg[\left(r(x,a,s)+\sum_{i=1}^{I}\left(\gamma^ig^i(x,a,s)+\lambda^i(a)(g^i(x,a,s)-\bar{g}^i)\right)\right)\\
&\qquad+\beta\mathbb{E}_{s}\tilde{D}(\gamma+\lambda(a),x',s')\Bigg],\\
        \text{where }&x'=\zeta(x,a,s),\text{ and }p(x,a,s)\ge 0.
    \end{aligned}
\end{equation}}
It then follows from Lemma \ref{lem:tech_infsup}  that
     {\footnotesize {\footnotesize $$
    \begin{aligned}
        \tilde{D}(\gamma,x,s)&=\inf_{\lambda\in \mathcal{A}\rightarrow \mathbb{R}_+^I}\sup_{a\in\mathcal{A}}\Bigg[\left(r(x,a,s)+\sum_{i=1}^{I} \left(\gamma^ig^i(x,a,s)+\lambda^i(a)(g^i(x,a,s)-\bar{g}^i)\right)\right)\\
&\qquad+\beta\mathbb{E}_{s}\tilde{D}(\gamma+\lambda(a),x',s')\Bigg]\\
        &=\sup_{a\in\mathcal{A}}\inf_{\lambda\in \mathbb{R}_+^I}\Bigg[\left(r(x,a,s)+\sum_{i=1}^{I}\left(\gamma^ig^i(x,a,s)+\lambda^i(g^i(x,a,s)-\bar{g}^i)\right)\right)\\
&\qquad+\beta\mathbb{E}_{s}\tilde{D}(\gamma+\lambda,x',s')\Bigg],
    \end{aligned}
$$}}
        where $ x'=\zeta(x,a,s)$ and $ p(x,a,s)\ge 0$.
\hfill $ \Box $
\subsubsection*{Proof of Lemma \ref{lem:equiv_subgradient_lot}}
    We define the set 
    {\small $$
    \begin{aligned}
    \mathcal{U}=\{(&u,v)\in\mathbb{R}^{I+1}| \exists P\in\mathcal{P}(\tilde{\mathcal{A}}^{\infty}(x_0,s_0)),\,\\
&\qquad P \text{ satisfies all constraints in \eqref{equ:simplified_CK_lot} for $x_0=x,\,s_0=s$;}\\
    &u=\mathbb{E}_{s_0}^{(a_t\sim P)}\sum_{t=0}^{\infty}\beta^t r(x_t,a_t,s_t);\
    v=\mathbb{E}_{s_0}^{(a_t\sim P)}\sum_{t=0}^{\infty}\beta^t g(x_t,a_t,s_t)\} .
    \end{aligned}
    $$}
    It is straightforward to see that $\mathcal{U}$ is a non-empty convex set. Furthermore, based on the compactness of $\mathcal{P}(\tilde{\mathcal{A}}^{\infty}(x_0,s_0))$ in the weak-* topology, we know that $\mathcal{U}$ is closed. We define
    $
    F(\gamma_1,\gamma)=\sup _{(u,v)\in\mathcal{U}}\gamma_1 u+\gamma\cdot v.
    $
    According to Corollary 23.5.3 in \cite{rockafellar2015convex}, we know that 
    $(u^*,v^*)\in\arg \sup_{(u,v)\in \mathcal{U}} \gamma_1u+\gamma\cdot v\Leftrightarrow(u^*,v^*)\in\partial F(\gamma_1,\gamma)$.
    The dual function equals
    $
    D(\gamma,x,s)=\sup_{(u,v)\in\mathcal{U}}u+\gamma\cdot v.
    $
    We take $A=[0,E_{I}],\,b=[1,0_{I\times1}]$, and have
    $
    D(\gamma,x,s)=F(1,\gamma)=F(A\gamma+b).
    $
    Theorem 23.9 in \cite{rockafellar2015convex} implies that
    $$
    \partial D(\gamma,x,s)= A^T\partial F(A\gamma+b).
    $$
    Therefore,
    {\small $$
    \begin{aligned}
        &\phi\in\partial D(\gamma,x,s)
        \Leftrightarrow \phi=A^T\phi_F,\,\text{for some }\phi_F\in \partial F(1,\gamma)\\
        \Leftrightarrow&\exists (\phi_F(1),\phi_F(2:I+1))\in\arg \sup_{(u,v)\in\mathcal{U}}u+\gamma v: \phi=(\phi_F(2),\phi_F(3),\cdots,\phi_F(I+1))\\
        \Leftrightarrow& \exists P^*\in\mathcal{P}(\tilde{\mathcal{A}}^{\infty}(x_0,s_0)), \text{$P^*$ solves \eqref{equ:simplified_CK_lot}, and } \mathbb{E}_{s_0}^{(a\sim P)}\sum_{t=0}^{\infty}\beta^t g(x_t,a_t,s_t)=\phi.
    \end{aligned}
    $$}
    This completes the proof. \hfill $ \Box $

\subsubsection*{Proof of Theorem \ref{thm:recover_policy_main}}
    Since $\phi\in\partial D(\gamma,x,s)$, we know that there exists $(\psi(h^t))$, s.t. $(\psi(h^t))$ is a solution to \eqref{equ:simplified_CK_lot_statewise_promised} with
    {\small $$
    \begin{aligned}
    \phi&=\mathbb{E}_{s_0}^{a(h^t)\sim\psi(h^t)}\sum_{t=0}^{\infty}\beta^tg(x_t,a_t,s_t)\\
    &=\sum_{a\in\mathcal{A}}\psi_0(a)\left[g(x_0,a,s_0)+\beta\mathbb{E}_{s_0}\mathbb{E}_{h^1=(s_0,a,s_1)}^{a(h^t)\sim\psi(h^t)}\sum_{t=1}^{\infty}\beta^{t-1}g(x_t,a_t,s_t)\right].
    \end{aligned}
    $$}We take $\psi^*=\psi_0$ and $\phi'^{*}(a,s_1)=\mathbb{E}_{h^1=(s_0,a,s_1)}^{a(h^t)\sim\psi(h^t)}\sum_{t=1}^{\infty}\beta^{t-1}g(x_t,a_t,s_t)$. Hence \eqref{equ:recover_policy_cond1} is satisfied.
    
    It is straightforward to verify that, for any solution $\lambda^*$ of \eqref{equ:recursive_dual} and $\mu^*=0$, we have
    {\footnotesize \begin{equation}\label{eq:policy_recover_proof1}
    \begin{aligned}
      &\sum_{a\in\mathcal{A}}\psi^*(a)\Big[\left(r(x,a,s)+\sum_{i=1}^{I}\gamma^ig^i(x,a,s)+\lambda^{*,i}(g^i(x,a,s)-\bar{g}^i)+\mu^{*,i}(g^i(x,a,s)-\phi^i)\right)\\
      &\qquad\qquad+\beta\mathbb{E}_{s}W(\gamma+\lambda^*+\mu^*,x',s',\phi'^{*}(a,s'))\Big]\\
      &\ge\mathbb{E}_{s_0}^{(a(h^t)\sim\psi(h^t))}\left[\sum_{t=0}^{\infty}\beta^tr(x_t,a_t,s_t)
      + \sum_{t=0}^{\infty}\beta^t\gamma g(x_t,a_t,s_t)+\lambda^*(-\bar{g}+\sum_{t=0}^{\infty}\beta^tg(x_t,a_t,s_t))\right]\\
      &=\mathbb{E}_{s_0}^{(a(h^t)\sim\psi(h^t))}\left[\sum_{t=0}^{\infty}\beta^tr(x_t,a_t,s_t)+ \sum_{t=0}^{\infty}\beta^t\gamma g(x_t,a_t,s_t)\right]+\lambda^*\left(\sum_{a\in\mathcal{A}}\psi(a)\left(g(x,a,s)+\beta\mathbb{E}_s\phi'(a,s')-\bar{g} \right)\right)\\
      &\ge\mathbb{E}_{s_0}^{(a(h^t)\sim\psi(h^t))}\left[\sum_{t=0}^{\infty}\beta^tr(x_t,a_t,s_t)
      + \sum_{t=0}^{\infty}\beta^t\gamma g(x_t,a_t,s_t)\right]=W(\gamma,x,s,\phi),
    \end{aligned}
    \end{equation}}
    with equality only if \eqref{equ:recover_policy_cond2} is satisfied. Furthermore, since
    \begin{equation}\label{eq:policy_recover_proof2}
    \begin{aligned}
    &\sup_{\psi\in\mathcal{P}(\tilde{\mathcal{A}}(x,s)),\phi'(a,s)\in\mathbb{R}^I}\sum_{a\in\mathcal{A}}\psi(a)\Bigg[\left(r(x,a,s)+\sum_{i=1}^{I} \left( \gamma^ig^i(x,a,s)\right.\right.\\&\left.+\lambda^{*,i}(g^i(x,a,s)-\bar{g}^i)+\mu^{*,i}(g^i(x,a,s)-\phi^i)\right)\Bigg)+\beta\mathbb{E}_{s}W(\gamma+\lambda^*+\mu^*,x',s',\phi'(a,s'))\Bigg]\\
    \le & \sup_{a\in\tilde{\mathcal{A}}(x,s)}r(x,a,s)+\sum_{i=1}^I\gamma^ig^i(x,a,s)+\lambda^{*,i}(g^i(x,a,s)-\bar{g}^i)+\beta \mathbb{E}_sD(\gamma+\lambda^*,x',s')\\
    =&D(\gamma,x,s)=W(\gamma,x,s,\phi),
\end{aligned}
    \end{equation}
    we know that \eqref{eq:policy_recover_proof1} holds as an equality. Moreover, $(\psi^*,\phi'^*)$ is a solution to the sup problem given $\mu^*=0$ and $\lambda^*$, 
    and $(\psi(h^t;s_0,a,s')(t\ge 1))$ is a solution to $W(\gamma+\lambda^*+\mu^*,x',s',\phi'^*(a,s'))$. Therefore, \eqref{equ:recover_policy_cond0} is satisfied by Lemma \ref{lem:equiv_subgradient_lot}.
    Therefore $(\lambda^*,0,\psi^*,\phi'^*)$ is a solution to \eqref{equ:recursive_dual_promisedW} satisfying \eqref{equ:recover_policy_cond0}, \eqref{equ:recover_policy_cond1} and \eqref{equ:recover_policy_cond2}.

For the second part of the result, observe that
    from the $\ge 0$ part of  \eqref{equ:recover_policy_cond1} and \eqref{equ:recover_policy_cond2}, and the fact that $\phi^*(h^t)$ are bounded, one can verify directly that $(\psi^*(h^t))$ satisfies all the constraints.
    Furthermore, we know that
    {\footnotesize 
    \begin{align}
        &W(\gamma_0,x_0,s_0,\phi_0) \notag \\
        &=\sum_{a_0\in\mathcal{A}}\psi_0^*(a_0)\biggl[r(x_0,a_0,s_0)+\gamma_0g(x_0,a_0,s_0)
        +\lambda_0^*(g(x_0,a_0,s_0)-\bar{g})+\mu_0^*(g(x_0,a_0,s_0)-\phi_0) \notag \\
        &\qquad\qquad+\beta\mathbb{E}_{s_0}W(\gamma_0+\lambda_0^*+\mu_0^*,x_1,s_1,\phi_1^*)\biggr] \notag\\
        &=\sum_{a_0\in\mathcal{A}}\psi_0^*(a_0)\biggl\{r(x_0,a_0,s_0)+\gamma_0g(x_0,a_0,s_0)+\lambda_0^*(g(x_0,a_0,s_0)-\bar{g})+\mu_0^*(g(x_0,a_0,s_0)-\phi_0) \notag \\
&+\beta\mathbb{E}_{s_0}\biggl[\sum_{a_1\in\mathcal{A}}\psi_1^*(a_1)\biggl(r(x_1,a_1,s_1)+(\gamma_0+\lambda_0^*+\mu_0^*)g(x_1,a_1,s_1)+\lambda_1^*(g(x_1,a_1,s_1)-\bar{g}) \notag \\
&\qquad\qquad+\mu_1^*(g(x_1,a_1,s_1)-\phi_1)+\beta\mathbb{E}_{s_1}W(\gamma_0+\lambda_0^*+\lambda_1^*+\mu_0^*+\mu_1^*,x_2,s_2,\phi_2^*)\biggr)\biggr]\biggr\} \label{eq:Wexpand}.
\end{align}}
We recursively expand $W(\gamma_0+\lambda_0^*+\lambda_1^*+\mu_0^*+\mu_1^*,x_2,s_2,\phi_2^*)$ on the right-hand side of \eqref{eq:Wexpand} up to period $T+1$. Rearranging terms, we obtain
\begin{equation}\label{eq:WexpandtoTresult}
\footnotesize
\begin{aligned}
&W(\gamma_0,x_0,s_0,\phi_0)\\
=&\mathbb{E}_{s_0}^{(a(h^t))\sim (\psi(h^t))}\biggl[ \sum_{t=0}^{T}\beta^tr(x_t,a_t,s_t)+\gamma_0\sum_{t=0}^T\beta^tg(x_t,a_t,s_t)+\sum_{t=0}^{T}\lambda^*_t\biggl(-\beta^t\bar{g}+\sum_{r=t}^{T}\beta^rg(x_r,a_r,s_r)\biggr)\\
&+\sum_{t=0}^{T}\mu^*_t\biggl(-\beta^t\phi_t^*+\sum_{r=t}^{T}\beta^rg(x_r,a_r,s_r)\biggr)+\beta^{T+1}W(\gamma_{T+1}^*,x_{T+1},s_{T+1},\phi^*_{T+1})\biggr].
    \end{aligned}
    \end{equation}
    Note that $\phi_{T+1}^*\in\partial D(\gamma_{T+1}^*,x_{T+1},s_{T+1})$, and by Lemma \ref{lem:equiv_subgradient_lot}, 
    $
W(\gamma_{T+1}^*,x_{T+1},s_{T+1},\phi_{T+1}^*)=u^*_{T+1}+\gamma_{T+1}^*\phi_{T+1}^*
    $
    for some $u^*_{T+1}\le L$. We plug this into the right-hand side of \eqref{eq:WexpandtoTresult} and obtain
    {\footnotesize $$
    \begin{aligned}
        &W(\gamma_0,x_0,s_0,\phi_0)\\
   =&\mathbb{E}_{s_0}^{(a(h^t))\sim (\psi(h^t))}\biggl[ \sum_{t=0}^{T}\beta^tr(x_t,a_t,s_t)+\beta^{T+1}u_{T+1}^*+\gamma_0\left(\sum_{t=0}^T\beta^tg(x_t,a_t,s_t)+\beta^{T+1}\phi_{T+1}'^*\right)\\
        &\qquad\qquad+\sum_{t=0}^{T}\lambda^*_t\biggl(-\beta^t\bar{g}+\sum_{r=t}^{T}\beta^rg(x_r,a_r,s_r)+\beta^{T+1}\phi'^*_{T+1}\biggr)\\
        &\qquad\qquad+\sum_{t=0}^{T}\mu^*_t\biggl(-\beta^t\phi_t^*+\sum_{r=t}^{T}\beta^rg(x_r,a_r,s_r)+\beta^{T+1}\phi'^*_{T+1}\biggr)\biggr]\\
        =&\mathbb{E}_{s_0}^{(a(h^t))\sim (\psi(h^t))}\left[\sum_{t=0}^{T}\beta^tr(x_t,a_t,s_t)+\beta^{T+1}u_{T+1}^*+\gamma_0\left(\sum_{t=0}^{T}\beta^tg(x_t,a_t,s_t)+\beta^{T+1}\phi'^*_{T+1}\right)\right],
    \end{aligned}
    $$}according to the complementary slackness conditions \eqref{equ:recover_policy_cond1} and \eqref{equ:recover_policy_cond2}.
    Since $    \beta^{T+1}u^*_{T+1}+\gamma_0\beta^{T+1}\phi_{T+1}^*
    $
    uniformly converges to 0, we obtain
    {\footnotesize $$
    W(\gamma_0,x_0,s_0,\phi_0)=\mathbb{E}_{s_0}^{(a_t(h^t))\sim (\psi_t(h^t))}\left(\sum_{t=0}^{\infty}\beta^tr(x_t,a_t,s_t)+\gamma_0\sum_{t=0}^{\infty}\beta^tg(x_t,a_t,s_t)\right).
    $$}
    This completes the proof. \hfill $ \Box $

\subsubsection*{Proof of Theorem \ref{thm:policy_recover_adapt}}
\textbf{Step 1. Existence.}  Since $\phi\in\partial D(\gamma,x,s)$, we know that there exists $(\psi(h^t))$, s.t. $(\psi(h^t))$ is a solution to \eqref{equ:simplified_CK_lot_statewise_promised} with
    $$
    \begin{aligned}
\phi&=\mathbb{E}_{s_0}^{a(h^t)\sim\psi(h^t)}\sum_{t=0}^{\infty}\beta^tg(x_t,a_t(h^t),s_t)\\    &=\sum_{a\in\mathcal{A}}\psi_0(a)\left[g(x_0,a,s_0)+\beta\mathbb{E}_{s_0}\mathbb{E}_{h^1=(s_0,a,s_1)}^{a(h^t)\sim\psi(h^t)}\sum_{t=1}^{\infty}\beta^{t-1}g(x_t,a_t,s_t)\right].
    \end{aligned}
    $$
    We take $\psi^*=\psi_0$ and $\phi'^{*}(a,s_1)=\mathbb{E}_{h^1=(s_0,a,s_1)}^{a(h^t)\sim\psi(h^t)}\sum_{t=1}^{\infty}\beta^{t-1}g(x_t,a_t,s_t)$. Hence \eqref{equ:adaptcond1} is satisfied.
    
    It is straightforward to verify that, for any solution $\lambda^*$ of \eqref{equ:recursive_dual} and $\mu^*=\gamma$, we have
 {\small \begin{equation}\label{eq:policy_recover_proof3}
    \begin{aligned}
      &\sum_{a\in\mathcal{A}}\psi^*(a)\left[\left(r(x,a,s)+\lambda^{*}(g(x,a,s)-\bar{g})+\mu^{*}(g(x,a,s)-\phi)\right)\right.\\
&\qquad\left.+\beta\mathbb{E}_{s}W(\lambda^*+\mu^*,x',s',\phi'^{*}(a,s'))\right] \\
\ge&\mathbb{E}_{s_0}^{(a(h^t)\sim\psi(h^t))}\left[\sum_{t=0}^{\infty}\beta^tr(x_t,a_t,s_t)+\lambda^*(-\bar{g}
+\sum_{t=0}^{\infty}\beta^tg(x_t,a_t,s_t))+\mu^*(-\phi+\sum_{t=0}^{\infty}\beta^tg(x_t,a_t,s_t))\right]\\    =&\mathbb{E}_{s_0}^{(a(h^t)\sim\psi(h^t))}\left[\sum_{t=0}^{\infty}\beta^tr(x_t,a_t,s_t)\right]+\lambda^*\left(\sum_{a\in\mathcal{A}}\psi^*(a)\left(g(x,a,s)+\beta\mathbb{E}_s\phi'^{*}(a,s')-\bar{g} \right)\right)
      \\
      &+\mu^*\left(\sum_{a\in\mathcal{A}}\psi^*(a)(g(x,a,s)+\beta \mathbb{E}_s\phi'^{*}(a,s')-\phi)\right)\\\ge&\mathbb{E}_{s_0}^{(a(h^t)\sim\psi(h^t))}\left[\sum_{t=0}^{\infty}\beta^tr(x_t,a_t,s_t)\right]=W(\gamma,x,s,\phi)-\gamma\phi=W(0,x,s,\phi),
    \end{aligned}
    \end{equation}}(the last equality can be easily verified when $\phi\in\partial D(\gamma,x,s)$ and we omit the details here)
    with equality only if \eqref{equ:adaptcond2} is satisfied. Furthermore, by an argument similar to that in \eqref{eq:policy_recover_proof2}, the inequality \eqref{eq:policy_recover_proof3} holds as an equality; $(\psi^*,\phi'^*)$ is a solution to the sup problem given $\lambda^*,\mu^*=\gamma$, and $\psi^*(h^t;s_0,a,s')(t\ge 1)$ is a solution to $W(\lambda^*+\mu^*,x',s',\phi'^*(a,s'))$, hence \eqref{equ:adaptcond0} is satisfied by Lemma \ref{lem:equiv_subgradient_lot}.
    Therefore $(\lambda^*,\mu^*=\gamma,\psi^*,\phi'^*)$ is a solution to \eqref{equ:recursive_dual_policy_adapt_promised} satisfying \eqref{equ:adaptcond0}, \eqref{equ:adaptcond1} and \eqref{equ:adaptcond2}, which completes Step~1.

 \noindent \textbf{Step 2. Relation.} Suppose that $(\lambda^*,\mu^*,\psi^*,\phi'^*)$ is a  solution to \eqref{equ:recursive_dual_policy_adapt_promised} satisfying \eqref{equ:adaptcond0}, \eqref{equ:adaptcond1} and \eqref{equ:adaptcond2}. Since $\phi'^*(a,s')\in\partial D(\lambda^*+\mu^*,\zeta(x,a,s),s')$, there exists $\psi^*(a,s')(\cdot)$ which is a solution to $W(\lambda^*+\mu^*,x',s',\phi'^*(a,s'))$ and satisfies
  $$
  \mathbb{E}_{s_0=s',x_0=x'}^{a(h^t)\sim \psi^*(a,s')(h^t)}\sum_{t=0}^{\infty}\beta^tg(x_t,a_t,s_t)=\phi'^*(a,s').
  $$ We define $\bar{\psi}^*(h^t)$ s.t. $\bar{\psi}^*(s_0)=\psi^*$ and $\bar{\psi}^*(s_0,a,s',\cdot)=\psi^*(a,s')(\cdot)$. It is then straightforward to verify that
 {\small 
    \begin{align*}
    & W(0,x,s,\phi)= \sum_{a\in\mathcal{A}}\psi^*(a)\\&\Big[\left(r(x,a,s)+\lambda^{*}(g(x,a,s)-\bar{g})+\mu^{*}(g(x,a,s)-\phi)\right)
+\beta\mathbb{E}_{s}W(\lambda^*+\mu^*,x',s',\phi'^{*}(a,s'))\Big] \\
      =&\mathbb{E}_{s_0}^{(a(h^t)\sim\bar{\psi}^*(h^t))}\left[\sum_{t=0}^{\infty}\beta^tr(x_t,a_t,s_t)+\lambda^*(-\bar{g}+\sum_{t=0}^{\infty}\beta^tg(x_t,a_t,s_t))+\mu^*(-\phi+\sum_{t=0}^{\infty}\beta^tg(x_t,a_t,s_t))\right]\\     =&\mathbb{E}_{s_0}^{(a(h^t)\sim\bar{\psi}^*(h^t))}\left[\sum_{t=0}^{\infty}\beta^tr(x_t,a_t,s_t)
+\lambda^*\left(\sum_{a\in\mathcal{A}}\psi^*(a)\left(g(x,a,s)+\beta\mathbb{E}_s\phi'^*(a,s')-\bar{g} \right)\right) \right.\\
      &+\mu^*\left(\sum_{a\in\mathcal{A}}\psi^*(a)(g(x,a,s)+\beta\mathbb{E}_s\phi'^*(a,s')-\phi)\right)\Bigg]\\ 
      =&\mathbb{E}_{s_0}^{(a(h^t)\sim\bar{\psi}^*(h^t))}\left[\sum_{t=0}^{\infty}\beta^tr(x_t,a_t,s_t)\right],
    \end{align*}
    }due to the complementary slackness conditions \eqref{equ:adaptcond1} and \eqref{equ:adaptcond2}. Furthermore, letting $\tilde{\lambda}^{*},\,\tilde{\mu}^*$ be chosen as in Theorem \ref{thm:recover_policy_main}, we have
{\small 
    \begin{align}
      &\sum_{a\in\mathcal{A}}\psi^*(a)\biggl[\left(r(x,a,s)+\gamma g(x,a,s)+\tilde{\lambda}^{*}(g(x,a,s)-\bar{g})+\tilde{\mu}^{*}(g(x,a,s)-\phi)\right) \notag\\
&+\beta\mathbb{E}_{s}W(\gamma+\tilde{\lambda}^*+\tilde{\mu}^*,x',s',\phi'^{*}(a,s'))\biggr] \notag \\
      \ge&\mathbb{E}_{s_0}^{(a(h^t)\sim\bar{\psi}^*(h^t))}\biggl[\sum_{t=0}^{\infty}\beta^tr(x_t,a_t,s_t)+\sum_{t=0}^{\infty}\gamma g(x_t,a_t,s_t)+ \notag \\
      &+\tilde{\lambda}^*(-\bar{g}+\sum_{t=0}^{\infty}\beta^tg(x_t,a_t,s_t))+\tilde{\mu}^*(-\phi+\sum_{t=0}^{\infty}\beta^tg(x_t,a_t,s_t))\biggr] \notag \\
      =&\mathbb{E}_{s_0}^{(a(h^t)\sim\bar{\psi}^*(h^t))}\left[\sum_{t=0}^{\infty}\beta^tr(x_t,a_t,s_t)\right]+\gamma \left(\sum_{a\in\mathcal{A}}\psi^*(a)\left(g(x,a,s)+\beta\mathbb{E}_s\phi'^*(a,s')\right)\right) \notag \\
      &+\tilde{\lambda}^*\left(\sum_{a\in\mathcal{A}}\psi^*(a)\left(g(x,a,s)+\beta\mathbb{E}_s\phi'^*(a,s')-\bar{g} \right)\right) \notag \\
&+\tilde{\mu}^*\left(\sum_{a\in\mathcal{A}}\psi^*(a)\left(g(x,a,s)+\beta\mathbb{E}_s\phi'^*(a,s')-\phi \right)\right) \notag
      \\
      \ge&\mathbb{E}_{s_0}^{(a(h^t)\sim\bar{\psi}^*(h^t))}\left[\sum_{t=0}^{\infty}\beta^tr(x_t,a_t,s_t)\right]+\gamma \phi =W(0,x,s,\phi)+\gamma \phi =W(\gamma,x,s,\phi),\label{eq:policy_recover_proof4}
    \end{align}
    }with equality only if \eqref{equ:recover_policy_cond1} and \eqref{equ:recover_policy_cond2} are satisfied. Furthermore, by an argument similar to that in \eqref{eq:policy_recover_proof2}, the inequality \eqref{eq:policy_recover_proof4} holds as an equality; $(\psi^*,\phi'^*)$ is a solution to the sup problem given $(\tilde{\lambda}^*,\tilde{\mu}^*)$ which is a solution to \eqref{equ:recursive_dual_promisedW}, and $\psi^*(a,s')(\cdot)$ is a solution to $W(\gamma+\tilde{\lambda}^*+\tilde{\mu}^*,x',s',\phi'^*(a,s'))$, hence Lemma \ref{lem:equiv_subgradient_lot} implies \eqref{equ:recover_policy_cond0}.     Therefore $(\tilde{\lambda}^*,\tilde{\mu}^*,\psi^*,\phi'^*)$ is a solution to \eqref{equ:recursive_dual_promisedW} satisfying \eqref{equ:recover_policy_cond0}, \eqref{equ:recover_policy_cond1} and \eqref{equ:recover_policy_cond2}. \hfill $ \Box $

\subsubsection*{Proof of Theorem \ref{thm:algocon}}
The convergence of Lagrange multipliers follows from Proposition 2.7 in \cite{nedic2001convergence}.
    \begin{itemize}
        \item \textbf{1.}  Since the Lagrange multipliers converge, there exists $\bar{N}_1>0$, s.t. when $k\ge \bar{N}_1$, $
        \|(\lambda^k,\mu^k)-(\lambda^*,\mu^*)\|\le \epsilon.
        $
        \item \textbf{2.}  It suffices to consider $a\in\tilde{\mathcal{A}}(x,s)$ s.t. $a^k=a$ for some $k\ge 0$. By the upper semicontinuity of the subdifferential, there exists $\delta>0$, s.t. when $a^k=a$ and
        $
        \|(\lambda^k,\mu^k)-(\lambda^*,\mu^*)\|\le \delta, 
        $
        we have
        $$
        \text{{dist}}(\phi^k(s'),\partial D(\lambda^*+\mu^*,x',s'))\le \frac{\epsilon}{2}.
        $$
        According to the convergence of $(\lambda^k,\mu^k)$, there exists $\bar{M}_1> 0$, s.t. when $k\ge \bar{M}_1$ then 
        $
        \|(\lambda^k,\mu^k)-(\lambda^*,\mu^*)\|\le \delta.
        $
       Therefore, when $N> \bar{M}_1$, we have\footnote{The first inequality uses the fact that $\partial D(\lambda^*+\mu^*,x',s')$ is a convex set.} 
        {\footnotesize 
        \begin{align*}
        &\text{dist}(\phi'^N(a,s'),\partial D(\lambda^*+\mu^*,x',s'))= \text{dist}(\frac{\sum_{k=1}^N1_{a^k=a}\sigma^k\phi^k(s')}{\sum_{k=1}^N1_{a^k=a}\sigma^k},\partial D(\lambda^*+\mu^*,x',s'))\\
        \le &\sum_{k=1}^{\bar{M}_1}1_{a^k=a}\tilde{\sigma}^k\text{dist}(\phi^k(s'),\partial D(\lambda^*+\mu^*,x',s'))\\
&+\sum_{k=\bar{M}_1+1}^N 1_{a^k=a}\tilde{\sigma}^k\text{dist}(\phi^k(s'),\partial D(\lambda^*+\mu^*,x',s'))\\        \le&\sum_{k=1}^{\bar{M}_1}1_{a^k=a}\tilde{\sigma}^k\text{dist}(\phi^k(s'),\partial D(\lambda^*+\mu^*,x',s'))+\sum_{k=\bar{M}_1+1}^N1_{a^k=a}\tilde{\sigma}^k\frac{\epsilon}{2}\\\le&\sum_{k=1}^{\bar{M}_1}1_{a^k=a}\tilde{\sigma}^k\text{dist}(\phi^k(s'),\partial D(\lambda^*+\mu^*,x',s'))+\frac{\epsilon}{2}, 
        \end{align*}
        }
        where
        $
        \tilde{\sigma}^k=\frac{\sigma^k}{\sum_{i=1}^N1_{a^i=a}\sigma^i}.
        $
        Since
        $
        \psi^N(a)\tilde{\sigma}^k=\frac{\sum_{i=1}^N1_{a^i=a}\sigma^i}{\sum_{i=1}^N\sigma^i}\cdot\frac{\sigma^k}{\sum_{i=1}^N1_{a^i=a}\sigma^i}=\frac{\sigma^k}{\sum_{i=1}^N\sigma^i},
        $
        we have
        {\footnotesize \begin{equation}\label{proof_algorithm_mid}
\psi^N(a)\text{dist}(\phi'^N(a,s'),\partial D(\lambda^*+\mu^*,x',s')) \le \sum_{k=1}^{\bar{M}_1}\frac{\sigma^k}{\sum_{i=1}^{N}\sigma^i}\text{dist}(\phi^k(s'),\partial D(\lambda^*+\mu^*,x',s'))+\frac{\epsilon}{2}.
        \end{equation}}
        Since $\sum_{i=1}^{N}\sigma^i\rightarrow\infty$, there exists $\bar{N}_2>\bar{M}_1$, s.t. when $N>\bar{N}_2$, the RHS of \eqref{proof_algorithm_mid} $\le \epsilon$. 
        \item \textbf{3.} 
        According to the updating rule for $\mu^k$, we have
        \begin{equation}\label{eq:policy_muineq}
        \begin{aligned}
            \mu^{k+1}&=\max\{\mu^k- \sigma^k\left(g(x,a^k,s)-\phi+\beta\mathbb E_s \phi^k(s')\right),0\}\\
            &\ge\mu^k- \sigma^k\left(g(x,a^k,s)-\phi+\beta\mathbb E_s \phi^k(s')\right).
        \end{aligned}
        \end{equation}
        Adding equations \eqref{eq:policy_muineq} for $k=1,\cdots,N$, we have
        $$
        \mu^{N+1}\ge \mu^1-\sum_{k=1}^{N}\sigma^k\left(g(x,a^k,s)-\phi+\beta\mathbb E_s \phi^k(s')\right),
        $$
        implying that
        \begin{align*}
            &\sum_{a}\psi^N(a)(g(x,a,s)+\beta\mathbb{E}_s\phi'^N(a,s'))-\phi\\
            =&\frac{\sum_{a}\left(\sum_{1\le k\le N,\,a^k=a}\sigma^k\right)\left(g(x,a,s)-\phi+\beta\mathbb{E}_s\frac{\sum_{1\le k\le N,\,a^k=a}\sigma^k\phi^k(s')}{\sum_{1\le k\le N,\,a^k=a}\sigma^k}\right)}{\sum_{k=1}^N\sigma^k}\\=&
            \frac{\sum_{a}\left[\left(\sum_{1\le k\le N,\,a^k=a}\sigma^k\right)\left(g(x,a,s)-\phi\right)+\beta\mathbb{E}_s\sum_{1\le k\le N,\,a^k=a}\sigma^k\phi^k(s')\right]}{\sum_{k=1}^N\sigma^k}\\=&\frac{\sum_{k=1}^{N}\sigma^k\left(g(x,a^k,s)-\phi+\beta\mathbb E_s \phi^k(s')\right)}{\sum_{k=1}^N\sigma^k}\ge\frac{\mu^1-\mu^{N+1}}{\sum_{k=1}^N\sigma^k},
        \end{align*}
        where the right-hand side tends to $0$ as $N\rightarrow\infty$. Therefore, there exists $\bar{N}_3>0$, s.t. when $N>\bar{N}_3$, \textbf{(3.)} holds.
        \item \textbf{4. } As in \textbf{(3.)}, there exists $\bar{N}_4>0$, s.t. \textbf{(4.)} holds.
        \item \textbf{5. } Without loss of generality, it suffices to consider the case $I=1$ and $\mu^*>0$. Since $\mu^k\rightarrow\mu^*$, there exists $\bar{M}_2$, such that when $k\ge \bar{M}_2$, $\mu^k>0$, and \eqref{eq:policy_muineq} holds as an equality. Therefore, 
        \begin{align*}
        &\Big\|\sum_{a}\psi^N(a)(g(x,a,s)+\beta\mathbb{E}_s\phi'^N(a,s'))-\phi\|
        =\|\frac{\sum_{k=1}^{N}\sigma^k\left(g(x,a^k,s)-\phi+\beta\mathbb E_s \phi^k(s')\right)}{\sum_{k=1}^N\sigma^k}\Big\|\\
        =&\Bigg\|\frac{\begin{aligned}[t]
        &\sum_{k=1}^{\bar{M}_2}\sigma^k\left(g(x,a^k,s)-\phi+\beta\mathbb E_s \phi^k(s')\right)+\sum_{k=\bar{M}_2+1}^{N}\sigma^k\left(g(x,a^k,s)-\phi+\beta\mathbb E_s \phi^k(s')\right)
        \end{aligned}}{\sum_{k=1}^N\sigma^k}\Bigg\|\\
        \le &\frac{\sum_{k=1}^{\bar{M}_2}\sigma^k}{\sum_{k=1}^N\sigma^k}\|g(x,a,s)-\phi+\beta\mathbb{E}_s\phi^k(s')\|_{\infty}
 +|\frac{\mu^{\bar{M}_2+1}-\mu^{N+1}}{\sum_{k=1}^N\sigma^k}|\rightarrow 0,
        \end{align*}
        implying that
        $$
        |\langle\mu^*,\sum_{a}\psi^N(a)(g(x,a,s)+\beta\mathbb{E}_s\phi'^N(a,s'))-\phi \rangle|\rightarrow 0,\text{ as }N\rightarrow \infty.
        $$
        Hence there exists $\bar{N}_5>\bar{M}_2$, such that when $N\ge \bar{N}_5$, we have
        $$
        |\langle\mu^*,\sum_{a}\psi^N(a)(g(x,a,s)+\beta\mathbb{E}_s\phi'^N(a,s'))-\phi \rangle|<\epsilon.
        $$
        \item \textbf{6. } As in \textbf{(5.)}, there exists $\bar{N}_6>0$, s.t. \textbf{(6.)} holds. Therefore, $\bar{N}$ can be taken as $\bar{N}=\max_{1\le i\le 6}\{\bar{N}_i\}$. \hfill $ \Box $
    \end{itemize}

\section{\texorpdfstring{Mathematical results}{Mathematical results}}\label{app:math}
\subsection{Preliminary results}
The following definitions are from \cite{rockafellar2015convex}.
\begin{definition}
    Let $X$ be a Banach space, $f:X\rightarrow\mathbb{R}\cup\{+\infty,-\infty\}$ be an extended real-valued functional. If $f(x)\not=-\infty,\,\forall x\in X$, and $f$ is not identically $+\infty$, then we call $f$ a \textbf{proper} functional. The \textbf{effective domain} of $f$, denoted $\text{dom}(f)$, is defined by
    $$
    \text{dom}(f)=\{x\in X|f(x)<+\infty\}.
    $$
\end{definition}
\begin{definition} \label{def:epigraph}
    Let $X$ be a Banach space, $f:X\rightarrow\mathbb{R}\cup\{+\infty,-\infty\}$ be an extended real-valued functional. The \textbf{epigraph} of $f$, denoted $\text{epi}(f)$, is defined by
    $$
    \text{epi}(f):=\{(x,r)\in X\times \mathbb{R}| f(x)\le r\}.
    $$
\end{definition}
The following propositions follow directly from the definitions:
\begin{proposition}\label{prop:convex}
    Let $X$ be a Banach space, $f:X\rightarrow\mathbb{R}\cup\{+\infty,-\infty\}$ be a proper extended real-valued functional. Then
    $$
    f \text{ is convex}\Leftrightarrow \text{epi}(f) \text{ is convex in }X\times \mathbb{R}.
    $$
\end{proposition}
\begin{proposition}\label{prop:lsc}
    Let $X$ be a Banach space, $f:X\rightarrow\mathbb{R}\cup\{+\infty,-\infty\}$ be a proper extended real-valued functional. Then
    $$
    f \text{ is \textbf{lower semicontinuous (l.s.c.)}}\Leftrightarrow \text{epi}(f) \text{ is closed in }X\times \mathbb{R}.
    $$
\end{proposition}
\begin{lemma}\label{lem:ascoli}(\textit{Ascoli Theorem}) Let $X$ be a Banach space, and $E\subset X$ be a closed convex set. Then for any $x_0\notin E$, there exists $x^*\in X^*$ and $\gamma\in \mathbb{R}$, such that
$$
\langle x^*,x\rangle<\gamma<\langle x^*,x_0\rangle,\quad \forall x\in E.
$$
\end{lemma}
\begin{proof}
    The lemma follows directly from Theorem 1.7 (Hahn--Banach, second geometric form) in \cite{brezis2011functional} by taking \(A=E\) and \(B=\{x_0\}\).
\end{proof}
\subsection{Theorem \ref{thm:dualgap} and randomization}\label{app:math_cor} 
The following corollaries show how Theorem \ref{thm:dualgap} is related to randomization.
\begin{corollary}\label{cor:finite}
    Assume that $\Omega=\mathbb{R}^m$,  $Y=\mathbb{R}^n$, $d<\infty$, and $f$ is proper and bounded from below, i.e., there exists $M\in\mathbb{R}$, s.t. $f(x)\ge M$, $\forall x\in \Omega$. Then for any $\epsilon>0$, there exists $y_\epsilon\in Y,\,\|y_{\epsilon}\|\le \epsilon$, $\{p_i\ge 0\}_{i=1}^{n+2}$, and $\{x_i\in \Omega\}_{i=1}^{n+2}$, s.t. 
    $$\sum_{i=1}^{n+2}p_i=1,\quad\sum_{i=1}^{n+2}p_ig(x_i)-y_{\epsilon}\le \theta_Y,
\mbox{ and }
    \sum_{i=1}^{n+2}p_if(x_i)\le d+\epsilon.
    $$
\end{corollary}
\begin{proof}
    Since $f$ is proper and bounded from below, we know that the perturbation function $v$ is also proper and bounded from below, and therefore satisfies the property \eqref{equ:regularity}.  Theorem \ref{thm:bicon} implies that
    $
    \text{epi} (v^{**})=\text{cl } \text{co }\text{epi} (v).
    $
    Note that $(\theta_Y,d)\in\text{epi}(v^{**})$ according to Theorem \ref{thm:dualgap}, hence there exists $(\bar{y}, \bar{v})\in \text{co }\text{epi} (v)\subset Y\times \mathbb{R}$, such that $\|(\bar{y},\bar{v})-(\theta_Y,d)\|_{Y\times \mathbb{R}}<\frac{\epsilon}{2},$ implying that
    $$
    \|\bar{y}\|_{Y}<\frac{\epsilon}{2}<\epsilon,\quad \|\bar{v}-d\|<\frac{\epsilon}{2}.
    $$
    By Carath\'eodory's Theorem, since $(\bar{y}, \bar{v})\in \text{co }\text{epi} (v)\subset \mathbb{R}^{n+1}$, there exists $\{p_i\ge 0\}_{i=1}^{n+2}$, $\{(y_i,v_i)\in \text{epi}(v)\}_{i=1}^{n+2}$, s.t.
    $$\sum_{i=1}^{n+2}p_i=1,\quad \sum_{i=1}^{n+2}p_i(y_i,v_i)=(\bar{y},\bar{v}).$$
    By the definition of $v$, for any $i\in\{1,\cdots,n+2\}$, there exists $x_i\in X$, such that
    $$
    g(x_i)\le y_i,\quad f(x_i)\le v_i+\frac{\epsilon}{2}.
    $$
    Therefore, we have
    $$
    \sum_{i=1}^{n+2}p_ig(x_i)-\bar{y}\le \sum_{i=1}^{n+2}p_iy_i-\bar{y}= \theta_Y,
    $$
    and
    $$
\sum_{i=1}^{n+2}p_if(x_i)\le\sum_{i=1}^{n+2}p_i(v_i+\frac{\epsilon}{2})=\bar{v}+\frac{\epsilon}{2}\le d+\epsilon. 
    $$
    This completes the proof.
\end{proof}
\begin{corollary}\label{cor:infinite}
    Assume that $\Omega$ is a non-empty set, $Y$ is a Banach space with a closed positive cone $P\subset Y$, $d<\infty$, and $f$ is proper and bounded from below,  i.e., there exists $M\in\mathbb{R}$, s.t. $f(x)\ge M$, $\forall x\in \Omega$. Then for any $\epsilon>0$, there exists $y_\epsilon\in Y,\,\|y_{\epsilon}\|\le \epsilon$, $N\in \mathbb{N}_+$, $\{p_i\ge 0\}_{i=1}^{N}$, and $\{x_i\in \Omega\}_{i=1}^{N}$, s.t. 
    $$\sum_{i=1}^{N}p_i=1,\quad\sum_{i=1}^{N}p_ig(x_i)-y_{\epsilon}\le \theta_Y
    \mbox{ and }
    \sum_{i=1}^{N}p_if(x_i)\le d+\epsilon.
    $$
\end{corollary}
\begin{proof}
    Similarly to the proof of Corollary \ref{cor:finite}, we can still find $(\bar{y},\bar{v})\in \text{co }\text{epi}(v)\subset Y\times \mathbb{R}$, such that
     $$
    \|\bar{y}\|_{Y}<\frac{\epsilon}{2}<\epsilon,\quad \|\bar{v}-d\|<\frac{\epsilon}{2}.
    $$
Note that $\text{co }\text{epi}(v)$ can be written as
$$
\text{co }\text{epi}(v)=\{\sum_{i=1}^{n}p_i(y_i,v_i)|n\in\mathbb{N}_+,\,p_i\ge 0,\,(y_i,v_i)\in \text{epi}(v),\,\sum_{i=1}^{n}p_i=1\}.
$$
Hence there exists $N\in \mathbb{N}_+$,  $\{p_i\ge 0\}_{i=1}^{N}$, $\{(y_i,v_i)\in \text{epi}(v)\}_{i=1}^{N}$, s.t.
    $$\sum_{i=1}^{N}p_i=1,\quad \sum_{i=1}^{N}p_i(y_i,v_i)=(\bar{y},\bar{v}).$$
The rest of the proof is the same as the proof of Corollary \ref{cor:finite}, and the details are omitted here.
\end{proof}
\begin{corollary}\label{cor:optimal_lot}
     Assume that $\Omega$ is a non-empty set, $Y$ is a Banach space with a closed positive cone $P\subset Y$, $d<\infty$, and $f$ is proper and bounded from below,  i.e., there exists $M\in\mathbb{R}$, s.t. $f(x)\ge M$, $\forall x\in \Omega$. Then for any  $N\in \mathbb{N}_+$, $\{p_i\ge 0\}_{i=1}^{N}$, and $\{x_i\in \Omega\}_{i=1}^{N}$, s.t. 
$\sum_{i=1}^{N}p_i=1,$ $\sum_{i=1}^{N}p_ig(x_i)\le \theta_Y,
    $
    we have
    $
    \sum_{i=1}^{N}p_if(x_i)\ge d.
    $
\end{corollary}
\begin{proof}
    For any $N\in\mathbb{N}_+$, s.t. $\sum_{i=1}^{N}p_i=1,\quad\sum_{i=1}^{N}p_ig(x_i)\le \theta_Y,
    $
    it is straightforward to verify that $(y_i,f(x_i))\in \text{epi}(v)$ for any $i\in\{1,\cdots,N\},\,y_i\ge g(x_i)$ by the definition of the perturbation functional $v$. Therefore, 
    $$
    (\sum_{i=1}^{N}p_iy_i,\sum_{i=1}^{N}p_if(x_i))\in \text{co }\text{epi}(v),\quad\forall y_i\ge g(x_i).
    $$
    Since
    $
    \sum_{i=1}^{N}p_ig(x_i)\le \theta_Y,
    $
    we have
    $
    (\theta_Y,\sum_{i=1}^{N}p_if(x_i))\in \text{co }\text{epi}(v)\subset \text{cl }\text{co }\text{epi}(v)=\text{epi}(v^{**}),
    $
    implying (by Theorem \ref{thm:dualgap}) that
    $
    \sum_{i=1}^{N}p_if(x_i)\ge v^{**}(\theta_Y)=d.
   $
\end{proof}
\clearpage
\section{\texorpdfstring{Existence of Lagrange multipliers (Online Only)}{Existence of Lagrange multipliers (Online Only)}}
\label{lagrange}
\subsection{The $\ell^{\infty}$-optimization problem}
In this appendix, we review in detail the existence of Lagrange multipliers in appropriate function spaces. Most of the material can be found in \cite{dechert1982lagrange} and \cite{pavonionline}. It is included here for completeness.

\begin{definition}\label{def:linfopt}
Let $f:\ell^{\infty}\rightarrow \mathbb{R}$ be a functional bounded from below, i.e., there exists $M\in\mathbb{R}$, s.t. $f(x)\ge M$, $\forall x\in \ell^{\infty}$; and $g:\ell^{\infty}\rightarrow \ell^{\infty}$ be an arbitrary functional.\footnote{The norm \(\|\cdot\|_\infty\) is used only on the image of \(g\); the domain \(\ell^\infty\) (and \(f\)'s argument) is treated as a vector space without topology.} The \textbf{$\ell^{\infty}$-optimization problem} is defined by
\begin{equation}\label{equ:def_linfopt}
     \inf_{x\in \ell^{\infty}}f(x)
    \text{ s.t. }g(x)\le 0,
\end{equation}
where $g(x)\le 0$ means $g_t(x)\le 0,\,\forall t\in\mathbb{Z}^+$.
\end{definition}
\begin{lemma}\label{lem:Yosida}{(Yosida--Hewitt decomposition theorem; see \cite{yosida1952finitely}, Theorems 1.23, 1.24, and 2.3)}
    For any $\lambda\in \ell^{\infty,*}(\text{ or }\ell^{\infty,*}_+)$, there exists a unique $\lambda^1\in \ell^1(\text{ or }\ell^1_+)$, and a unique $\lambda^s\in \ell^s(\text{ or }\ell^s_+)$, s.t. $\lambda=\lambda^1+\lambda^s$. Here $\ell^{s}$ is the set of \textbf{purely finitely additive measures}.
\end{lemma}
\begin{remark}
\begin{enumerate}
    \item A well-known example of $\lambda\in \ell^{\infty,*}-\ell^{1}$ is constructed using the Hahn--Banach theorem. To be precise, we define the set $c\subset \ell^{\infty}$ as
    $$
    c=\{x\in \ell^{\infty}|\lim_{t\rightarrow\infty}x_t\text{ exists}\}.
    $$
    Therefore the operator $f_c$ defined by
    $$
    f_c(x)=\lim_{t\rightarrow\infty}x_t,\quad\forall x\in c,
    $$
    is a bounded linear operator on $c$. By the Hahn--Banach theorem, this operator $f_c$ can be extended to be a bounded linear operator on $\ell^{\infty}$, or equivalently, to be an element of $\ell^{\infty,*}$, denoted by $f$. It is straightforward to check that $f\notin \ell^{1}$. 
    \item Another approach to construct an element belonging to $\ell^{\infty,*}-\ell^1$ is to use \textbf{ultrafilters}. To be precise, for any $A\subseteq \mathbb{Z}^+$, let $e_A$ be the indicator function of $A$, i.e., $$
    e_A(t)=\begin{cases}
        1,&t\in A;\\
        0,&t\notin A.
    \end{cases}
    $$We consider the set\footnote{As shown in \cite{dechert1982lagrange} (p.~289), $M$ is indeed the set of extreme points of
    $$
    G:=\{\lambda\in\ell^{\infty,*}|\langle \lambda,e_{\mathbb{Z}^+}\rangle =1\}.
    $$.} 
    $$
    M:=\{\lambda\in \ell^{\infty,*}|\langle\lambda,e_{\mathbb{Z}^+}\rangle=1;\,\langle\lambda,e_A\rangle=0\text{ or }1,\,\forall A\subset\mathbb{Z}^+\}.$$
    We note that when $\lambda\in M$, such that $\langle \lambda,e_{A}\rangle=0$ for any finite sets $A\subset \mathbb{Z}^+$, then $\lambda\in \ell^{\infty,*}-\ell^1$.
    \item We define $c_0\subset c\subset \ell^{\infty}$ as
    $$
    c_0=\{x\in \ell^{\infty}|\lim_{t\rightarrow\infty}x_t=0\}.
    $$
    Then $c_0^{*}=\ell^{1}$. Various asymptotic behaviors of $x\in\ell^{\infty}$ give rise to bounded linear functionals in $\ell^{\infty,*}-\ell^1$. This is why we need restrictions on asymptotic behaviors of $g$ below for further discussions.
 \end{enumerate}
\end{remark}
\begin{definition}\label{def:AI_ANA}
Let $g:\ell^{\infty}\rightarrow \ell^{\infty}$ be a functional. For any $u,\,v\in\ell^{\infty}$, let 
$$
x^T(u,v):=\begin{cases}
    u_t,&t\le T;\\
    v_t,&t>T.
\end{cases}
$$
We say the functional $g$ is \textbf{Asymptotically Insensitive (AI)}, if
$$
\lim_{t\rightarrow \infty}\left[g_t(x^T(u,v))-g_t(v)\right]=0,\quad \forall u,\,v\in \ell^{\infty},\,T\in\mathbb{Z}^+.
$$
We say the functional $g$ is \textbf{Asymptotically Non-Anticipatory (ANA)} if 
$$
\lim_{T\rightarrow \infty}\left[g_t(x^T(u,v))-g_t(u)\right]=0,\quad \forall u,\,v\in \ell^{\infty},\,t\in\mathbb{Z}^+.
$$
\end{definition}
\begin{remark}
    \begin{enumerate}
        \item This terminology was first introduced in \cite{dechert1982lagrange} and was adapted for the dual approach to dynamic problems in \cite{pavoni2018dual}. Here we use the definitions from \cite{pavoni2018dual}.
        \item Intuitively, $g$ being AI means that, for any $t\in \mathbb{Z}^+$, $g_t$ depends little on $x_r$ with $r\ll t$; and $g$ being ANA means that, for any $t\in \mathbb{Z}^+$, $g_t$ depends little on $x_r$ with $r\gg t$. A typical type of constraint in dynamic problems is
        $$
        g^{D}_t(x):=\sum_{r=0}^{\infty}\beta^{r}h_{t+r}(x_{t+r})\le 0,\quad \forall t\in \mathbb{Z}^+,
        $$
        where $0<\beta<1$ denotes the discount factor. It is straightforward to check that $g^D$ is both AI and ANA under some mild regularity assumptions on $h$.
    \end{enumerate}
\end{remark}
\begin{lemma}\label{lem:AI_lambdas}
    Let $g:\ell^{\infty}\rightarrow\ell^{\infty}$ be AI. Then for any $\lambda^s\in\ell^{s}$, $u,\,v\in\ell^{\infty}$, $T\in\mathbb{Z}^+$, we have
    $$
    \langle \lambda^s,g(x^{T}(u,v))\rangle=\langle\lambda^s,g(v)\rangle.
    $$
\end{lemma}  
\begin{proof}
    See \cite{pavonionline} (p.~11, Lemma C.1).
\end{proof}
\begin{lemma}\label{lem:ANA_lambdas}
     Let $g:\ell^{\infty}\rightarrow\ell^{\infty}$ be ANA. Given $u,\,v\in\ell^{\infty}$. Moreover, assume that $g(x^T(u,v))$ is uniformly bounded in $\ell^{\infty}$ with respect to $T\in\mathbb{Z}^{+}$. Then for any $\lambda^1\in\ell^{1}$, we have
     $$
     \lim_{T\rightarrow\infty}\langle \lambda^1,g(x^T(u,v))\rangle=\langle\lambda^1,g(u)\rangle.
     $$
\end{lemma}
\begin{proof}
    See \cite{pavonionline} (p.~11, Lemma C.2).
\end{proof}
\begin{theorem}\label{thm:linfdual=l1}
    Let $f:\ell^{\infty}\rightarrow\mathbb{R}$ satisfy
    $
    \lim_{T\rightarrow\infty}f(x^T(u,v))=f(u),
    $
    for any $u,v\in\ell^{\infty}$. Let $g:\ell^{\infty}\rightarrow\ell^{\infty}$ be AI and ANA, and assume that there exists $u_0\in \ell^{\infty}$ such that $\limsup_{t\rightarrow\infty} g_t(u_0)\le 0$\footnote{This assumption is weaker than the assumption that there exists at least one feasible point to the problem \eqref{equ:def_linfopt}. Indeed, if there exists a feasible point $v_0$ s.t. $g(v_0)\le 0$, then $v_0$ satisfies $\limsup_{t\rightarrow\infty}g_t(v_0)\le 0$.}. Assume also that, for any $u,\,v\in\ell^{\infty}$, $g(x^T(u,v))$ is uniformly bounded in $\ell^{\infty}$ with respect to $T\in\mathbb{Z}^{+}$. Then the dual problem of \eqref{equ:def_linfopt} satisfies
    \begin{equation}\label{equ:thm_linfdual=l1}
    d=\sup_{\lambda\in\ell^{\infty,*}_+}\inf_{x\in \ell^{\infty}}f(x)+\langle\lambda,g(x)\rangle=\sup_{\lambda\in\ell^{1}_+}\inf_{x\in \ell^{\infty}}f(x)+\langle\lambda,g(x)\rangle.
    \end{equation}
\end{theorem}
\begin{proof}
\begin{itemize}
    \item For step 1, we aim to show that
    \begin{equation}\label{proof:linfdual=l1_step1}
    \inf_{x\in \ell^{\infty}}f(x)+\langle\lambda^1+\lambda^s,g(x)\rangle\le \inf_{x\in \ell^{\infty}}f(x)+\langle\lambda^1,g(x)\rangle,\quad\forall \lambda^1\in \ell^1_+,\,\lambda^s\in\ell^{s}_+.
    \end{equation}
    Given $\lambda^1\in\ell^1_{+}$, $\lambda^s\in \ell^{s}_+$. For any $u\in\ell^{\infty}$, according to Lemma \ref{lem:AI_lambdas}, we know that
    $$
    \langle \lambda^s,g(x^T(u,u_0))\rangle=\langle \lambda^s,g(u_0)\rangle\le 0.
    $$
    Therefore,
    \begin{equation}\label{proof:linfdual=l1_step1_mid1}
        \begin{aligned}
        &\inf_{x\in\ell^{\infty}}f(x)+\langle\lambda^1+\lambda^s,g(x)\rangle
          \le  f(x^T(u,u_0))+\langle  \lambda^1+\lambda^s, g(x^T(u,u_0))\rangle\\
          \le& f(x^T(u,u_0))+\langle\lambda^1,g(x^T(u,u_0))\rangle,\quad\forall u\in \ell^{\infty},\,T\in\mathbb{Z}^+.
        \end{aligned}
    \end{equation}
    According to Lemma \ref{lem:ANA_lambdas}, we know that
    $$
    \lim_{T\rightarrow\infty}\langle \lambda^1,g(x^T(u,u_0))\rangle=\langle \lambda^1,g(u)\rangle.
    $$
    Therefore, we take the limit $T\rightarrow\infty$ in \eqref{proof:linfdual=l1_step1_mid1} and obtain
    \begin{equation}\label{proof:linfdual=l1_step1_mid2}
    \begin{aligned}
        &\inf_{x\in\ell^{\infty}}f(x)+\langle\lambda^1+\lambda^s,g(x)\rangle
        \le \lim_{T\rightarrow\infty}f(x^T(u,u_0))+\langle\lambda^1,g(x^T(u,u_0))\rangle\\
        =&f(u)+\langle\lambda^1,g(u)\rangle,\quad\forall u\in \ell^{\infty}.
    \end{aligned}
    \end{equation}
Since $(\lambda^1,\lambda^s)$ in \eqref{proof:linfdual=l1_step1_mid2} is arbitrarily chosen from $(\ell^1_+,\ell^{s}_+)$, we can conclude \eqref{proof:linfdual=l1_step1} from \eqref{proof:linfdual=l1_step1_mid2}.

    \item  For step 2, we use the Yosida--Hewitt decomposition to finish the proof. According to Lemma \ref{lem:Yosida}, any $\lambda\in\ell_+^{\infty,*}$ can be decomposed as $\lambda=\lambda^1+\lambda^s$, where $\lambda^1\in\ell_+^1$ and $\lambda^s\in\ell_+^s$. It is then straightforward to conclude \eqref{equ:thm_linfdual=l1} from \eqref{proof:linfdual=l1_step1}.
\end{itemize}
\end{proof}

\begin{remark}
    Under the same assumptions as Theorem \ref{thm:linfdual=l1}, but with the assumption of the existence of $u_0\in \ell^{\infty}$ such that $\limsup_{t\rightarrow\infty}g_t(u_0)\le 0$ replaced by the \textbf{Slater's condition}, it is shown in \cite{pavonionline} (p.~11, Theorem C.2) that if $(x^*,\lambda^*)\in \ell^{\infty}\times \ell^{\infty,*}$ is a saddle point of $L(x,\lambda)=f(x)+\langle \lambda,g(x)\rangle$, then  $(x^*,\lambda^*)\in \ell^{\infty}\times \ell^{1}$. The properties of the saddle points for the Lagrangian are studied in \cite{dechert1982lagrange} and \cite{pavoni2018dual}, while Theorem \ref{thm:linfdual=l1} addresses the dual problem.
\end{remark}

    In practical applications, it is more common to consider optimization problems with a bounded feasible set $\mathcal{A}\subset\ell^{\infty}$. The requirement for $f$ in Theorem \ref{thm:linfdual=l1} and the ANA property for $g$ are easily satisfied \textbf{uniformly} with respect to  $u\in\mathcal{A},\,v\in\ell^{\infty}$,  s.t. $x^T(u,v)\in\mathcal{A}$ based on the existence of a discount factor $0<\beta<1$ over periods. However, the AI property for $g$ is mainly satisfied by problems with only forward constraints and backward constraints that only affect a few periods. Therefore, for problems with a bounded feasible set $\mathcal{A}\subset \ell^{\infty}$, a discount factor $0<\beta<1$ across periods, forward-looking constraints, and general \textbf{backward-looking constraints}, different assumptions should be made.

    Let $\mathcal{A}\subset\ell^{\infty}$. For any $u\in\mathcal{A}$, we define
    $$
    \mathcal{B}_T(u)=\{v\in\ell^{\infty}|x^T(u,v)\in\mathcal{A}\}.
    $$
    In this section, we consider a slight variant of \eqref{equ:def_linfopt}:
    \begin{equation}\label{equ:def_linfopt_bddA}
\begin{aligned}
     &\inf_{x\in \mathcal{A}}f(x),\\
    \text{s.t. }&g(x)\le 0.
    \end{aligned}
    \end{equation}
    
\begin{definition}\label{def:UANA}
    We say the functional $g:\mathcal{A}\rightarrow\ell^{\infty}$ is \textbf{Uniformly Asymptotically Non-Anticipatory (UANA)}, if for any $t\in \mathbb{Z}^+$ and $u\in\mathcal{A}$,
    $$
    \lim_{T\rightarrow \infty}\sup_{v\in\mathcal{B}_T(u)}|g_t(x^T(u,v))-g_t(u)|=0
    $$
    \end{definition}
\begin{lemma}\label{lem:UANA_lambdas}
     Let $g:\mathcal{A}\rightarrow\ell^{\infty}$ be UANA. Given $u\in\mathcal{A}$, assume additionally that $g(x^T(u,v))$ is uniformly bounded in $\ell^{\infty}$ with respect to $T\in\mathbb{Z}^{+}$ and $v\in \mathcal{B}_T(u)$. Then for any $\lambda^1\in\ell^{1}$, we have
     $$
     \lim_{T\rightarrow\infty}\sup_{v\in\mathcal{B}_T(u)}|\langle \lambda^1,g(x^T(u,v))\rangle-\langle\lambda^1,g(u)\rangle|=0
     $$
\end{lemma}
\begin{proof}
    Since $g(x^T(u,v))$ is uniformly bounded in $\ell^{\infty}$, there exists $M>0$, s.t.
    $$
    \|g(x^T(u,v))-g(u)\|_{\ell^{\infty}}<M, \quad \forall T\in \mathbb{Z}^+,\,v\in \mathcal{B}_T(u).
    $$
    Since $\lambda^1\in \ell^1$, for any $\epsilon>0$, there exists $\bar{t}>0$, such that
    $$
    \sum_{t=\bar{t}+1}^{\infty}|\lambda^1_t|<\frac{\epsilon}{2M}.
    $$
    Since $g$ is UANA, there exists $\bar{T}>0$, such that
    $$
    |g_t(x^T(u,v))-g_t(u)|<\frac{\epsilon}{2\|\lambda^1\|_{\ell^1}},\quad \forall  T\ge\bar{T},\,v\in\mathcal{B}_T(u),\,t\le\bar{t}.
    $$
    Therefore, for $\lambda^1\not=0$, we have
    $$
    \begin{aligned}
    &|\langle \lambda^1, g(x^T(u,v))-g(u)\rangle|\\
    \le&\sum_{t=1}^{\bar{t}}|\lambda^1_t|\left[\max_{1\le t\le \bar{t}}|g_t(x^T(u,v))-g_t(u)|\right]+\sum_{t=\bar{t}+1}^{\infty}|\lambda_t^1|\|g(x^T(u,v))-g(u)\|_{\ell^{\infty}}\\
    \le&\|\lambda^1\|_{\ell^1}\frac{\epsilon}{2\|\lambda^1\|_{\ell}^1}+\frac{\epsilon}{2M}\cdot M=\epsilon,\quad \forall T\ge \bar{T},\, v\in \mathcal{B}_T(u). 
    \end{aligned}
    $$
    Since $\epsilon$ is chosen arbitrarily, we finish the proof.
\end{proof}
\begin{assumption}\label{ass:replace_AI}
    We assume that for any $u\in \mathcal{A}$ and $\bar{T}\in\mathbb{Z}^+$, there exists $T\ge \bar{T}$ and $v\in\mathcal{B}_T(u)$, s.t.
    $$
    \limsup_{t\rightarrow\infty}g_t(x^T(u,v))\le 0.
    $$
\end{assumption}
\begin{theorem}\label{thm:linfdual=l1_adaption}
    Let $f:\mathcal{A}\rightarrow\mathbb{R}$ satisfy
    \begin{equation}\label{cond:f_dual=l1_adaption}
    \lim_{T\rightarrow\infty}\sup_{v\in\mathcal{B}_T(u)}|f(x^T(u,v))-f(u)|=0.
    \end{equation}
    Let $g:\mathcal{A}\rightarrow\ell^{\infty}$ be UANA, and satisfy Assumption \ref{ass:replace_AI}. Assume additionally that, for any $u\in\mathcal{A}$, $g(x^T(u,v))$ is uniformly bounded in $\ell^{\infty}$ with respect to $T\in\mathbb{Z}^{+}$ and $v\in\mathcal{B}_T(u)$. Then the dual problem of \eqref{equ:def_linfopt} satisfies
    \begin{equation}\label{equ:thm_linfdual=l1_adaption}
    d=\sup_{\lambda\in\ell^{\infty,*}_+}\inf_{x\in \mathcal{A}}f(x)+\langle\lambda,g(x)\rangle=\sup_{\lambda\in\ell^{1}_+}\inf_{x\in \mathcal{A}}f(x)+\langle\lambda,g(x)\rangle.
    \end{equation}
\end{theorem}
\begin{proof}
The proof follows the same logic as that of Theorem \ref{thm:linfdual=l1}. We again divide it into two steps.
\begin{itemize}
    \item For step 1, we aim to show that
    \begin{equation}\label{proof:linfdual=l1_step1_adaption}
    \inf_{x\in \mathcal{A}}f(x)+\langle\lambda^1+\lambda^s,g(x)\rangle\le \inf_{x\in \mathcal{A}}f(x)+\langle\lambda^1,g(x)\rangle,\quad\forall \lambda^1\in \ell^1_+,\,\lambda^s\in\ell^{s}_+.
    \end{equation}
    Given $\lambda^1\in\ell^1_{+}$, $\lambda^s\in \ell^{s}_+$. For any $u\in\mathcal{A}$ and $\bar{T}>0$, by Assumption \ref{ass:replace_AI}, we know that there exists $T\ge \bar{T}$ and $v(u,\bar{T})\in\mathcal{B}_T(u)$, s.t.
    $$
    \langle \lambda^s,g(x^T(u,v(u,\bar{T})))\rangle\le 0.
    $$
    Therefore,
    \begin{equation}\label{proof:linfdual=l1_step1_mid1_adaption}
        \begin{aligned}
        &\inf_{x\in\mathcal{A}}f(x)+\langle\lambda^1+\lambda^s,g(x)\rangle\\
          \le&  f(x^T(u,v(u,\bar{T})))+\langle  \lambda^1+\lambda^s,g(x^T(u,v(u,\bar{T})))\rangle\\
          \le& f(x^T(u,v(u,\bar{T})))+\langle\lambda^1,g(x^T(u,v(u,\bar{T})))\rangle,\quad\forall u\in \mathcal{A},\,\bar{T}\in\mathbb{Z}^+.
        \end{aligned}
    \end{equation}
    According to Lemma \ref{lem:UANA_lambdas}, and the assumption that 
     $$
    \lim_{T\rightarrow\infty}\sup_{v\in\mathcal{B}_T(u)}|f(x^T(u,v))-f(u)|=0,
    $$
    we take the limit $T\rightarrow\infty$ in \eqref{proof:linfdual=l1_step1_mid1_adaption} and obtain
    \begin{equation}\label{proof:linfdual=l1_step1_mid2_adaption}
    \begin{aligned}
        &\inf_{x\in\mathcal{A}}f(x)+\langle\lambda^1+\lambda^s,g(x)\rangle\\
        \le
        &f(u)+\langle\lambda^1,g(u)\rangle,\quad\forall u\in \mathcal{A}.
    \end{aligned}
    \end{equation}
Since $(\lambda^1,\lambda^s)$ in \eqref{proof:linfdual=l1_step1_mid2_adaption} is arbitrarily chosen from $(\ell^1_+,\ell^{s}_+)$, we can conclude \eqref{proof:linfdual=l1_step1_adaption} from \eqref{proof:linfdual=l1_step1_mid2_adaption}.

    \item  For step 2, we use the Yosida--Hewitt decomposition to finish the proof. According to Lemma \ref{lem:Yosida}, any $\lambda\in\ell_+^{\infty,*}$ can be decomposed as $\lambda=\lambda^1+\lambda^s$, where $\lambda^1\in\ell_+^1$ and $\lambda^s\in\ell_+^s$. It is then straightforward to conclude \eqref{equ:thm_linfdual=l1_adaption} from \eqref{proof:linfdual=l1_step1_adaption}.
\end{itemize}
\end{proof}
 The following theorem shows that the dual problem can be formulated with multipliers in $\Lambda$.

\begin{theorem}\label{dual_for}
    Under Assumption \ref{ass:dynamic}, the dual Lagrangian problem of \eqref{equ:CK_equiv} can be formulated as
    \begin{equation}\label{equ:dual_CK}d=\inf_{(\lambda_t^i(h^t))\in\Lambda}\sup_{(a(s^t))\in \tilde{\mathcal{A}}^{\infty}(x_0,s_0)}  L((a(s^t)),(\lambda^i(h^t));(\gamma^i),x_0,s_0),
    \end{equation}
    where $L$ is the Lagrangian functional defined in \eqref{equ:Lag_CK_refor}.
\end{theorem}

\begin{proof}
The problem \eqref{equ:CK_equiv} can be reformulated as
\begin{equation}\label{equ:CK_reformulate}
\begin{aligned}
    &\max_{(a(s^t))\in\mathcal{\tilde{A}}^{\infty}(x_0,s_0)\subset\mathcal{A}^{\infty}} f((a(s^t))),\\
    \textbf{s.t. }&g((a(s^t)))\ge 0,
\end{aligned}
\end{equation}
where $f$ and $g$ are defined in \eqref{equ:deff_lot} and \eqref{equ:defg_lot}.
It then suffices to verify all the conditions in Theorem \ref{thm:linfdual=l1_adaption}:
 Due to $0<\beta<1$, and the boundedness of $r$ and $g^i$, it is straightforward to verify that $f$ satisfies \eqref{cond:f_dual=l1_adaption}, $g$ is UANA, and  $g$ is bounded in $\ell^{\infty}$.
 For any $u=(u(s^t))\in \tilde{\mathcal{A}}^{\infty}(x_0,s_0)$, we denote $x^u_t(s^t)$ the state induced by $u$ and the exogenous shock $s^t$. For any  $T\in \mathbb{N}, s^T\in\mathcal{S}^T$, according to part 5 of Assumption \ref{ass:dynamic}, the problem with the initial state $(x^u_{T+1}(s^{T+1}),s_{T+1})$ has a feasible point. Therefore, there exists $v=(v(s^t))$, s.t. $v(s^{T+1},s^{t-T-1})_{t\ge T+1}\in \tilde{\mathcal{A}}^{\infty}(x_{T+1}^u(s^{T+1}),s_{T+1})$ for any $s^{T+1}\in\mathcal{S}^{T+1}$, such that when considering $a^T(u,v)=([a^T(u,v)](s^t))\in \tilde{\mathcal{A}}^{\infty}(x_0,s_0)$ defined as
    $$
    [a^T(u,v)](s^t)=\begin{cases}
        u(s^t),&t\le T;\\
        v(s^t),&t>T,
    \end{cases}
    $$
$(g(a^T(u,v)))_{t,\,h^t,\,i}\ge 0$ when $t\ge T+1$. Therefore, Theorem \ref{thm:linfdual=l1_adaption} applies.
\end{proof}
\subsection{\texorpdfstring{Slater's condition}{Slater's condition}}
We introduce a variant of Slater's condition. 
\begin{assumption}\label{ass:Slater}
    There exists $\epsilon>0$, s.t. for any $s_0\in\mathcal{S},\,x_0\in\mathcal{X}$, there exists a feasible point $(\psi(h^t))\in\Pi_{h^t\in\mathcal{H}^t}{\cal P}(\tilde{A}(h^t))$ to 
    \eqref{equ:simplified_CK_lot_statewise} satisfying for all $ i=1,\ldots,I$,
    \begin{equation}\label{eq:feas_Slater}
    \mathbb{E}_{s_0}^{(a(h^t)\sim \psi(h^t))}\left(\sum_{t=0}^{\infty}\beta^tg^i(x_{t},a_{t},s_t)-\bar{g}^i\right)\ge \epsilon.
    \end{equation}
\end{assumption}

\begin{theorem}\label{thm:exist_bdd_lag}
    Under Assumption \ref{ass:Slater}, for every $x\in\mathcal{X},\,s\in\mathcal{S},\,\gamma\in\mathbb{R}_+^I$, there exists $\lambda^*$ that solves the inf-sup problem \eqref{equ:recursive_dual}, with $D$ given as the dual value function.
\end{theorem}

\begin{proof}

    Since the objective in \eqref{equ:recursive_dual} is convex in $\lambda$, to prove the existence of a solution $\lambda^*$, it suffices to show that, there exists $C>0$, s.t. when $\lambda^i> C$ for some $i$, 
    {\footnotesize \begin{equation}\label{eq:proof_exist_bdd_lag}
\begin{aligned}
&\sup_{a\in\mathcal{A}}\Big[r(x,a,s)+\sum_{i=1}^{I}\gamma^ig^i(x,a,s)+\lambda^i(g^i(x,a,s)-\bar{g}^i)+\beta\mathbb{E}_sD(\gamma+\lambda,x',s')\Big]>D(\gamma,x,s).
\end{aligned}
    \end{equation}}
    
    According to Theorem \ref{thm:lot_dual_equiv}, we know that the LHS of \eqref{eq:proof_exist_bdd_lag} is not less than
    {\small $$
    \begin{aligned}
    &\mathbb{E}_{s_0}^{a(h^t)\sim\psi(h^t)} \sum_{t=0}^{\infty}\beta^tr(x_t,a_t,s_t)+\sum_{i=1}^{I}\gamma^i\mathbb{E}_{s_0}^{a(h^t)\sim\psi(h^t)} \sum_{t=0}^{\infty}\beta^tg^i(x_t,a_t,s_t)\\
&\qquad+\sum_{i=1}^{I}\lambda^i\mathbb{E}_{s_0}^{a(h^t)\sim\psi(h^t)} \left(\sum_{t=0}^{\infty}\beta^tg^i(x_t,a_t,s_t)-\bar{g}^i\right)\\
    \ge &-\left(1+\sum_{i=1}^I\gamma^i\right)\frac{\max\{\|r\|_{\infty},\|g\|_{\infty}\}}{1-\beta}+\sum_{i=1}^{I}\lambda^i\epsilon,
    \end{aligned}
    $$}where $(\psi(h^t))$ is chosen to be the feasible point satisfying \eqref{eq:feas_Slater}. It is then straightforward to verify that, taking 
    $$
    C=2\frac{D(\gamma,x,s)+\left(1+\sum_{i=1}^{I}\gamma^i\right)\frac{\max\{\|r\|_{\infty},\|g\|_{\infty}\}}{1-\beta}}{\epsilon},
    $$
    when $\lambda^i>C$ for some $i$, \eqref{eq:proof_exist_bdd_lag} is satisfied.

\end{proof}

\section{Relation to \cite{bloise2022negishi} (Online Only)} \label{Bloise}
We consider a one-principal, one-agent problem
$$
     \begin{aligned}
			&V(x_0,s_0,\mu_0)=\sup_{\{a_t\}}\mathbb{E}_0\sum_{t=0}^{\infty}\beta^t(u(x_t,a_t,s_t)+\mu _0v(x_t,a_t,s_t))\label{pp}\\
		\textbf{s.t. }&x_{t+1}=\zeta(x_t,a_t,s_{t}),\quad p(x_t,a_t,s_t)\ge 0,\quad \forall t\ge 0,\,s^t\in\mathcal{S}^t\label{dynamic_constraint}\\
		&\mathbb{E}_t\sum_{n=0}^{\infty}\beta^nv(x_{t+n},a_{t+n},s_{t+n})\ge \bar{v},\quad\forall t\ge 0,\,s^t\in\mathcal{S}^t.\label{pp_lookingforward}
	\end{aligned}
    $$
   We define
    $$
    W^1(\mathbb{E}_su',x,a,s)=u(x,a,s)+\beta \mathbb{E}_su',\, W^2(\mathbb{E}_sv',x,a,s)=v(x,a,s)+\beta \mathbb{E}_sv'.
    $$
    Applying the Negishi method of \cite{bloise2022negishi}, we obtain the following recursive formula:
    {\small $$
    \begin{aligned}
&J(\theta^1,\theta^2,x,s)=\sup_{a,u',v'}\theta^1W^1(\mathbb{E}_su',x,a,s)+\theta^2W^2(\mathbb{E}_sv',x,a,s)),\\
    \textbf{s.t. }&p(x,a,s)\ge 0, x'=\zeta(x,a,s),\\
    &v(x,a,s)+\beta \mathbb{E}_{s}v'\ge \bar{v},\\
    &(u'(x',s'),v'(x',s'))\in\\&(\mathcal{U}(x',s'))=\{\tilde\theta^1u'(x',s')+\tilde\theta^2v'(x',s')\le J(\tilde{\theta}^1,\tilde\theta^2,x',s'),\,\forall \tilde\theta^1,\tilde\theta^2\ge 0\}
    \end{aligned}
    $$}

    It is straightforward to see that $\mathcal{U}(x',s')$ is closed and convex for any $x'\in\mathcal{X},\,s'\in\mathcal{S}$. Furthermore, $J$ is homogeneous of degree 1 w.r.t. $(\theta^1,\theta^2)$, i.e., $J(t\theta^1,t\theta^2,x,s)=tJ(\theta^1,\theta ^2,x,s)$ for all $t\ge 0$; and it is straightforward to see that $J$ is convex w.r.t. $(\theta_1,\theta_2)$. Fixing $\theta^1=1$, we define $D(\theta^2,x,s)=J(1,\theta^2,x,s)$. Then $D$ is convex w.r.t. $\theta^2$. The Negishi recursive formula for $D$ is
    {\small \begin{equation}\label{appD_NegishiD}
    \begin{aligned}
&D(\theta^2,x,s)=\sup_{a,u',v'}W^1(\mathbb{E}_su',x,a,s)+\theta^2W^2(\mathbb{E}_sv',x,a,s)),\\
    \textbf{s.t. }&p(x,a,s)\ge 0, x'=\zeta(x,a,s),\\
    &v(x,a,s)+\beta \mathbb{E}_{s}v'\ge \bar{v},\\
    &(u'(x',s'),v'(x',s'))\in\\&(\mathcal{U}(x',s'))=\{u'(x',s')+\tilde{\theta}^2v'(x',s')\le D(\tilde{\theta}^2,x',s'),\,\forall \tilde{\theta}^2\ge 0\}
    \end{aligned}
    \end{equation}}
    We denote the set of all feasible $(a,u',v')$ by $\mathcal{F}(D)$. Problem \eqref{appD_NegishiD} can be rewritten as
     \begin{equation}\label{D_rec}
    D(\theta^2,x,s)=\sup_{(a,u',v')\in\mathcal{F}(D)}F(a,u',v'):={u(x,a,s)+\theta^2v(x,a,s)+\beta \mathbb{E}_su'+\theta^2\beta\mathbb{E}_sv'}
    \end{equation}
    
    \begin{enumerate}
    \item \textbf{We show that for any $D$ that is convex and continuous in $\theta^2$,  
    {\footnotesize \begin{equation}\label{appD_Negishi_step1}
\begin{aligned}
&\sup_{(a,u',v')\in\mathcal{F}(D)}F(a,u',v')\le \sup_{a\in\tilde{\mathcal{A}}(x,s)}\inf_{\lambda\ge0}u(x,a,s)\\
&\qquad+\theta^2(v(x,a,s)+\lambda(v(x,a,s)-\bar{v})+\beta\mathbb{E}_sD(\theta^2+\lambda,x',s').
\end{aligned}
    \end{equation}}}
    Note that, for all $(a,u',v')\in\mathcal{F}(D)$ and
    $\lambda\ge 0$, we have
    {\small $$
    \begin{aligned}
    F(a,u',v')&=u(x,a,s)+\theta^2v(x,a,s)+\beta \mathbb{E}_su'+\theta^2\beta\mathbb{E}_sv'\\
    &\le u(x,a,s)+\theta^2v(x,a,s)+\lambda(v(x,a,s)+\beta\mathbb{E}_sv'-\bar{v})+\beta\mathbb{E}_su'+\theta^2\beta\mathbb{E}_sv'
    \\
    &= u(x,a,s)+\theta^2v(x,a,s)+\lambda(v(x,a,s)-\bar{v})+\beta\mathbb{E}_s(u'+(\theta^2+\lambda)v')\\
    &\le u(x,a,s)+\theta^2v(x,a,s)+\lambda(v(x,a,s)-\bar{v})+\beta\mathbb{E}_sD(\theta^2+\lambda,x',s').
    \end{aligned}
    $$}
    Since $\lambda$ is arbitrarily chosen, we have
    $$
    F(a,u',v')\le \inf_{\lambda}u(x,a,s)+\theta^2(v(x,a,s)+\lambda(v(x,a,s)-\bar{v})+\beta\mathbb{E}_sD(\theta^2+\lambda,x',s'),
    $$
    for all $(a,u',v')\in\mathcal{F}$, implying that
    {\small $$
\begin{aligned}
&\sup_{(a,u',v')\in\mathcal{F}(D)}F(a,u',v')\le \sup_{(a,u',v')\in\mathcal{F}}\inf_{\lambda}u(x,a,s)\\
&\qquad+\theta^2(v(x,a,s)+\lambda(v(x,a,s)-\bar{v})+\beta\mathbb{E}_sD(\theta^2+\lambda,x',s')
\end{aligned}
    $$}
    Therefore, \eqref{appD_Negishi_step1} is satisfied.
    \item \textbf{We show that for any $D$ that is convex and continuous in $\theta^2$,}
    {\footnotesize \begin{equation}\label{appD_Negishi_step2}
\begin{aligned}
&\sup_{(a,u',v')\in\mathcal{F}(D)}F(a,u',v')\ge \sup_{a\in\tilde{\mathcal{A}}(x,s)}\inf_{\lambda\ge0}u(x,a,s)\\
&\qquad+\theta^2(v(x,a,s)+\lambda(v(x,a,s)-\bar{v})+\beta\mathbb{E}_sD(\theta^2+\lambda,x',s').
\end{aligned}
    \end{equation}}
     By the convexity of $D$, for any $x'\in\mathcal{X},\,s'\in\mathcal{S},\,\lambda\ge 0$, for each $v'_{\lambda}(x',s')\in\partial D(\lambda,x',s')$, there exists $u'_{\lambda}(x',s')$, s.t. $(u'_{\lambda}(x',s'),v'_{\lambda}(x',s'))\in\mathcal{U}(x',s')$, and $D(\lambda,x',s')=u_{\lambda}'(x',s')+\lambda v'_{\lambda}(x',s')$. For any $a\in\tilde{\mathcal{A}}(x,s)$, we have
    {\footnotesize $$
    \begin{aligned}
&\inf_{\lambda}u(x,a,s)+\theta^2v(x,a,s)+\lambda(v(x,a,s)-\bar{v})+\beta\mathbb{E}_sD(\theta^2+\lambda,x',s')\\
=&\inf_{\lambda} u(x,a,s)+\theta^2v(x,a,s)+\lambda(v(x,a,s)-\bar{v})+\beta\mathbb{E}_s(u'_{\theta_2+\lambda}+(\theta^2+\lambda)v'_{\theta_2+\lambda}))\\
=&\inf_{\lambda} u(x,a,s)+\theta^2v(x,a,s)+\lambda(v(x,a,s) +\beta\mathbb{E}_sv'_{\theta_2+\lambda}-\bar{v})+\beta\mathbb{E}_s(u'_{\theta_2+\lambda}+\theta^2v'_{\theta_2+\lambda})).
\end{aligned}
    $$}
If there exists $(a,u'_{\theta^2},v'_{\theta^2})\in\mathcal{F}(D)$, then 
{\footnotesize \begin{equation}\label{appD_Negishi_step2mid}
\begin{aligned}
    &\inf_{\lambda} u(x,a,s)+\theta^2v(x,a,s)+\lambda(v(x,a,s)+\beta\mathbb{E}_sv'_{\theta_2+\lambda}-\bar{v})+\beta\mathbb{E}_s(u'_{\theta_2+\lambda}+\theta^2v'_{\theta_2+\lambda}))\\
    \le &u(x,a,s)+\theta^2v(x,a,s)+\beta\mathbb{E}_s(u'_{\theta^2}+\theta^2v'_{\theta^2})\le \sup_{(a,u',v')\in\mathcal{F}(D)}F(a,u',v').
\end{aligned}
\end{equation}}
If $(a,u'_{\theta^2},v'_{\theta^2})\not\in\mathcal{F}(D)$, then one of the following happens:
\begin{itemize}
    \item For all $s'\in\mathcal{S}$, $(u'(\zeta(x,a,s),s'),v'(\zeta(x,a,s),s'))\in\mathcal{U}(\zeta(x,a,s),s')$, $(a,u',v')\not\in\mathcal{F}$. According to the closedness of $\mathcal{U}$, if $\lim_{\lambda\rightarrow\infty}u'_{\theta_2+\lambda}>-\infty$\footnote{When $\lim_{\lambda\rightarrow\infty}u'_{\theta_2+\lambda}=-\infty$, it is straightforward to verify \eqref{eq:bloise_corner}.}, then there exists $\underline{L}> 0$, s.t.
    $$
v(x,a,s)+\beta\mathbb{E}_sv'_{\theta_2+\lambda}-\bar{v}\le  -\underline{L},
    $$
    implying that
    {\small \begin{equation}\label{eq:bloise_corner}
\begin{aligned}
&\inf_{\lambda} u(x,a,s)+\theta^2v(x,a,s)+\lambda(v(x,a,s)+\beta\mathbb{E}_sv'_{\theta_2+\lambda}-\bar{v})\\
&\qquad+\beta\mathbb{E}_s(u'_{\theta_2+\lambda}+\theta^2v'_{\theta_2+\lambda}))=-\infty.
\end{aligned}
    \end{equation}}
    \item We can find a $\lambda^*>0$, s.t. $v'_{\theta_2+\lambda^*}$ binds the incentive constraint $v(x,a,s)+\beta\mathbb{E}_sv'_{\theta_2+\lambda^*}-\bar{v}=0$ due to the convexity of $D$. Similar arguments imply the same inequality \eqref{appD_Negishi_step2mid}.
\end{itemize}
Therefore, \eqref{appD_Negishi_step2} is satisfied.
\end{enumerate}
We combine \textbf{1.} and \textbf{2.} and see that the largest fixed point of \eqref{appD_NegishiD} coincides with the largest fixed point of the sup-inf formula.

\section{\texorpdfstring{Finite-horizon forward-looking constraints (Online Only)}{Finite-horizon forward-looking constraints (Online Only)}} \label{app:finho}
In this appendix, we consider incentive constraints that involve only finitely many periods; for simplicity, we focus on the two-period case. The incentive constraints in the deterministic problems are replaced by
$$
\mathbb{E}_{s_t}\sum_{n=0}^{1}\beta^ng^i(x_{t+n},a(s^{t+n}),s_{t+n})\ge \bar{g}^i,\quad \forall t\in\mathbb{N},\,\forall s\in \mathcal{S}^t,\,\forall i\in\{1,\cdots,I\},
$$
and the deterministic problem is\footnote{Note that the objective function is also different from problem \eqref{equ:CK}.}
{\footnotesize \begin{equation}\label{equ:CK_2p}
    \begin{aligned}
      \max_{(a_t(s^t))\in \mathcal{A}^{\infty}\subset \ell^{\infty}} &\sum_{i=1}^{I}\gamma^ig^i(x_0,a(s^0),s_0)+\mathbb{E}_{s_0}\sum_{t=0}^{\infty}\beta^tr(x_t,a_t(s^t),s_t)\\
      \textbf{s.t. }&\mathbb{E}_{s_t}\sum_{n=0}^{1}\beta^n g^i(x_{t+n},a_{t+n}(s^{t+n}),s_{t+n})\ge\bar{g}^i,\\
&\qquad \quad \forall t\in\mathbb{N},\,\forall s^t\in \mathcal{S}^t,\,\forall i\in\{1,\cdots,I\},\\
      \text{where }&x_{t+1}=\zeta(x_t,a_t(s^t),s_t)\text{ and }p(x_t,a_t(s^t),s_t)\ge 0,\quad\forall t\in \mathbb{N},\,s^t\in \mathcal{S}^t
    \end{aligned}
\end{equation}}
The Lagrangian is then defined as
{\footnotesize \begin{equation}\label{eq:Lag_CKextend_2p}
\begin{aligned}
    &L((a(s^t)),(\lambda^i(h^t));(\gamma^i),x_0,s_0)\\
=&\sum_{i=1}^{I}\gamma^ig^i(x_0,a(s^0),s_0)+\mathbb{E}_{s_0}\sum_{t=0}^{\infty}\beta^t\biggl(r(x_t,a(s^t),s_t)\biggr.\\
&\left.+\lambda^i(s_0,a(s_0),\cdots,s_{t-1},a(s^{t-1}),s_t)\left(\sum_{n=0}^{1}\beta^n g^i(x_{t+n},a(s^{t+n}),s_{t+n})-\bar{g}^i\right)\right).\\
\end{aligned}
\end{equation}}
\begin{theorem}\label{thm:recursive_extend}
    For any $x\in \mathcal{X},\,s\in \mathcal{S},\,\gamma\in\mathbb{R}_+^{I}$, the dual value function $D(\gamma,x,s)$ defined as
    $$
D(\gamma,x,s_0):=\inf_{(\lambda^i(h^t))\in\Lambda}\sup_{(a(s^t))\in\tilde{\mathcal{A}}^{\infty}}L((a(s^t)),(\lambda^i(h^t));(\gamma^i),x_0,s_0)
    $$
    satisfies the following recursive equation
    {\footnotesize \begin{equation}\label{equ:recursive_dual_extend55}
    \begin{aligned}
        &D(\gamma,x,s)=\inf_{\lambda\in \mathbb{R}_+^I}\sup_{a\in\tilde{\mathcal{A}}(x,s)}\Big[\left(r(x,a,s)+\sum_{i=1}^{I}\left( \gamma^ig^i(x,a,s)+\lambda^i(g^i(x,a,s)-\bar{g}^i)\right)\right)+\beta\mathbb{E}_{s}D(\lambda,x',s')\Big],\\
        &\text{where }x'=\zeta(x,a,s).
    \end{aligned}
    \end{equation}}
\end{theorem}
\begin{proof}
By straightforward algebraic simplification, we have
  {\footnotesize \begin{equation}
\begin{aligned}
    &L((a(s^t)),(\lambda^i(h^t));(\gamma^i),x_0,s_0)\\
=&\sum_{i=1}^{I}\gamma^ig^i(x_0,a(s^0),s_0)+\mathbb{E}_{s_0}\sum_{t=0}^{\infty}\beta^t\biggl(r(x_t,a(s^t),s_t)\biggr.\\
&\left.+\sum_{i=1}^I \lambda_t^i(s_0,a(s_0),\cdots,s_{t-1},a(s^{t-1}),s_t)\left(\sum_{n=0}^{1}\beta^n g^i(x_{t+n},a(s^{t+n}),s_{t+n})-\bar{g}^i\right)\right)\\
=&\left(r(x_0,a_0,s_0)+\sum_{i=1}^{I} \left( \gamma^ig^i(x_0,a_0,s_0)+\lambda^i_0(g^i(x_0,a_0,s_0)-\bar{g})\right) \right)\\
&+\beta\mathbb{E}_{s_0}\left(\sum_{i=1}^{I}\lambda_0^ig^i(x_1,a_1,s_1)+\mathbb{E}_{s_1}\sum_{t=1}^{\infty}\beta^{t-1}\biggl(r(x_t,a(s^t),s_t)\right.\\
&+\left.\left. \sum_{i=1}^I\lambda_t^i(s_0,a(s_0),\cdots,s_{t-1},a(s^{t-1}),s_t)\left(\sum_{n=0}^{1}\beta^n g^i(x_{t+n},a(s^{t+n}),s_{t+n})-\bar{g}^i\right)\right)\right).
\end{aligned}
\end{equation}}  
The rest of the proof follows the same procedure as in the proof of Theorem \ref{thm:recursive1}, and we omit the details here.
\end{proof}

\section{\texorpdfstring{Existence of a solution to the FE (Online Only)}{Existence of a solution to the FE (Online Only)}}
\label{app:fixedpoint}
 We now establish the existence of solutions to both functional equations. Although both inf-sup and sup-inf Bellman operators exhibit monotonicity and discounting properties, the unbounded nature of Lagrange multipliers and the resulting unboundedness of the dual value function preclude the direct application of standard contraction arguments from \cite{stokey1989recursive}. 
\cite{bloise2022negishi} demonstrate the existence of fixed points for the Negishi operator using Tarski's fixed point theorem\footnote{However, Tarski's fixed point theorem is generally non-constructive when the continuity of the Negishi operator remains unverified. For instance, $f:[0,1]\rightarrow[0,1]$ defined as
$$
f(x)=\begin{cases}
    \frac{1}{2}+(x-\frac{1}{2})^2,&\frac{1}{2}<x\le 1\\
    x^2,&0\le x\le 1/2
\end{cases}
$$ is a monotone operator, and the existence of a fixed point is guaranteed by Tarski's fixed point theorem. However, $f^{(\infty)}(1)=\frac{1}{2}$ is not a fixed point of $f$, due to the discontinuity of $f$ at $x=\frac{1}{2}$.}. Alternatively, \cite{pavoni2018dual} employ Thompson's metric to establish contraction, requiring the additional technical assumption that there exists some function $F$ satisfying $\mathcal{B}(F) > F + \|\gamma\|\epsilon$. The framework achieves a contraction factor of $1 - O(\epsilon)$ (under Thompson's metric).
We develop a constructive proof of the existence of fixed points that directly leverages the monotonicity structure of the Bellman operator. Our method provides an explicit iterative scheme that converges to the solution while avoiding the additional assumptions required in \cite{pavoni2018dual}.
 
Throughout, we require the value function to be Lipschitz continuous with respect to the multiplier, and we define the following upper bound on the Lipschitz constant: $$L=\frac{\|r\|_{\infty}+\sum_{i=1}^{I}\|g_i\|_{\infty}}{1-\beta}.$$
\begin{definition}\label{def:function_space}
For any $k>0$, we define 
$
B(k)=\{\gamma\in\mathbb{R}_+^I:\|\gamma\|_{\infty}\le k\}.
$ We define the function space
{\footnotesize $$
\begin{aligned}
    \mathcal{M}:=\{F\in C(\mathbb{R}_+^{I}\times \mathcal{X}\times\mathcal{S};\mathbb{R}):
    & \text{(i)} F(\cdot,x,s)\in L^{\infty}(B(k)),\,\forall k\in\mathbb{N}_+,x\in\mathcal{X},\,s\in\mathcal{S}; \\
    &\text{(ii) }\sum_{i=1}^{|\mathcal{S}|}\frac{1}{2^i}\left[\sum_{j=1}^{\infty}\frac{1}{2^j}\left(\sum_{k=1}^{\infty}\frac{1}{2^{k}}\|F(\cdot,x_j,s_i)\|_{L^{\infty}(B(k))}\right)\right]<\infty\}.
\end{aligned}
$$}
For any $F\in\mathcal{M}$, the norm\footnote{Convergence in this norm implies local uniform convergence.} of $F$ is defined by 
$$
\|F\|_{\mathcal{M}}:=\sum_{i=1}^{|\mathcal{S}|}\frac{1}{2^i}\left[\sum_{j=1}^{\infty}\frac{1}{2^j}\left(\sum_{k=1}^{\infty}\frac{1}{2^{k}}\|F(\cdot,x_j,s_i)\|_{L^{\infty}(B(k))}\right)\right].
$$
\end{definition}
It is easy to verify that the space $(\mathcal{M},\|\cdot\|_{\mathcal{M}})$ is a complete metric space.
\begin{definition}\label{def:NinM}
    We define the subset $\mathcal{N}\subset\mathcal{M}$ as follows
   \begin{align}
    \mathcal{N}:=&\{F\in \mathcal{M}: \notag \\
    &\text{(i) }F(\cdot,x,s) \text{ is convex};\text{(ii)}F(\cdot,x,s) \text{ is $L$-Lipschitz};\\&\text{(iii)} F(\gamma,x,s)\ge v^0+\sum_{i=1}^{I}\gamma^iv^i, \label{eq:F_inequality} \\
    &\text{where} 
    \begin{cases}
        v^0=\mathbb{E}_{s_0=s}\sum_{t=0}^{\infty}\beta^tr(x_t,a(s^t),s_t); \\
        v^i=\mathbb{E}_{s_0=s}\sum_{t=0}^{\infty}\beta^tg^i(x_t,a(s^t),s_t),\quad i\in\{1,\cdots,I\},
    \end{cases} \notag \\
    &\text{for any feasible }(a_t(s^t))_{t\ge 0} \text{ to problem \eqref{equ:CK} with }x_0=x,\,s_0=s\\  &\text{(iv)}F(\gamma,x,s)\le (1+\sum_{i=1}^I\gamma^i)L\}. \label{eq:F_inequality_upper}
\end{align}
\end{definition}
Note that (iii) implies that 
    $
    F(\gamma,x,s)\ge -(1+\sum_{i=1}^{I}\gamma^i)L.         
    $
It is straightforward to verify that 
    $\mathcal{N}$ is a closed subset of $\mathcal{M}$.
We define the Bellman operator as follows.
\begin{definition}\label{def:bellman_operator}
    Let $\mathcal{F}$ denote the set of functions $F:\mathbb{R}_+^{I}\times\mathcal{X}\times \mathcal{S}\rightarrow\mathbb{R}$. The Bellman operator $\mathcal{B}:\mathcal{N}\rightarrow\mathcal{F}$ is defined by
    {\small $$
\begin{aligned}
&    \mathcal{B}(F)(\gamma,x,s):=\inf_{\lambda\in \mathbb{R}_+^I}\sup_{a\in\tilde{\mathcal{A}}(x,s)}\Big[\left(r(a,x,s)+\sum_{i=1}^{I} \left( \gamma^ig^i(a,x,s)+\lambda^i(g^i(a,x,s)-\bar{g}^i)\right)\right)\\
&\qquad+\beta\mathbb{E}_{s}F(\gamma+\lambda,x',s')\Big]
\end{aligned}
    $$}
    where $x'=\zeta(x,a,s)$.
\end{definition}
The following monotonicity property of $\mathcal{B}$ can be verified directly:
\begin{lemma}\label{lem:monotonicity}
    For $F_1,\,F_2\in\mathcal{N}$, s.t. $F_1(x,\gamma,s)\le F_2(x,\gamma,s),\,\forall (x,\gamma,s)\in\mathcal{X}\times\mathbb{R}_+^{I}\times \mathcal{S}$, we have $\mathcal{B}(F_1)(x,\gamma,s)\le \mathcal{B}(F_2)(x,\gamma,s),\,\forall (x,\gamma,s)\in \mathcal{X}\times\mathbb{R}_+^{I}\times \mathcal{S}$.
\end{lemma}
Moreover, the next lemma shows that $\mathcal{B}$ maps $\mathcal{N}$ into itself.
\begin{lemma}\label{lem:Bellman_restriction} Under Assumption \ref{ass:dynamic}, the Bellman operator $\mathcal{B}$ maps $\mathcal{N}$ to itself, that is,
    $$\mathcal{B}(\mathcal{N})\subset\mathcal{N}.$$
\end{lemma}

\begin{proof}

    \begin{itemize}
        \item \textbf{Show that $\mathcal{B}$ preserves convexity.} See the proof of Lemma 2 in \cite{pavoni2018dual}.
        \item \textbf{Show that $\mathcal{B}$ preserves $L$-Lipschitz continuity.} Assume that $F$ is $L$-Lipschitz. Given $\gamma_1,\,\gamma_2\in \mathbb{R}_+^I$, we have
        { $$\footnotesize
        \begin{aligned}
        &\mathcal{B}(F)(x,\gamma_1,s)\\=&\inf_{\lambda\in \mathbb{R}_+^I}\sup_{a\in\tilde{\mathcal{A}}(x,s)}\left[\left(r(a,x,s)+\sum_{i=1}^I\left(\gamma_1^ig^i(a,x,s)+\lambda^i(g^i(a,x,s)-\bar{g}^i)\right)\right)+\beta\mathbb{E}_sF(\gamma_1+\lambda,x',s')\right]\\
        =&\inf_{\lambda\in \mathbb{R}_+^I}\sup_{a\in\tilde{\mathcal{A}}(x,s)}\left[\left(r(a,x,s)+\sum_{i=1}^I\left(\gamma_2^ig^i(a,x,s)+\lambda^i(g^i(a,x,s)-\bar{g}^i)\right)\right)+\beta\mathbb{E}_sF(\gamma_2+\lambda,x',s')\right.\\&\left.+\sum_{i=1}^{I}(\gamma_1^i-\gamma_2^i)g^i(a,x,s)+\beta\mathbb{E}_s(F(\gamma_1+\lambda,x',s')-F(\gamma_2+\lambda,x',s'))\right].
        \end{aligned}
        $$}
        Since 
        {\small $$
\begin{aligned}
&\sum_{i=1}^{I}(\gamma_1^i-\gamma_2^i)g^i(a,x,s)+\beta\mathbb{E}_s(F(\gamma_1+\lambda,x',s')-F(\gamma_2+\lambda,x',s'))\\
        \le& \|\gamma_1-\gamma_2\|_{1}\max_{i}\|g^i\|_{\infty}+\beta L\|\gamma_1-\gamma_2\|_{1}\le L\|\gamma_1-\gamma_2\|_{1},
\end{aligned}
        $$}
        we have
        { $$\footnotesize
         \begin{aligned}
        &\mathcal{B}(F)(x,\gamma_1,s)\\
        =&\inf_{\lambda\in \mathbb{R}_+^I}\sup_{a\in\tilde{\mathcal{A}}(x,s)}\left[\left(r(a,x,s)+\sum_{i=1}^I \left( \gamma_2^ig^i(a,x,s)+\lambda^i(g^i(a,x,s)-\bar{g}^i)\right)\right)+\beta\mathbb{E}_sF(\gamma_2+\lambda,x',s')\right.\\&\left.+\sum_{i=1}^{I}(\gamma_1^i-\gamma_2^i)g^i(a,x,s)+\beta\mathbb{E}_s(F(\gamma_1+\lambda,x',s')-F(\gamma_2+\lambda,x',s'))\right]\\
        \le&\inf_{\lambda\in \mathbb{R}_+^I}\sup_{a\in\tilde{\mathcal{A}}(x,s)}\left[\left(r(a,x,s)+\sum_{i=1}^I\left(\gamma_2^ig^i(a,x,s)+\lambda^i(g^i(a,x,s)-\bar{g}^i)\right)\right)\right.\\
&+\beta\mathbb{E}_sF(\gamma_2+\lambda,x',s')+L\|\gamma_1-\gamma_2\|_{1}\Bigg]=\mathcal{B}(F)(x,\gamma_2,s)+L\|\gamma_1-\gamma_2\|_{1}
        \end{aligned}
        $$}
        Hence $\mathcal{B}$ preserves $L$-Lipschitz continuity.
        \item \textbf{Show that $\mathcal{B}$ preserves the property \eqref{eq:F_inequality}.}
        Assume that $F$ satisfies \eqref{eq:F_inequality}. Given $x\in \mathcal{X},\,\gamma\in\mathbb{R}_+^I
        ,\,s\in\mathcal{S}$. Assume that $(a(s^t))_{t\ge 0}$ is a feasible control to problem \eqref{equ:CK} with $x_0=x$, $s_0=s$, satisfying
        $$
        \begin{cases}
            v^0=\mathbb{E}_{s_0}\sum_{t=0}^{\infty}\beta^tr(x_t,a(s^t),s_t);\\
            v^i=\mathbb{E}_{s_0}\sum_{t=0}^{\infty}\beta^tg^i(x_t,a(s^t),s_t).
        \end{cases}
        $$ 
        Let $x'=\zeta(a(s_0),x_0,s_0)$. It is straightforward to verify that 
        $
        a(s^t|(s,s'))
        $
        is a feasible control to problem \eqref{equ:CK} with $x_0=x'$, $s_0=s'$. Therefore, for any $\lambda\in\mathbb{R}_+^I$, we have
        {\small $$
        \begin{aligned}
        &\sup_{a\in\mathcal{A}}\left[\left(r(a,x,s)+\sum_{i=1}^{I}\left(\gamma^ig^i(a,x,s)+\lambda^i(g^i(a,x,s)-\bar{g}^i)\right)\right)+\beta\mathbb{E}_{s}F(\gamma+\lambda,x',s')\right]\\
        \ge & r(a(s_0),x_0,s_0)+\sum_{i=1}^{I}\left(\gamma^ig^i(a(s_0),x_0,s_0)+\lambda^i(g^i(a(s_0),x_0,s_0)-\bar{g}^i)\right)\\
&+\beta\mathbb{E}_{s_0}\sum_{t=1}^{\infty}\beta^{t-1}r(x_t,a(s^t),s_t)+\sum_{i=1}^{I}(\gamma^i+\lambda^i)\beta^{t-1}g^i(x_t,a(s^t),s_t)\\
        =&v^0+\sum_{i=1}^{I}\left(\gamma^i v^i+\lambda^i(v^i-\bar{g}^i)\right)
        \ge v^0+\sum_{i=1}^{I}\gamma^iv^i,
        \end{aligned}
        $$}where the first inequality holds because $F$ satisfies \eqref{eq:F_inequality}. Hence $\mathcal{B}(F)(\gamma,x,s)\ge v^0+\sum_{i=1}^{I}\gamma^iv^i$. Since $x\in\mathcal{X},\,\gamma\in\mathbb{R}_+^I,\,s\in\mathcal{S}$, and the feasible control $(a(s^t))_{t\ge 0}$ with $x_0=x,\,s_0=s$ are chosen arbitrarily, it is then straightforward to conclude that $\mathcal{B}(F)$ satisfies \eqref{eq:F_inequality}.
        \item \textbf{Show that $\mathcal{B}$ preserves the property \eqref{eq:F_inequality_upper}. }  By the definition of $\mathcal{B}$, we have
        { $$\footnotesize
        \begin{aligned}
             &\mathcal{B}(F)(x,\gamma,s)\\=&\inf_{\lambda\in \mathbb{R}_+^I}\sup_{a\in\tilde{\mathcal{A}}(x,s)}\left[\left(r(a,x,s)+\sum_{i=1}^I\left(\gamma^ig^i(a,x,s)+\lambda^i(g^i(a,x,s)-\bar{g}^i)\right)\right)+\beta\mathbb{E}_sF(\gamma+\lambda,x',s')\right]\\
             \le &\sup_{a\in\tilde{\mathcal{A}}(x,s)}\left[\left(r(a,x,s)+\sum_{i=1}^I\gamma^ig^i(a,x,s)\right)+\beta\mathbb{E}_sF(\gamma,x',s')\right]\\
             \le & \|r\|_{\infty}+\sum_{i=1}^{I}\gamma^i\|g_i\|_{\infty}+\beta \left(1+\sum_{i=1}^{I}\gamma^i\right)L\le \left(1+\sum_{i=1}^I\gamma^i\right)L,
        \end{aligned}
        $$}where the last step uses the identity $ (1-\beta) L=\|r\|_{\infty}+\sum_i \|g_i \|_{\infty}$ together with non-negativity.
        \item \textbf{Show that $\mathcal{B}$ preserves the finiteness of the norm $\|\cdot\|_{\mathcal{M}}$.}  This is a direct corollary of the fact that $\mathcal{B}(F)$ preserves the properties \eqref{eq:F_inequality} and \eqref{eq:F_inequality_upper}.\hfill $ \Box $
  
    \end{itemize}
\end{proof}
With this we can state the following theorem.
\begin{theorem}\label{thm:contraction}
Assume that $F_0(\gamma,x,s)=\left(1+\sum_{i=1}^{I}\gamma^i\right)L$. Then there exists $F^*\in\mathcal{N}$, such that
$$
\|\mathcal{B}^{(n)}(F_0)-F^*\|_{\mathcal{M}}\rightarrow0,\text{ as $n\rightarrow\infty$},
$$
and $F^*$ is the largest fixed point of $\mathcal{B}$ in $\mathcal{N}$, i.e.,\begin{enumerate}
    \item $\mathcal{B}(F^*)=F^*$, and
    \item For any $F\in\mathcal{N}$ s.t. $\mathcal{B}(F)=F$, we have $F\le F^*$.
    \end{enumerate}
    \end{theorem}
\begin{proof}
    \begin{itemize}
        \item \textbf{Step 1. Show that there exists a pointwise limit of $\mathcal{B}^{(n)}(F_0)$ in $\mathcal{N}$.} According to Lemma \ref{lem:Bellman_restriction}, since $F_0\in \mathcal{N}$, we know that
        $
        \mathcal{B}(F_0)\in \mathcal{N},
        $
        implying that $\mathcal{B}(F_0)\le F_0$.
        By induction and monotonicity we have
        $
        \mathcal{B}^{(n)}(F_0)\le \mathcal{B}^{(n-1)}(F_0),\,\forall n\in\mathbb{N}_+. 
        $
        Therefore, for any $\gamma\in\mathbb{R}_+^I,\,x\in\mathcal{X},\,s\in\mathcal{S}$, the sequence
        $
        \{\mathcal{B}^{(n)}(F_0)(\gamma,x,s)\}_{n\in\mathbb{N}_+}
        $
        is a decreasing, bounded sequence, and hence has a limit $F^*(\gamma,x,s)$. It is straightforward to verify that $F^*\in \mathcal{N}$.
        \item \textbf{Step 2. Show that $\mathcal{B}^{(n)}(F_0)$ converges to $F^*$ in norm $\|\cdot\|_{\mathcal{M}}$.}  For any $x\in\mathcal{X},\,s\in\mathcal{S}$, we know that
        $
        \mathcal{B}^{(n)}(F_0)(\gamma,x,s) 
        $
        is $L$-Lipschitz and uniformly bounded on any compact ball $B(k)$; hence, by the Arzel\`a--Ascoli theorem, there exists a subsequence $$
        \begin{aligned}
        &\{\mathcal{B}^{(t_n)}(F_0)(\gamma,x,s) \}\subset \{\mathcal{B}^{(n)}(F_0)(\gamma,x,s)\},\\
        &\mbox{ s.t. }
        \mathcal{B}^{(t_n)}(F_0)(\cdot,x,s)\rightrightarrows \tilde{F}_k(\cdot,x,s), \text{ as $n\rightarrow\infty$ in $B(k)$,}
        \end{aligned}
        $$
        for some $\tilde{F}_{k}(\cdot, x,s)$, where $ \rightrightarrows $ denotes uniform convergence.  From step 1, we have 
        $
        \lim_{n\rightarrow\infty}\mathcal{B}^{(t_n)}(F_0)(\gamma,x,s)=F^*(\gamma,x,s),\,\forall \gamma\in B(k).
        $
        Therefore, we have $\tilde{F}_k(\gamma,x,s)=F^*(\gamma,x,s),\,\forall \gamma\in B(k)$, and hence
         $$
        \mathcal{B}^{(t_n)}(F_0)(\cdot,x,s)\rightrightarrows F^*(\cdot,x,s), \text{ as $n\rightarrow\infty$ in $B(k)$}.
        $$
        According to the monotonicity, the whole sequence uniformly converges, i.e.,
        $$
         \mathcal{B}^{(n)}(F_0)(\cdot,x,s)\rightrightarrows F^*(\cdot,x,s), \text{ as $n\rightarrow\infty$ in $B(k)$},
        $$
        yielding that
        \begin{equation}\label{eq:needtouse1}
        \begin{aligned}
        &\|\mathcal{B}^{(n)}(F_0)(\cdot,x_j,s_i)-F^*(\cdot,x_j,s_i)\|_{L^{\infty}(B(k))}\rightarrow0,\\
        &\qquad \forall s_i\in\mathcal{S},\,x_j\in\mathcal{X},\,k\in\mathbb{N}_+.
        \end{aligned}
        \end{equation}
        Since functions in $\mathcal{N}$ are bounded by affine functions, there exists $C>0$, s.t.
         \begin{equation}\label{eq:needtouse2}
        \begin{aligned}
        &\|\mathcal{B}^{(n)}(F_0)(\cdot,x_j,s_i)-F^*(\cdot,x_j,s_i)\|_{L^{\infty}(B(k))}\le C(k+1),\\
        &\qquad \forall s_i\in\mathcal{S},\,x_j\in\mathcal{X},\,k\in\mathbb{N}_+.
        \end{aligned}
        \end{equation}
        Therefore, applying \eqref{eq:needtouse2} we have
        {$$\footnotesize
        \begin{aligned}
            &\|\mathcal{B}^{(n)}(F_0)-F^*\|_{\mathcal{M}}\\
            =&\sum_{i=1}^{|\mathcal{S}|}\frac{1}{2^i}\left[\sum_{j=1}^{\infty}\frac{1}{2^j}\left(\sum_{k=1}^{\infty}\frac{1}{2^{k}}\|\mathcal{B}^{(n)}(F_0)(\cdot,x_j,s_i)-F^*(\cdot,x_j,s_i)\|_{L^{\infty}(B(k))}\right)\right]\\
            \le&\sum_{i=1}^{|\mathcal{S}|}\frac{1}{2^i}\left\{\sum_{j=1}^{J}\frac{1}{2^j}\left[\left(\sum_{k=1}^{K}\frac{1}{2^{k}}\|\mathcal{B}^{(n)}(F_0)(\cdot,x_j,s_i)-F^*(\cdot,x_j,s_i)\|_{L^{\infty}(B(k))}\right)\right.\right.\\
&\qquad+\underbrace{\left(\sum_{k>K}\frac{(k+1)C}{2^k}\right)}_{I}\Bigg]\Bigg\}+\underbrace{\sum_{i=1}^{|\mathcal{S}|}\frac{1}{2^i} \sum_{j>J}\frac{1}{2^j}\sum_{k=1}^{\infty}\frac{(k+1)C}{2^k}}_{II}.
        \end{aligned}
        $$}For any $\epsilon>0$, there exists $J>0$ and $K>0$, s.t. $I<\epsilon/3,\,II<\epsilon/3$. Furthermore, according to \eqref{eq:needtouse1}, there exists $N\in\mathbb{N}_+$, s.t. when $n>N$, we have
        $$
        \sum_{i=1}^{|\mathcal{S}|}\frac{1}{2^i}\sum_{j=1}^{J}\frac{1}{2^j}\left[\left(\sum_{k=1}^{K}\frac{1}{2^{k}}\|\mathcal{B}^{(n)}(F_0)(\cdot,x_j,s_i)-F^*(\cdot,x_j,s_i)\|_{L^{\infty}(B(k))}\right)\right]<\frac{\epsilon}{3}.
        $$
        Therefore, when $n>N$, we have
        {\footnotesize $$
         \begin{aligned}
            &\|\mathcal{B}^{(n)}(F_0)-F^*\|_{\mathcal{M}}\\
            \le&\sum_{i=1}^{|\mathcal{S}|}\frac{1}{2^i}\left\{\sum_{j=1}^{J}\frac{1}{2^j}\left[\left(\sum_{k=1}^{K}\frac{1}{2^{k}}\|\mathcal{B}^{(n)}(F_0)(\cdot,x_j,s_i)-F^*(\cdot,x_j,s_i)\|_{L^{\infty}(B(k))}\right)+\frac{\epsilon}{3}\right]\right\}+\frac{\epsilon}{3}\\
            \le& \sum_{i=1}^{|\mathcal{S}|}\frac{1}{2^i}\sum_{j=1}^{J}\frac{1}{2^j}\left[\left(\sum_{k=1}^{K}\frac{1}{2^{k}}\|\mathcal{B}^{(n)}(F_0)(\cdot,x_j,s_i)-F^*(\cdot,x_j,s_i)\|_{L^{\infty}(B(k))}\right)\right]+\frac{2\epsilon}{3}<\epsilon
            \end{aligned}
        $$}
        Therefore $\|\mathcal{B}^{(n)}(F_0)-F^*\|_{\mathcal{M}}\rightarrow0$ by definition.

        \textbf{Step 3. Show that $\mathcal{B}(F^*)=F^*$}. According to monotonicity, we have
        $
        \mathcal{B}(F^*)\le \mathcal{B}^{(n)}(F_0),\,\forall n\in\mathbb{N}_+.
        $
        For any $\gamma\in\mathbb{R}_+^I,\,x\in\mathcal{X},\,s\in\mathcal{S}$, we take $n\rightarrow\infty$ and have 
        $
        \mathcal{B}(F^*)\le F^*.
        $
        Suppose that there exists $\ell>0$ such that $\mathcal{B}(F^*)(\gamma,x,s)=F^*(\gamma,x,s)-\ell$ for some $x\in\mathcal{X},\,s\in\mathcal{S},\,\gamma\in\mathbb{R}_+^I$. According to the definition of $\mathcal{B}$, we know that there exists $\lambda\in\mathbb{R}_+^I$, s.t.
        {\footnotesize \begin{equation}\label{eq:fixed_ineq}
\begin{aligned}
&\sup_{a\in\tilde{A}(x,s)}\Big[\left(r(a,x,s)+\sum_{i=1}^{I}\gamma^ig^i(a,x,s)+\lambda^i(g^i(a,x,s)-\bar{g}^i)\right)\\
&\qquad+\beta\mathbb{E}_sF^*(\gamma+\lambda,x',s')\Big]<F^*(\gamma,x,s)-\frac{\ell}{2}.
\end{aligned}
        \end{equation}}
        Since $\mathcal{B}^{(n)}(F_0)\rightarrow F^*$, there exists $N\in\mathbb{N}_+$, s.t.
        {\small $$
        \mathcal{B}^{(N)}(F_0)(\gamma+\lambda,x',s')-F^*(\gamma+\lambda,x',s')\le \frac{\ell}{4}
        \qquad \forall x'\in\{\zeta(x,a,s)|a\in\tilde{\mathcal{A}}(x,s)\},\,s'\in\mathcal{S}.
        $$}
        Therefore, 
    {\small $$
    \begin{aligned}
    &\mathcal{B}^{(N+1)}(F_0)(\gamma,x,s)\\
    \le& \sup_{a\in\tilde{A}(x,s)}\left[\left(r(a,x,s)+\sum_{i=1}^{I}\gamma^ig^i(a,x,s)+\lambda^i(g^i(a,x,s)-\bar{g}^i)\right)\right.\\
&\qquad\left.+\beta\mathcal{B}^{(N)}(F_0)(\gamma+\lambda,x',s')\right]\\<&\sup_{a\in\tilde{A}(x,s)}\left[\left(r(a,x,s)+\sum_{i=1}^{I}\gamma^ig^i(a,x,s)+\lambda^i(g^i(a,x,s)-\bar{g}^i)\right)\right.\\
&\qquad\left.+\beta F^*(\gamma+\lambda,x',s')+\frac{\ell}{4}\right]<F^*(\gamma,x,s)-\frac{\ell}{4},
    \end{aligned}
    $$}
    contradicting the fact that $\mathcal{B}^{(N+1)}(F_0)(\gamma,x,s)\ge F^*(\gamma,x,s)$. Hence $\mathcal{B}(F^*)\ge F^*$, and therefore $\mathcal{B}(F^*)=F^*$. 
    \item \textbf{Step 4. Show that $F^*$ is the largest fixed point in $\mathcal{N}$.} For any $F\in\mathcal{N}$, which is a fixed point of $\mathcal{B}$, we know that $F\le F_0$. Therefore, by the monotonicity of $\mathcal{B}$, we have
    $
    F=\mathcal{B}^{(n)}(F)\le \mathcal{B}^{(n)}(F_0),\,\forall n\in\mathbb{N}_+.
    $
    We take $n\rightarrow\infty$ and obtain $F\le F^*$.
    \end{itemize}
\end{proof}
To prove Theorem \ref{thm:contraction_infsup}, i.e., to show that the largest fixed point obtained from Theorem \ref{thm:contraction} equals the dual value function $D$, it suffices to verify that $D\in\mathcal{N}$, and that $D$ is not less than any fixed point in $\mathcal{N}$. 

The next lemma shows $D\in\mathcal{N}$. According to Theorem \ref{thm:lot_dual_equiv}, it suffices to verify that $V\in\mathcal{N}$, where $V$ is the maximum of the ex-ante lottery problem \eqref{equ:simplified_CK_lot}.
\begin{lemma}\label{lem:verification}
    The dual value function satisfies $D\in\mathcal{N}$. 
\end{lemma}
\begin{proof}
    It suffices to show that $V\in\mathcal{N}$. Let $P^*(\gamma,x_0,s_0)$ denote a maximizer for any $\gamma\in\mathbb{R}_+^I$, $x_0\in\mathcal{X}$, $s_0\in\mathcal{S}$.
    \begin{itemize}
        \item \textbf{Convexity}. For any $\gamma_0,\,\gamma_1\in\mathbb{R}_+^I$, and $\lambda\in[0,1]$, we denote $\gamma_{\lambda}=\lambda\gamma_1+(1-\lambda)\gamma_0$. Then $P^*(\gamma_{\lambda},x_0,s_0)$ is a feasible point for  the problem \eqref{equ:simplified_CK_lot} with $x_0,s_0$. Therefore,
        $$
        \begin{aligned}
        &\lambda V(\gamma_1,x_0,s_0)+(1-\lambda)V(\gamma_0,x_0,s_0)\\
        \ge &\lambda \mathbb{E}_{s_0}^{(a(s^t))\sim P^*(\gamma_{\lambda},x_0,s_0)}\sum_{t=0}^{\infty}\beta^t\left(r(x_t,a(s^t),s_t)+\sum_{i=1}^{I}\gamma^ig^i(x_t,a(s^t),s_t)\right)\\
        &+(1-\lambda) \mathbb{E}_{s_0}^{(a(s^t))\sim P^*(\gamma_{\lambda},x_0,s_0)}\sum_{t=0}^{\infty}\beta^t\left(r(x_t,a(s^t),s_t)+\sum_{i=1}^{I}\gamma^ig^i(x_t,a(s^t),s_t)\right)\\
        =&V(\gamma_{\lambda},x_0,s_0).
        \end{aligned}
        $$
        \item \textbf{$L$-Lipschitz}
        For any $\gamma_0,\,\gamma_1\in\mathbb{R}_+^I$, we have
        {\footnotesize $$
        \begin{aligned}
        V(\gamma_1,x_0,s_0)\ge& V(\gamma_0,x_0,s_0)+\mathbb{E}_{s_0}^{(a(s^t))\sim P^*(\gamma_0,x_0,s_0)}\sum_{t=0}^{\infty}\beta^t\sum_{i=1}^{I}(\gamma^i_1-\gamma^i_0)g^i(x_t,a(s^t),s_t)\\
        \ge& V(\gamma_0,x_0,s_0)- \frac{\|\gamma_1-\gamma_0\|_1\max_{i}\|g^i\|_{\infty}}{1-\beta}\\\ge &V(\gamma_0,x_0,s_0)-L\|\gamma_1-\gamma_0\|_1
        \end{aligned}
        $$}
        By symmetry, we have
        $
         V(\gamma_0,x_0,s_0)\ge V(\gamma_1,x_0,s_0)-L\|\gamma_1-\gamma_0\|_1.
        $
        \item \textbf{The bound \eqref{eq:F_inequality} and \eqref{eq:F_inequality_upper}.}
        \eqref{eq:F_inequality} follows from the definition of $V$, and \eqref{eq:F_inequality_upper} follows from the upper bounds on $r$ and $g^i$. 
    \end{itemize}
\end{proof}
The last lemma shows that $D$ is not less than any fixed point in $\mathcal{N}$.
\begin{lemma}\label{lem:verification_largestfp}
    Suppose that $F\in\mathcal{N}$ is a fixed point of $\mathcal{B}$. Then $F(x,\gamma,s)\le D(x,\gamma,s)$ for all $x\in\mathcal{X},\,\gamma\in\mathbb{R}_+^I,\,s\in\mathcal{S}$.
\end{lemma}

\begin{proof}
According to the definition of the dual value function $D(\gamma,x_0,s_0)$, for any $\epsilon>0$, there exists $(\lambda_t^i(h^t))\in\Lambda$, s.t.
\begin{equation}\label{eq:proof_largestfp_1}
\sup_{a(s^t)\in\tilde{\mathcal{A}}^{\infty}(x_0,s_0)}L((a_t),(\lambda_t^i);(\gamma^i),x_0,s_0)
\le D(\gamma,x_0,s_0)+\epsilon.
\end{equation}
We define $M=\sum_{t=0}^{\infty}\sum_{h^t\in\mathcal{H}^t}\sum_{i=1}^{I}\beta^t\lambda_t^i(h^t)\pi^t(s^t|s_0)$. 
\begin{itemize}
\item \textbf{Step 1. }We first show that, 
\begin{equation}\label{proof_largestfp_step1}
\begin{aligned}
 &\sup_{(a_t)\in\tilde{\mathcal{A}}^{\infty}(x_0,s_0)}L((a_t),(\lambda_t^i);(\gamma^i),x_0,s_0)\\=&\lim_{T\rightarrow\infty}\sup_{(a_t)\in\tilde{\mathcal{A}}^{\infty}(x_0,s_0)}
\mathbb{E}_{s_0}\sum_{t=0}^{T}\beta^t\left(r(x_t,a_t,s_t)+\sum_{i=1}^{I}\gamma^ig^i(x_t,a_t,s_t)\right.\\
&\left.+\lambda_t^i(s_0,a_0,\cdots,s_{t-1},a_{t-1},s_t)\left(\sum_{n=0}^{T-t}\beta^n g^i(x_{t+n},a_{t+n},s_{t+n})-\bar{g}^i\right)\right).
 \end{aligned}
\end{equation}
Since $r$ and $(g^i)$ are bounded, for any $T\ge 0$ we have
\begin{equation}\label{proof_largestfp_step1.5}
\begin{aligned}
&\sup_{(a_t)\in\tilde{\mathcal{A}}^{\infty}(x_0,s_0)}L((a_t),(\lambda_t^i);(\gamma^i),x_0,s_0)
\\=&\sup_{(a_t)\in\tilde{\mathcal{A}}^{\infty}(x_0,s_0)}\mathbb{E}_{s_0}\sum_{t=0}^{\infty}\beta^t\left(r(x_t,a_t,s_t)+\sum_{i=1}^{I}\gamma^ig^i(x_t,a_t,s_t)\right.\\
&\left.+\lambda_t^i(s_0,a_0,\cdots,s_{t-1},a_{t-1},s_t)\left(\sum_{n=0}^{\infty}\beta^n g^i(x_{t+n},a_{t+n},s_{t+n})-\bar{g}^i\right)\right).
\end{aligned}
\end{equation}
The operator \(\sum_{t=0}^{\infty}\left(\cdots\left(\sum_{n=0}^{\infty}\cdots\right)\right)\) on the right-hand side of \eqref{proof_largestfp_step1.5} can be split into 
$$
\begin{aligned}
&\sum_{t=0}^{\infty}\left(\cdots\left(\sum_{n=0}^{\infty}\cdots\right)\right)\\=&\underbrace
{\sum_{t=0}^{T}\left(\cdots\left(\sum_{n=0}^{T-t}\cdots\right)\right)}_{I}+\underbrace{\sum_{t=T+1}^{\infty}\left(\cdots\left(\sum_{n=0}^{\infty}\cdots\right)\right)}_{II}+\underbrace{\sum_{t=0}^{T}\left(\cdots\left(\sum_{n=T-t+1}^{\infty}\cdots\right)\right)}_{III}.
\end{aligned}
$$
To prove \eqref{proof_largestfp_step1}, it suffices to show that $II\rightarrow0$ and $III\rightarrow 0$ uniformly for any $(a_t)\in\tilde{\mathcal{A}}^{\infty}(x_0,s_0)$ as $T\rightarrow\infty$. For $II$, according to the boundedness of $r$ and $g^i$, 
\begin{equation}\label{eq:proof_largeestfp_II1}
\mathbb{E}_{s_0}\sum_{t=T+1}^{\infty}\beta^t\left(r(x_t,a_t,s_t)+\sum_{i=1}^{I}\gamma^ig^i(x_t,a_t,s_t)\right)\rightarrow0
\end{equation}
uniformly w.r.t. $(a_t)$. Furthermore, since
$
\sum_{t=T+1}^{\infty}\sum_{h^t\in\mathcal{H}^t}\sum_{i=1}^{I}\beta^t\pi(s^t|s_0)\lambda_t^i(h^t)\rightarrow0,
$
and since $|\sum_{n=0}^{\infty}\beta^ng^i(x_{t+n},a_{t+n},s_{t+n})-\bar{g}^i|$ is bounded above by $\frac{\|g^i\|_{\infty}}{1-\beta}+\|\bar{g}^i\|_{\infty}$, we have
{\footnotesize \begin{equation}\label{eq:proof_largeestfp_II2}
\begin{aligned}
&\mathbb{E}_{s_0}\sum_{t=T+1}^{\infty}\beta^t\sum_{i=1}^{I}\left(\lambda_t^i(s_0,a_0,\cdots,s_{t-1},a_{t-1},s_t)\left(\sum_{n=0}^{\infty}\beta^n g^i(x_{t+n},a_{t+n},s_{t+n})-\bar{g}^i\right)\right)\rightarrow0
\end{aligned}
\end{equation}}uniformly w.r.t. $(a_t)$. Combining \eqref{eq:proof_largeestfp_II1} and \eqref{eq:proof_largeestfp_II2} we know that $II\rightarrow0$ as $T\rightarrow\infty$ uniformly w.r.t. $(a_t)$.
\newline For $III$, the operator \(\sum_{t=0}^{T}\left(\cdots\left(\sum_{n=T-t+1}^{\infty}\cdots\right)\right)\) can be further split into 
$$
\small
\sum_{t=0}^{T}\left(\cdots\left(\sum_{n=T-t+1}^{\infty}\cdots\right)\right)=\underbrace{\sum_{t=0}^{[\frac{T}{2}]}\left(\cdots\left(\sum_{n=T-t+1}^{\infty}\cdots\right)\right)}_{III(1)}+\underbrace{\sum_{t=[\frac{T}{2}]+1}^{T}\left(\cdots\left(\sum_{n=T-t+1}^{\infty}\cdots\right)\right)}_{III(2)}.
$$
Since $\mathbb{E}_{s_0}\sum_{t=0}^{[\frac{T}{2}]}\allowbreak \sum_{i=1}^{I}\allowbreak \beta^t\lambda_t^i(s_0,a_0,\cdots, s_{t-1},a_{t-1},s_t)$ is bounded above by $M$, and $\bigl(\sum_{n=T-t+1}^{\infty}\allowbreak \beta^ng^i(x_{t+n},a_{t+n},s_{t+n})\bigr)$ is bounded above by $\frac{\beta^{T-[\frac{T}{2}]+1}}{1-\beta}\|g^i\|_{\infty}$, we know that $III(1)\rightarrow0$ uniformly w.r.t. $(a_t)$ when $T\rightarrow\infty$. \newline
Since 
{\footnotesize {\footnotesize $$
\begin{aligned}
&\mathbb{E}_{s_0}\sum_{t=[\frac{T}{2}]+1}^{T}\sum_{i=1}^{I}\beta^t\lambda_t^i(s_0,a_0,\cdots, s_{t-1},a_{t-1},s_t)\\
\le& \mathbb{E}_{s_0}\sum_{t=[\frac{T}{2}]+1}^{\infty}\sum_{i=1}^{I}\beta^t\lambda_t^i(s_0,a_0,\cdots, s_{t-1},a_{t-1},s_t)\rightarrow 0
\end{aligned}
$$}}uniformly w.r.t. $(a_t)$, and $|\sum_{n=T-t+1}^{\infty}\beta^ng^i(x_{t+n},a_{t+n},s_{t+n})|$ are bounded above by $\frac{\|g_i\|_{\infty}}{1-\beta}$, we know that $III(2)\rightarrow0$ uniformly w.r.t. $(a_t)$ when $T\rightarrow\infty$. Hence $III\rightarrow0$  uniformly w.r.t. $(a_t)$ as $T\rightarrow\infty$, and therefore, \eqref{proof_largestfp_step1} holds.
\item \textbf{Step 2. }We then show that, given $x_0\in\mathcal{X},\,s_0\in\mathcal{S}$, we have
{\footnotesize \begin{equation}\label{proof_largestfp_step2}
\begin{aligned}
&\limsup_{T\rightarrow\infty}\sup_{(a_t(s^t))\in\tilde{\mathcal{A}}^{\infty}(x_0,s_0)}\mathbb{E}_{s_0}\beta^{T+1}F(\gamma+\sum_{t=0}^{T}\lambda_t(s_0,a_0,\cdots,s_{t-1},a_{t-1},s_t),x_{T+1},s_{T+1})\le 0.
\end{aligned}
\end{equation}}
According to \eqref{eq:F_inequality_upper}, we know that for any $(a_t(s^t))\in\tilde{\mathcal{A}}^{\infty}(x_0,s_0)$, we have
$$
\begin{aligned}
&\mathbb{E}_{s_0}\beta^{T+1}F(\gamma+\sum_{t=0}^{T}\lambda_t(s_0,a_0,\cdots,s_{t-1},a_{t-1},s_t),x_t,s_t)\\
\le &\mathbb{E}_{s_0}\beta^{T+1}\left[1+\sum_{i=1}^{I}\left(\gamma^i+\sum_{t=0}^{T}\lambda_t^i(s_0,a_0,\cdots,s_{t-1},a_{t-1},s_t)\right)\right]L\\
=&\beta^{T+1}\left(1+\sum_{i=1}^I\gamma^i\right)L+\beta^{T+1}L\mathbb{E}_{s_0}\sum_{t=0}^{T}\sum_{i=1}^{I}\lambda_t^i(s_0,a_0,\cdots,s_{t-1},a_{t-1},s_t)\\
\le&\underbrace{\beta^{T+1}\left(1+\sum_{i=1}^I\gamma^i\right)L}_{IV}+\underbrace{\beta^{T+1}L\sum_{t=0}^{T}\sum_{h^t\in\mathcal{H}^t}\sum_{i=1}^{I}\pi(s^t|s_0)\lambda_t^i(h^t)}_{V}
\end{aligned}
$$
Obviously $IV\rightarrow0$ as $T\rightarrow\infty$, and 
{\footnotesize $$
\begin{aligned}
V=&L\sum_{t=0}^{T}\beta^{T+1-t}\sum_{h^t\in\mathcal{H}^t}\sum_{i=1}^{I}\left(\beta^t\pi(s^t|s_0)\lambda_t^i(h^t)\right)\\
\le &L\sum_{t=0}^{[\frac{T}{2}]}\beta^{T+1-[\frac{T}{2}]} \sum_{h^t\in\mathcal{H}^t}\sum_{i=1}^{I}\left(\beta^t\pi(s^t|s_0)\lambda_t^i(h^t)\right)+L\sum_{t=[\frac{T}{2}]+1}^{\infty}\sum_{h^t\in\mathcal{H}^t}\sum_{i=1}^{I}\beta^t\pi(s^t|s_0)\lambda_t^i(h^t)\\
\le &L\beta^{T+1-[\frac{T}{2}]}M+L\sum_{t=[\frac{T}{2}]+1}^{\infty}\sum_{h^t\in\mathcal{H}^t}\sum_{i=1}^{I}\beta^t\pi(s^t|s_0)\lambda_t^i(h^t)\rightarrow0 (T\rightarrow\infty).
\end{aligned}
$$} Since $IV$ and $V$ do not depend on $(a_t(s^t))$, we obtain \eqref{proof_largestfp_step2}.
\item \textbf{Step 3. } According to \eqref{eq:proof_largestfp_1} and \eqref{proof_largestfp_step1}, there exists $T_1>0$, s.t. when $T\ge T_1$, 
{\footnotesize \begin{equation}\label{eq:proof_largestfp_step3mid1}
\begin{aligned}
D(\gamma,x_0,s_0)+2\epsilon\ge& \sup_{(a_t(s^t))\in\tilde{\mathcal{A}}^{\infty}(x_0,s_0)}\mathbb{E}_{s_0}\sum_{t=0}^{T}\beta^t\left(r(x_t,a_t,s_t)+\sum_{i=1}^{I}\gamma^ig^i(x_t,a_t,s_t)\right.\\&\left.+\lambda_t^i(s_0,a_0,\cdots,s_{t-1},a_{t-1},s_t)\left(\sum_{n=0}^{T-t}\beta^n g^i(x_{t+n},a_{t+n},s_{t+n})-\bar{g}^i\right)\right).
\end{aligned}
\end{equation}}Combining the inequality \eqref{eq:proof_largestfp_step3mid1} and \eqref{proof_largestfp_step2}, we know that there exists $T_2>T$, s.t. when $T>\max\{T_1,T_2\}$,
{\small \begin{equation}\label{eq:proof_largestfp_step3mid2}
\begin{aligned}
&D(\gamma,x_0,s_0)+3\epsilon\\\ge& \sup_{(a_t)\in\tilde{\mathcal{A}}^{\infty}(x_0,s_0)}\mathbb{E}_{s_0}\sum_{t=0}^{T}\beta^t\left(r(x_t,a_t,s_t)+\sum_{i=1}^{I}\gamma^ig^i(x_t,a_t,s_t)\right.\\
&\left.+\lambda_t^i(s_0,a_0(s_0),\cdots,s_{t-1},a_{t-1}(s^{t-1}),s_t)\left(\sum_{n=0}^{T-t}\beta^n g^i(x_{t+n},a_{t+n},s_{t+n})-\bar{g}^i\right)\right)\\
&+\sup_{(a_t)\in\tilde{\mathcal{A}}^{\infty}(x_0,s_0)}\mathbb{E}_{s_0}\beta^{T+1}F(\gamma+\sum_{t=0}^{T}\lambda_t(s_0,a_0,\cdots,s_{t-1},a_{t-1},s_t),x_{T+1},s_{T+1})\\
\ge &\sup_{(a_t)\in\tilde{\mathcal{A}}^{\infty}(x_0,s_0)}\left[\mathbb{E}_{s_0}\sum_{t=0}^{T}\beta^t\left(r(x_t,a_t,s_t)+\sum_{i=1}^{I}\gamma^ig^i(x_t,a_t,s_t)\right.\right.\\
&\left.+\lambda_t^i(s_0,a_0,\cdots,s_{t-1},a_{t-1},s_t)\left(\sum_{n=0}^{T-t}\beta^n g^i(x_{t+n},a_{t+n},s_{t+n})-\bar{g}^i\right)\right)\\
&\left.+\mathbb{E}_{s_0}\beta^{T+1}F(\gamma+\sum_{t=0}^{T}\lambda_t(s_0,a_0,\cdots,s_{t-1},a_{t-1},s_t),x_{T+1},s_{T+1})\right].
\end{aligned}
\end{equation}}
We rearrange the right-hand side of \eqref{eq:proof_largestfp_step3mid2} and have
{\small \begin{equation}\label{eq:proof_largestfp_step3mid3}
\begin{aligned}
    &\text{RHS of \eqref{eq:proof_largestfp_step3mid2}}\\
=&\sup_{(a_t(s^t))\in\tilde{\mathcal{A}}^{\infty}(x_0,s_0)}\left[\mathbb{E}_{s_0}\sum_{t=0}^{T}\beta^t\left(r(x_t,a_t,s_t)\right.\right.\\&\left.\left.+\sum_{i=1}^{I}\left(\gamma^i+\sum_{n=0}^{t-1}\lambda_n^i(s_0,a_0,\cdots, s_{n-1},a_{n-1},s_n)\right)g^i(x_t,a_t(s^t),s_t)\right.\right.\\
&\left.+\lambda_t^i(s_0,a_0,\cdots,s_{t-1},a_{t-1},s_t)\left(g^i(x_{t},a_{t},s_{t})-\bar{g}^i\right)\right)\\&\left.+\mathbb{E}_{s_0}\beta^{T+1}F(\gamma+\sum_{t=0}^{T}\lambda_t(s_0,a_0,\cdots,s_{t-1},a_{t-1},s_t),x_{T+1},s_{T+1})\right]
\end{aligned}
\end{equation}}
Since $F$ is a fixed point of $\mathcal{B}$, we have
$$
\text{RHS of \eqref{eq:proof_largestfp_step3mid3}}
\ge\sup_{(a_t)\in\tilde{\mathcal{A}}^{\infty}(x_0,s_0)}\mathbb{E}_{s_0}\beta^TF(\gamma+\sum_{t=0}^{T-1}\lambda_t(s_0,a_0,\cdots,s_{t-1},a_{t-1},s_t),x_T,s_T).
$$
Hence by induction, the right-hand side of \eqref{eq:proof_largestfp_step3mid3} is bounded below by
$
F(\gamma,x_0,s_0)
$, and \eqref{eq:proof_largestfp_step3mid2} implies that $D(\gamma,x_0,s_0)+3\epsilon\ge F(\gamma,x_0,s_0)$. Since $\epsilon$ is arbitrarily chosen, we conclude that $D(\gamma,x_0,s_0)\ge F(\gamma,x_0,s_0)$. 
\end{itemize}
\end{proof}
The same argument can be used for the sup-inf problem to establish Theorem \ref{thm:contraction_supinf}.
\begin{theorem}\label{cor:contraction_supinf}
The following functional equation, 
 {\footnotesize $$ \begin{aligned} &F(\gamma,x,s)=\sup_{a\in\tilde{\mathcal{A}}(x,s)}\inf_{\lambda\in \mathbb{R}_+^I}\Big[\left(r(a,x,s)+\sum_{i=1}^{I}\left(\gamma^ig^i(a,x,s)+\lambda^i(g^i(a,x,s)-\bar{g}^i)\right)\right)\\
&\qquad +\beta\mathbb{E}_{s}F(\gamma+\lambda,x',s')\Big] \end{aligned} $$}
 where $x'=\zeta(x,a,s)$, has a solution in $\mathcal{N}$.
Moreover, $\tilde{D}(\gamma,x,s)$ is the largest solution in $\mathcal{N}$.
\end{theorem}

\begin{proof}
We define
{\small $$
\begin{aligned}
&\tilde{\mathcal{B}}(F)=\sup_{a\in{\tilde{\mathcal{A}}(x,s)}}\inf_{\lambda\in\mathbb{R}_+^I}\left[\left(r(a,x,s)+\sum_{i=1}^{I}\gamma^ig^i(a,x,s)+\lambda^i(g^i(a,x,s)-\bar{g}^i)\right)\right.\\
&\qquad\left.+\beta\mathbb{E}_sF(\gamma+\lambda,x',s')\right].
\end{aligned}
$$}
According to Lemma \ref{lem:tech_infsup}, we have
{\small $$
\begin{aligned}
\tilde{\mathcal{B}}(F)&=\sup_{a\in{\tilde{\mathcal{A}}(x,s)}}\inf_{\lambda\in\mathbb{R}_+^I}\left[\left(r(a,x,s)+\sum_{i=1}^{I}\gamma^ig^i(a,x,s)+\lambda^i(g^i(a,x,s)-\bar{g}^i)\right)\right.\\
&\qquad\left.+\beta\mathbb{E}_sF(\gamma+\lambda,x',s')\right]\\
&=\inf_{(\lambda_a)_{a\in\mathcal{A}}\in\mathbb{R}_+^{I}}\sup_{a\in{\tilde{\mathcal{A}}(x,s)}}\left[\left(r(a,x,s)+\sum_{i=1}^{I}\gamma^ig^i(a,x,s)+\lambda_a^i(g^i(a,x,s)-\bar{g}^i)\right)\right.\\
&\qquad\left.+\beta\mathbb{E}_sF(\gamma+\lambda_a,x',s')\right].
\end{aligned}
$$}
The proof that $\tilde{D}\in\mathcal{N}$ is identical to the proof of Lemma \ref{lem:verification}. To prove the theorem, it suffices to show that the operator $\tilde{\mathcal{B}}$ satisfies:
\begin{itemize}
    \item $\tilde{\mathcal{B}}(\mathcal{N})\subset \mathcal{N}$;
    \item $\tilde{\mathcal{B}}$ satisfies monotonicity,
\end{itemize}
so that the proofs of Theorem \ref{thm:contraction} and Lemma \ref{lem:verification_largestfp} carry over.
It is straightforward to verify the monotonicity. It suffices to check that $\tilde{\mathcal{B}}(\mathcal{N})\subset \mathcal{N}$.
\begin{enumerate}
    \item \textbf{Show that $\tilde{\mathcal{B}}$ preserves convexity.} 
    If $F$ is convex in $\gamma$, then, for given $a\in\mathcal{A}$, $\bigl[\bigl(r(a,x,s)+\sum_{i=1}^{I}\gamma^ig^i(a,x,s)+\lambda_a^i(g^i(a,x,s)-\bar{g}^i)\bigr)+\beta\mathbb{E}_sF(\gamma+\lambda_a,x',s')\bigr]$ is convex in $(\gamma,(\lambda_a)_{a\in\mathcal{A}})$, and the proof of Lemma 2 in \cite{pavoni2018dual} applies.
     \item \textbf{Show that $\tilde{\mathcal{B}}$ preserves \eqref{eq:F_inequality}, \eqref{eq:F_inequality_upper}, and the $L$-Lipschitz continuity.} The proof is analogous to that of Lemma \ref{lem:Bellman_restriction}, with $\lambda$ replaced by $\lambda_a$.
     \item \textbf{Show that $\tilde{\mathcal{B}}$ preserves the finiteness of the norm $\|\cdot\|_{\mathcal{M}}$.}  This is a direct corollary of the fact that $\tilde{\mathcal{B}}(F)$ preserves the properties \eqref{eq:F_inequality} and \eqref{eq:F_inequality_upper}. 
\end{enumerate}
\end{proof}
\section{\texorpdfstring{A numerical example (Online Only)}{A numerical example (Online Only)}}
\label{App:Num_ex}
Consider an infinite-horizon principal--agent problem in which last period's compensation determines the agent's outside option in the current period. That is, we take last period's compensation to be the physical state in the current period and assume that the outside option depends on this state. The problem can be written as follows.
\begin{equation}
\begin{aligned}
    &\max_{\{0\le c_t\le 1\}}\sum_{t=0}^{\infty}-\beta^tc_t,\\
    \text{s.t. }&\mathbb{E}_t\sum_{n=0}^{\infty}\beta^nu(c_{t+n})\ge v(x_t),
    \text{ where }&x_t=c_{t-1} \mbox{ for } t\ge 1,\,x_0=0.\\
\end{aligned}
\end{equation}
For simplicity, we consider only two possible values for the outside option and assume that there is some cut-off compensation level $ \underline{c} $ such that
$$
v(x_t)=\begin{cases}
    \frac{\beta+(1-\beta)u(\tilde{c})}{1-\beta},&x_t\le \underline{c},\\
        \frac{1}{1-\beta},&x_t> \underline{c},
\end{cases}
$$
with
$
0<\underline{c}<\tilde{c}<1.
$
We set $u(c_t)=\sqrt{c_t}$ and $ 0<\beta<1 $. It is easy to see that the optimal deterministic solution is given by
$
c_0=\tilde{c},\,c_t=1 \text{ for } t\ge 1.
$

Now we consider a feasible lottery allocation. We assume that the agent has to agree to the contract before the lottery is realized (i.e., we consider an ex-ante lottery). At each $t$, the principal chooses $c_t$ to be equal to  $\underline{c}$ (with probability $p$) and $\overline{c}\in(\tilde{c},1]$ (with probability $1-p$) whenever $x_t\le \underline{c}$, and $c_t=1$ whenever $x_t>\underline{c}$. For the participation constraint to be satisfied, we need
$$
\frac{pu(\underline{c})+(1-p)u(\overline{c})}{1-\beta p}+(\frac{1}{1-\beta}-\frac{1}{1-\beta p}) \ge \frac{\beta}{1-\beta}+u(\tilde{c}),
$$
implying that
\begin{equation}\label{cond1}
\frac{1}{1-\beta p}=\frac{1-u(\tilde{c})}{1-pu(\underline{c})-(1-p)u(\overline{c})}
\end{equation}
when the participation constraint is binding. The principal's utility can be written as
$$
\begin{aligned}
&-\frac{p\underline{c}+(1-p)\overline{c}}{1-\beta p}-\left(\frac{1}{1-\beta}-\frac{1}{1-\beta p}\right)\\
=&-\frac{pu(\underline{c})+(1-p)u(\overline{c})}{1-\beta p}-\left(\frac{1}{1-\beta}-\frac{1}{1-\beta p}\right)-\frac{p(\underline{c}-u(\underline{c}))+(1-p)(\overline{c}-u(\overline{c}))}{1-\beta p}\\
=&-\left[\tilde{c}+\frac{\beta}{1-\beta}\right]-(1-u(\tilde{c}))\underbrace{\left[\frac{p(\underline{c}-u(\underline{c}))+(1-p)(\overline{c}-u(\overline{c}))}{1-pu(\underline{c})-(1-p)u(\overline{c})}-\frac{\tilde{c}-u(\tilde{c})}{1-u(\tilde{c})}\right].}_{I}
\end{aligned}
$$
In the deterministic solution, the principal's payoff is 
$ -\left[\tilde{c}+\frac{\beta}{1-\beta}\right] $ and hence there is an improving lottery if and only if $ I<0 $ and (\ref{cond1}) holds.
The following lemma provides a simple, general, necessary and sufficient condition for this to be true.
\begin{lemma}
\label{lem:lottinPA}
Within the family of lotteries constructed above, there exists an improving ex-ante lottery if and only if
\begin{equation}
\label{cond_star}
(\sqrt{\tilde{c}}-\sqrt{\underline{c}})^2<\beta(1-\sqrt{\tilde{c}})^2.
\end{equation}
\end{lemma}
\begin{proof}
   For any $\bar{c}\in(\tilde{c},1)$ and $p\in (0,1)$, we can construct a random variable $Y$, s.t. $Y=\sqrt{\underline{c}}$ with probability $p$, and $Y=\sqrt{\bar{c}}$ with probability $1-p$. Then condition \eqref{cond1} is equivalent to
\begin{equation}\label{cond1_genequiv}
\mathbb{E}(Y)=p\sqrt{\underline{c}}+(1-p)\sqrt{\bar{c}}= 1-(1-\beta p)(1-\sqrt{\tilde{c}}),
\end{equation}
and $I<0$ is equivalent to
\begin{equation}\label{cond2_genequiv}
    \frac{\mathbb{E}(Y^2)-\mathbb{E}(Y)}{1-\mathbb{E}(Y)}+\sqrt{\tilde{c}}=-\mathbb{E}(Y)+\frac{\mathrm{Var}(Y)}{1-\mathbb{E}(Y)}+\sqrt{\tilde{c}}<0 \Leftrightarrow \mathrm{Var}(Y)< \beta p(1-\beta p)(1-\sqrt{\tilde{c}})^2. 
\end{equation}
From \eqref{cond1_genequiv}, we obtain
$
\mathrm{Var}(Y)=\frac{p(\mathbb{E}(Y)-\sqrt{\underline{c}})^2}{1-p}=\frac{p(\sqrt{\tilde{c}}+\beta p(1-\sqrt{\tilde{c}})-\sqrt{\underline{c}})^2}{1-p}.
$
Therefore, conditions \eqref{cond1_genequiv} and \eqref{cond2_genequiv} are equivalent to
\begin{equation}\label{cond_final}
\begin{aligned}
&\mathrm{Var}(Y)=\frac{p(\sqrt{\tilde{c}}+\beta p(1-\sqrt{\tilde{c}})-\sqrt{\underline{c}})^2}{1-p}<\beta p(1-\beta p)(1-\sqrt{\tilde{c}})^2\\
\Leftrightarrow&\underbrace{\frac{(\sqrt{\tilde{c}}+\beta p(1-\sqrt{\tilde{c}})-\sqrt{\underline{c}})^2}{1-p}}_{II}<\underbrace{\beta (1-\beta p)(1-\sqrt{\tilde{c}})^2}_{III}.
\end{aligned}
\end{equation}
Since $II$ is increasing w.r.t. $p$ and $III$ is decreasing w.r.t. $p$, ex-ante lotteries are improving if and only if
$$
\lim_{p\rightarrow 0+}(II-III)<0\Leftrightarrow (\sqrt{\tilde{c}}-\sqrt{\underline{c}})^2<\beta (1-\sqrt{\tilde{c}})^2.
$$
\end{proof}

While ex-ante lotteries often improve welfare in this example, ex-post lotteries never do. If the agent can decide whether to participate after the realization of the lottery, the agent participates at $ t=0 $ only if $ c_0 \ge \tilde{c} $. This rules out any compensation below 1 in future periods.

We set $\underline{c}=0.09$, $\tilde{c}=0.1$, and $\beta=0.7$, which ensures that ex-ante lotteries improve welfare. The inf-sup functional equation for this problem is
$$
V(x,\gamma)=\inf_{\lambda\ge 0}\sup_{c\in[0,1]}-c+\gamma\sqrt{c}+\lambda(\sqrt{c}-v(x))+\beta V(c,\lambda+\gamma).
$$
Note that \(V(x,\gamma)\) depends on \(x\) only through whether \(x\) is above or below the threshold \(\underline{c}\). Hence, there exist two functions \(L(\gamma)\) and \(R(\gamma)\) such that
\[
V(x,\gamma)=
\begin{cases}
R(\gamma), & x>\underline{c},\\
L(\gamma), & x\le \underline{c}.
\end{cases}
\]
The compensation set $[0,1]$ is discretized as the grid $\{0,0.01,\ldots,1\}$. Figure~\ref{fig:dual} shows the shape of the dual function $V(x,\gamma)$ via $L(\gamma)$ and $R(\gamma)$.
    \begin{figure}[th!]
    \centering
    \begin{subfigure}{0.45\textwidth}
        \includegraphics[width=0.95\textwidth]{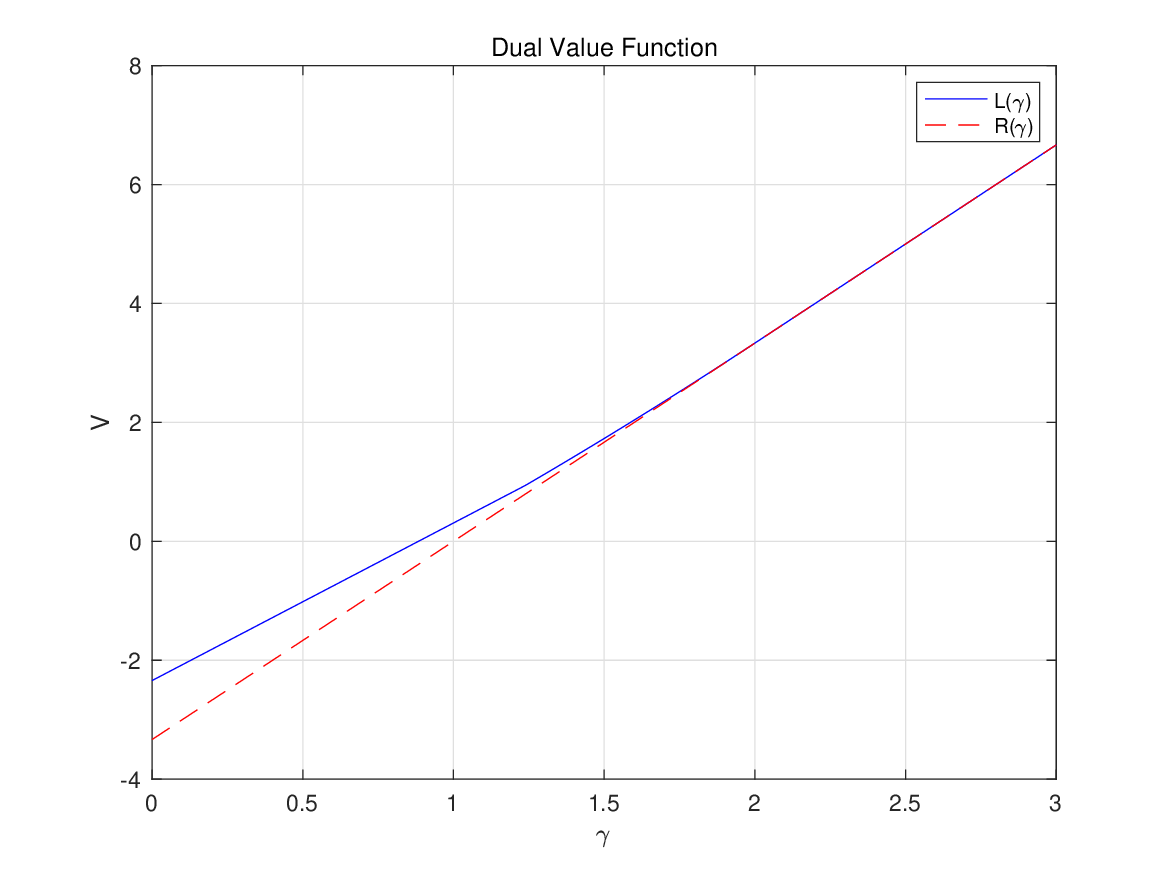}
    \caption{Dual function $V(x,\gamma)$.}
    \label{fig:dual}
    \end{subfigure}
    \hfill
    \begin{subfigure}{0.45\textwidth}
        \includegraphics[width=0.95\textwidth]{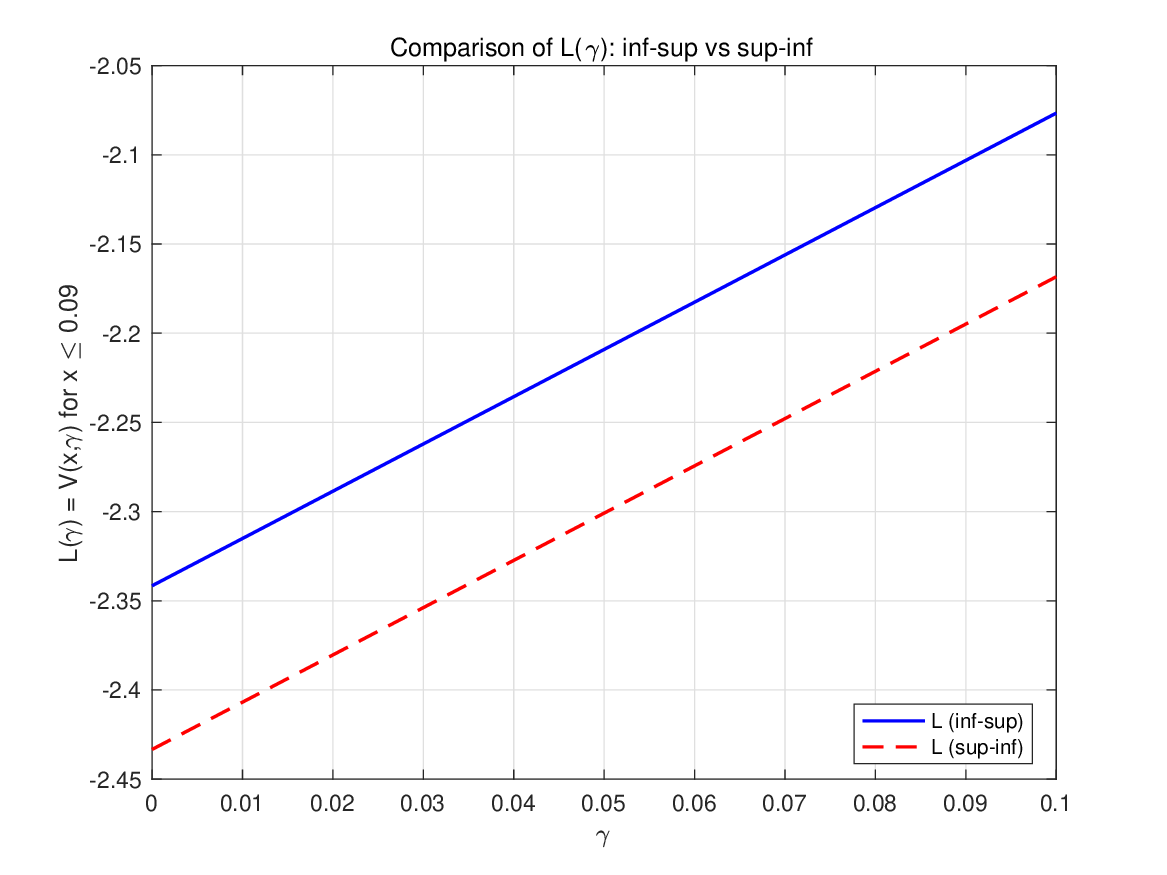}
     \caption{ $V(x,\gamma)$ and $\tilde{V}(x,\gamma)$ when $x\le \underline{c}$.}
    \label{fig:dual_comp}
    \end{subfigure}

 \end{figure}
The function \(R(\gamma)\) is \(R(\gamma)=(\gamma-1)/(1-\beta)\), corresponding to the fact that for \(x>\underline{c}\), the only feasible compensation is \(c_t=1\) for all \(t\ge 0\). For \(x\le \underline{c}\) and small \(\gamma\), we have \(L(\gamma)>R(\gamma)\), meaning that welfare can be improved by occasionally choosing \(c_t<1\). As \(\gamma\) grows, the optimal compensation plan approaches \(c_t\equiv 1\), and hence \(L(\gamma)\) converges downward to \(R(\gamma)\).

To better understand the welfare improvement obtained from ex-ante lotteries, we also consider the sup-inf functional equation for this problem:
$$
\tilde{V}(x,\gamma)=\sup_{c\in[0,1]}\inf_{\lambda\ge 0} -c+\gamma\sqrt{c}+\lambda(\sqrt{c}-v(x))+\beta \tilde{V}(c,\lambda+\gamma).
$$
Figure~\ref{fig:dual_comp} compares \(V(x,\gamma)\) and \(\tilde{V}(x,\gamma)\) for \(x\le \underline{c}\) in the neighborhood of \(\gamma=0\). 
For \(x\le \underline{c}\), we have \(\tilde{V}(x,0)=-2.43\), which equals the value of the optimal deterministic solution \(-(0.1+\beta/(1-\beta))\). This indicates that ex-post lotteries do not improve welfare in this case. In contrast, there is a clear gap between \(V(x,0)\) and \(\tilde{V}(x,0)\) for \(x\le \underline{c}\), suggesting that ex-ante lotteries significantly improve welfare when \(x\le \underline{c}\).

Finally, Algorithm~\ref{algorithm} is used to compute the approximate optimal policy with ex-ante lotteries. The initial Lagrange multipliers are set to zero, and the learning rates are chosen as
\[
\sigma^k = \frac{1}{k+5}.
\]
The total number of iterations is \(N=10000\), and the lottery is constructed from the last $N/10=1000$ iterations. Table~\ref{tab:multi_period_full} reports the optimal compensation policy with ex-ante lotteries. In each period, when \(x\le \underline{c}\), the compensation lottery is supported on two points: \(0.09\,(=\underline{c})\) and \(0.38\,(>\underline{c})\). Whenever the higher compensation \(0.38\) is realized, all subsequent compensation levels are set to \(1\). The algorithm does not capture a stationary lottery; instead, it finds a solution in which the participation constraint for \(x\le \underline{c}\) binds only at \(t=0\).

\begin{table}[htbp]
\centering
\caption{Multi-period results with ex-ante lotteries}
\label{tab:multi_period_full}
\begin{tabular}{c|c|c|c|c|c}
\hline
Period & $x\le \underline{c}$ (frequencies) & $\phi$ & $\lambda$ & $\mu$ & $x>\underline{c}$ \\
\hline
$0$ & $0.09\ (49.2\%),\ 0.38\ (50.8\%)$ & 2.9129 & 1.2383 & -- & -- \\
$1$ & $0.09\ (6.1\%),\ 0.38\ (93.9\%)$ & 2.9132 & 0.0000 & 1.2383 & 1 \\
$2$ & $0.09\ (6.0\%),\ 0.38\ (94.0\%)$ & 2.9132 & 0.0000 & 1.2383 & 1 \\
$\vdots$ & $\vdots$ & $\vdots$ & $\vdots$ & $\vdots$ & $\vdots$ \\
$9$ & $0.09\ (6.0\%),\ 0.38\ (94.0\%)$ & 2.9132 & 0.0000 & 1.2383 & 1 \\
\hline
\end{tabular}
\begin{tablenotes}
\footnotesize
\item Notes: Frequencies are normalized and rounded to one decimal place. $\phi$ is the weighted average $\partial_2 V$ for the smallest compensation ($0.09$). $\lambda$ and $\mu$ are final multipliers from the subgradient descent. $x>\underline{c}$ always yields compensation $1$ (except period 0 where no such branch appears).
\end{tablenotes}
\end{table}

\end{appendices}
\end{document}